\documentclass[a4paper,12pt]{report}
\usepackage{fullpage}
\usepackage{graphicx} 
\usepackage{listings}
\usepackage{pgfplots}
\usepackage{pgfgantt}
\usepackage{subcaption}
\usepackage{array}
\usepackage{amsthm}
\usepackage{amssymb}
\usepackage{makecell}
\usepackage{longtable}

\usepackage{soul}
\usepackage{booktabs}

\usepackage{amsmath}
\usepackage{xcolor}
\usepackage{comment}
\usepackage{mathtools}
\usepackage{float}

\newtheorem{lemma}{Lemma}
\newtheorem{theorem}[lemma]{Theorem}

\newtheorem{corollary}[lemma]{Corollary}

\newtheorem{observation}[lemma]{Observation}
\DeclarePairedDelimiter\ceil{\lceil}{\rceil}
\DeclarePairedDelimiter\floor{\lfloor}{\rfloor}
\title{An Implementation and Experimental Comparison of Dynamic Ordered Sets}
\author{Jordan Malek \\ jmalek@cs.toronto.edu}

\definecolor{codegreen}{rgb}{0,0.6,0}
\definecolor{codegray}{rgb}{0.5,0.5,0.5}
\definecolor{codepurple}{rgb}{0.58,0,0.82}
\definecolor{backcolour}{rgb}{0.95,0.95,0.92}
\lstdefinestyle{mystyle}{
	backgroundcolor=\color{backcolour},
	commentstyle=\color{codegreen},
	keywordstyle=\color{magenta},
	numberstyle=\tiny\color{codegray},
	stringstyle=\color{codepurple},
	basicstyle=\ttfamily\scriptsize,
	breakatwhitespace=false,
	breaklines=true,
	captionpos=b,
	keepspaces=true,
	numbers=left,
	escapeinside={@}{@}, 
	numbersep=5pt,
	showspaces=false,
	showstringspaces=false,
	showtabs=false,
	morekeywords={CAS, return},
	tabsize=4,
	emph={addr, break, while},
	emphstyle={\color{blue}}
}
\usepackage{tikz}
\usetikzlibrary{calc,shapes.multipart,chains,arrows,positioning}
\usetikzlibrary{chains,decorations.pathreplacing}
\usetikzlibrary{positioning}
\usetikzlibrary{graphs}
\usetikzlibrary{shapes.geometric}
\tikzset{subtree/.style = {isosceles triangle, isosceles triangle apex angle=60, text width=1em, draw=black, shape border rotate=90, align=center, anchor=north}}
\tikzset{vertex/.style = {circle,draw=black, scale=0.7}}
\tikzset{edge/.style = {->, thick}}
\tikzset{uEdge/.style = {-,thick}}
\tikzset{dottedEdge/.style = {->,thin,dotted}}
\tikzset{dottedVertex/.style = {circle,draw=black,dotted}}
\tikzset{tail/.style = {rectangle, very thick, inner sep=5pt, color=black, draw, text=black, fill=blue!20}}
\tikzset{head/.style = {very thick, rectangle split,rectangle split parts=3, draw, rectangle split, minimum size=15pt, inner sep=5pt, text=black, rectangle split part fill={blue!20, white!20, blue!20}}}
\tikzset{updateNode/.style = {very thick, rectangle split,rectangle split parts=4, draw, rectangle split, minimum size=15pt, inner sep=5pt, text=black, rectangle split part fill={blue!20, white!20, blue!20}}}
\tikzset{markedNode/.style = {very thick, rectangle split,rectangle split parts=4, draw, rectangle split, minimum size=15pt, inner sep=5pt, text=black, rectangle split part fill={blue!20, red!20, blue!20}}}
\tikzset{delFlagNode/.style = {very thick, rectangle split,rectangle split parts=4, draw, rectangle split, minimum size=15pt, inner sep=5pt, text=black, rectangle split part fill={blue!20, orange!20, blue!20}}}
\tikzset{insFlagNode/.style = {very thick, rectangle split,rectangle split parts=4, draw, rectangle split, minimum size=15pt, inner sep=5pt, text=black, rectangle split part fill={blue!20, violet!20, blue!20}}}
\tikzset{descNode/.style = {very thick, rectangle split,rectangle split parts=3, draw, rectangle split, minimum size=18pt, inner sep=5pt, text=black, rectangle split part fill={green!20, green!20}}}
\newcommand{\twoField}[2]{{\langle}#1,#2{\rangle}}
\newcommand{\threeField}[3]{{\langle}#1,#2,#3{\rangle}}

\usepackage[natbib=true]{biblatex}
\usepackage{url}
\usepackage{hyperref}
\usepackage[hyphenbreaks]{breakurl}
\usepackage{theoremref}
\usepackage[noend]{algpseudocode}
\usepackage{algorithm}
\newcommand{\algorithmicbreak}{\textbf{break}}
\newcommand{\Break}{\State \algorithmicbreak}

\algdef{SE}[SUBALG]{Indent}{EndIndent}{}{\algorithmicend\ }%
\algdef{SE}[DOWHILE]{Do}{doWhile}{\algorithmicdo}[1]{\algorithmicwhile\ #1}
\newcommand{\alglinenoNew}[1]{\newcounter{ALG@line@#1}}
\newcommand{\alglinenoPop}[1]{\setcounter{ALG@line}{\value{ALG@line@#1}}}
\newcommand{\alglinenoPush}[1]{\setcounter{ALG@line@#1}{\value{ALG@line}}}
\algtext*{Indent}
\algtext*{EndIndent}
\algrenewcommand\algorithmicindent{1.0em}
\begin{document}

\maketitle

\section*{Acknowledgements}
I would like to thank my supervisor, Faith Ellen. 
She has been enormously generous with her time and effort
towards this report. Her guidance and insights have 
left an immeasurable impact on me, both as a computer scientist and as a writer.
It has been a pleasure working and learning from not only her, but also her graduate students, Jeremy Ko and Jason Liu.
They have provided feedback and encouragement throughout this project.
I have adored the experience of working with Jeremy and exchanging ideas 
with him as he refined his theoretical implementation
and I refined my practical one of his work. 

I am grateful to Eric Ruppert, as without his encouragement and mentorship during my time at York University, 
I would not have pursued graduate studies.
The research project that I worked on with him 
during the Summer of 2021 was my first 
experience working with shared data structures.

I would also like to thank Trevor Brown for the helpful advice
he provided as I worked on this project.
He graciously allowed me to run experiments for this project on the Multicore Lab machines at the University of Waterloo.
This was vital for comparing my implementations and 
verifying their correctness.
\tableofcontents
\chapter{Introduction}

A \textit{dynamic ordered set} stores a dynamic subset, $S$, of keys 
from an ordered universe, \textit{U}.
It supports the following operations, where $x \in U$:
\begin{itemize}
    \item insert($x$): If $x \notin S$ then $x$ is inserted into $S$, otherwise this operation has no effect. 
    \item remove($x$): If $x \in S$ then $x$ is removed from $S$, otherwise this operation has no effect.
    \item predecessor($x$):  Returns the largest key in $S$ that is less than $x$, or a key that is not in $U$
    if there is no key smaller than $x$ in $S$.
    \item search($x$): Returns $True$ if $x \in S$, $False$ otherwise.
\end{itemize}

A binary trie is a sequential data structure which implements this abstract data type,
with $U = \{0, \dots, 2^k - 1\}$, for some $k \ge 0$.
It supports search with $O(1)$ worst-case step complexity and insert, remove and predecessor 
with $O(k)$ worst-case step complexity.
It consists of a perfect binary tree of height $k$ in which the leaves
are associated with distinct keys in $U$, in increasing order from left to right.
Each node in the binary tree stores a bit.
The bit stored in leaf $i$ is 1 if and only if 
key $i$ is in the set.
The bit of any internal node is 1 if and only if there is a leaf in the node's subtree whose bit is 1.
\begin{figure}[H]
        \centering
        \begin{tikzpicture}
            \node[vertex](root) at (0,3){1};
            \node[vertex](a0) at (-1.5,2){1};
            \node[vertex](a1) at (1.5,2){1};
            \node[vertex](b0) at (-2.2,1){0};
            \node(b0Label) at (-2.2,0.4){0};
            \node(b1Label) at (-0.8,0.4){1};
            \node(b2Label) at (0.8,0.4){2};
            \node(b3Label) at (2.2,0.4){3};
            \node[vertex](b1) at (-0.8,1){1};
            \node[vertex](b2) at (0.8,1){0};
            \node[vertex](b3) at (2.2,1){1};

            \draw[dottedEdge](root) to (a0);
            \draw[dottedEdge](root) to (a1);
            \draw[dottedEdge](a0) to (b0);
            \draw[dottedEdge](a0) to (b1);
            \draw[dottedEdge](a1) to (b2);
            \draw[dottedEdge](a1) to (b3);
        \end{tikzpicture}
        \caption*{A sequential binary trie representing the set $\{1, 3\}$}
\end{figure}
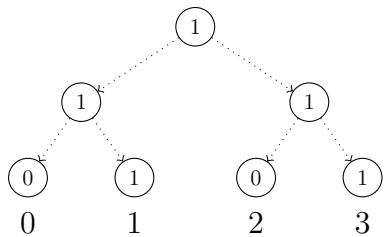

A $search(x)$ operation reads the leaf corresponding to key $x$ and returns \textit{True} 
if its bit is 1, otherwise it returns \textit{False}.
An $insert(x)$ operation sets the bits of the leaf corresponding to key $x$ 
and every ancestor of this leaf to 1.
A $delete(x)$ operation first sets the bit of the leaf corresponding to key $x$ to 0.
Then it traverses up the tree and sets the bit of every ancestor of this leaf to 0, provided 
both of its children have their bits equal to 0.
A $predecessor(x)$ operation first traverses up the tree from the leaf corresponding to key $x$.
If it encounters a node whose left child, $L$, has bit 1 and is not on this path, it 
returns the key corresponding to the right-most leaf with bit 1 in $L$'s subtree.
Otherwise it returns $-1$.

We provide an implementation of Jeremy Ko's lock-free binary trie\cite{jeremy_ko_trie,DBLP:journals/corr/abs-2405-06208}.
The implementation uses \textit{CAS} objects, \textit{writeable CAS} objects, \textit{SWAP} objects, \textit{fetch-and-add} objects, 
bounded min registers and lock-free linked lists.
We consider wait-free implementations of bounded min registers in Chapter \ref{minRegChapter}.
Chapter \ref{trieImplementationChapter} discusses our implementation of 
the lock-free binary trie.
In Chapter \ref{experimentChapter}, we discuss experiments that we used to evaluate 
our implementation.
We compare its performance with that of other implementations of lock-free dynamic 
ordered sets and discuss the results of our experiments.

One interesting component of our implementation of Ko's trie is a lock-free linked 
list that allows many processes to try to insert the same node.
This contrasts with existing implementations of lock-free linked lists, 
in which any particular node may be inserted by at most one process.
The implementation of this linked list is presented in Section \ref{lockFreeLinkedLists}.
In Section \ref{memoryReclamationSection}, we first generalize an epoch-based reclamation scheme\cite{10.1145/2767386.2767436,DBLP:journals/corr/abs-1712-01044} 
and show it can be used to reclaim memory under certain conditions.
Then we show that our implementation of Ko's trie, which uses this variant, 
satisfies these conditions.


\chapter{Model}
We consider an asynchronous shared memory model\cite{doi:https://doi.org/10.1002/0471478210.ch4}.  
The system consists of $N \ge 2$ processes, which communicate by performing operations on 
shared objects in shared memory.
Every process is given an id from $0$ to $N - 1$, and we denote the process with id $i$ as $p_i$.
Every shared object has a set of possible values and a set of 
operations that it supports.
When a process performs an operation instance on a shared object, it 
receives a response, modifies the object, or both.
A \textit{base object} is a shared object that is provided by the system.
A \textit{configuration} contains the values of all base objects in shared memory and the state of every process.
A \textit{step} consists of an operation instance that is performed by a process on a base object, 
the process that performs the instance, the base object the instance is applied to 
and the response of the instance.
A step, $s$, by process $p$ in configuration $C$ leads to a new configuration, $C'$, which is the same as $C$, 
except that the contents of the shared object that $p$ 
performed $s$ on and $p$'s state may have changed.
There is an initial configuration, $C_0$, consisting of the initial state of all processes and initial value of all shared objects.
An \textit{execution} is a sequence beginning from $C_0$ which alternates between steps and configurations.
If the execution is finite then it ends in a configuration.

The base objects offered by the system include \textit{registers}, \textit{AND} objects, 
\textit{fetch-and-add} objects, \textit{SWAP} objects, \textit{CAS} (\textit{compare-and-swap}) objects,
and \textit{writeable CAS} objects. 

A \textit{register}, $r$, stores a value and supports \textit{read()} and \textit{write(v)} operations.
A \textit{read()} operation on a shared object retrieves the value stored by the shared object.
A \textit{write(v)} operation on a shared object overwrites its value with $v$.

A $b$-bit \textit{AND} object, $r$, stores a value of $b$ bits. 
It supports \textit{read()} and \textit{AND(v)} operations, where $v$ is a string of $b$ bits.
An $\mathit{AND(v)}$ operation on $r$ atomically 
computes the bitwise-\textit{AND} of $r$ and $v$ and stores the result in $r$.

A \textit{fetch-and-add} object, $r$, stores an integer. 
It supports \textit{read()} and \textit{fetch-and-add(v)} operations, where $v$ is an integer.
A \textit{fetch-and-add(v)} operation on $r$ atomically stores the sum of $r$'s value and $v$ in $r$, 
while returning the value of $r$ prior to the operation.
An increment of $r$ may be performed using a \textit{fetch-and-add(1)} operation,
whereas a decrement may be performed using a \textit{fetch-and-add(-1)} operation.

A \textit{SWAP} object, $r$, stores a value. It supports \textit{read()} and \textit{swap(v)} operations.
A $\mathit{swap(v)}$ operation on $r$ atomically overwrites the value of 
$r$ with $v$ and returns the value of $r$ prior to the operation.

A \textit{CAS} object, $r$, stores a value from an arbitrary domain. 
It supports \textit{read()} and \textit{CAS(v,k)} operations.
A \textit{CAS(v,k)} operation on $r$ atomically checks if $r$'s value is equal to $v$ and if so,
updates the value of $r$ to $k$. 
We say that the operation was successful if and only if $r$'s value was updated by the 
operation to $k$.
The value that $r$ held prior to the operation is returned. 
This return value can be used to verify if the operation was successful.

A \textit{writeable CAS} object, $r$, is the same as a \textit{CAS} object except that $r$ also supports \textit{write(v)} operations.

Consider two processes, $p$ and $q$, and a \textit{CAS} (or \textit{writeable CAS}) 
object, $r$, whose initial value is $A$.
Suppose both $p$ and $q$ would like to change the value of $r$ to $B$,
provided its value has never previously been equal to $B$, 
before changing the value of $r$ back to $A$.
To accomplish this, each process will first perform a $\mathit{CAS(A, B)}$
operation on $r$ and, if this \textit{CAS} is successful, they
will perform a $\mathit{CAS(B, A)}$ operation on $r$.
Suppose $p$ performs its \textit{CAS} first, which is successful,
then $p$ performs the \textit{CAS} to update the value of $r$ 
back from $B$ to $A$.
Now when $q$ performs its first \textit{CAS}, it is successful 
even though $r$ had previously been updated to $B$,
violating our requirements.
Process $q$ expected its \textit{CAS} would only succeed if 
$r$ had not previously been changed, but since $p$ 
updated $r$ back to $A$, the \textit{CAS} by $q$ was successful.
This is known as the \textit{ABA problem} since it occurs when the 
value of a \textit{CAS} object is updated back to a value it held 
previously (in this case, $A$) and at least one process that is unaware that 
its value has changed is about to perform a \textit{CAS} to update it from this value.
We will discuss other examples of the \textit{ABA problem} in 
this report.

Some shared objects, such as \textit{LL/SC} objects, are not available in the system.
An \textit{LL/SC} object, \textit{x}, stores a value and supports two operations,
\textit{LL (load-linked)} and \textit{SC (store conditional)}.
An \textit{x.LL()} operation returns the value of \textit{x}. Any process, $p$, must have performed an 
\textit{x.LL()} instance before it may perform an instance of \textit{SC} on $x$.
An \textit{x.SC(v)} operation on \textit{x} by $p$ will update the value of \textit{x} to $v$ and return \textit{True},
provided no process has performed an \textit{SC} on \textit{x} since $p$ performed its most recent \textit{x.LL()} instance.
Otherwise, the value of \textit{x} is not changed and the operation returns \textit{False}.
Since an \textit{SC} instance by a process only succeeds if the value of the \textit{LL/SC} object
has not changed since the last \textit{LL} instance by the process, 
\textit{LL/SC} objects are not susceptible to the \textit{ABA problem}.

We discuss implementations of shared objects which are not available in the system.
An implementation provides a representation of the implemented object in shared memory using base objects and provides algorithms for the operations that the implemented object supports.
We consider algorithms in which at least one operation on a base object is performed.
An instance of an operation by a process on an implemented object starts when 
the instance performs the first operation on a base object in the algorithm for the operation. 
This operation instance ends when it has performed 
the last operation on a base object in this algorithm.
Two operation instances, $A$ and $B$, \emph{overlap} if $A$ does not end before $B$ starts and $B$ does not end before $A$ starts.
An implemented object is \textit{linearizable} if, for any execution, $e$, there exists a sequence, $o$, 
of all completed instances of operations
on the object in $e$, and a subset of all incomplete instances, such that:
\begin{itemize}
    \item if an instance, $A$, ends in $e$ before the start of another instance, $B$, then $A$ precedes $B$ in $o$, and,
    \item every complete instance in $e$ receives the same response as 
    it would in an execution in which the instances in $o$ 
    are performed sequentially.
\end{itemize}

An implemented object is \textit{wait-free} if, whenever any process is performing an operation 
instance on the object and takes a sufficient finite number of steps, 
this operation instance ends.
An implemented object is \textit{lock-free} if, whenever any process is performing an operation instance on the object and takes a sufficient finite number of steps, 
some operation instance on the object ends. 
However, this may not be the operation instance that this process was performing.
Every wait-free implemented object is also lock-free.
An operation instance by some process on an object that is lock-free but not wait-free 
could take unbounded numbers of steps without ending, provided operation instances by other processes continue to finish.

\chapter{Related Work}
This section introduces some related work which is relevant to our implementation of Ko's Trie.

Valois\cite{DBLP:conf/podc/Valois95}, Harris\cite{10.5555/645958.676105}, and Fomitchev and Ruppert\cite{10.1145/1011767.1011776,fomitchevThesis} 
produced implementations of shared singly linked lists from \textit{CAS} objects, 
which they used to implement dynamic ordered sets.

Valois implemented the first shared linked list based on \textit{CAS}
objects that did not use locks and was not based on a universal construction.
His linked list consists of normal nodes and auxiliary nodes.
An auxiliary node only has a \textit{next} field. 
There is at least one auxiliary node between every two normal nodes in the list.
Every normal node stores a \textit{key} from a totally ordered universe, a \textit{next} field and a \textit{backlink}.
For any normal node that is in the list, its \textit{next} field points to the auxiliary node
following it in the list.
The \textit{backlink} of a normal node is set to point to a normal node which occurs earlier in the list.
This allows a process to resume from a normal node that occurs earlier in the list 
when the normal node it is currently visiting is removed.
Two sentinel normal nodes, $first$ and $last$, are always at the beginning and end of the list respectively.
At any time, the set represented by the data structure consists of the keys of normal nodes which are between $first$ and $last$.
The normal nodes in the list are in strictly ascending order: if there is a path 
of auxiliary nodes from a normal node $A$ to a normal node $B$, then $A.key < B.key$.

If a key $k$ is not in the set, a process may insert it into the set.
First, the process traverses to the earliest normal node, $v$, whose key is greater than $k$.
Next, the process allocates an auxiliary node, which points to $v$,
and a normal node, $u$, which has \textit{key} $k$ and whose \textit{next} field points to the auxiliary node.
Finally, the process performs a \textit{CAS} on the auxiliary node 
preceding $v$, to update its \textit{next} field from pointing to $v$ to pointing to $u$.

If a key is in the set, a process may remove it from the set by performing a \textit{CAS} on the auxiliary 
node that precedes the normal node containing the key.
The preceding auxiliary node is updated to point to the auxiliary node which follows the normal node with the key.
After this, the process will try to 
remove the auxiliary node which followed the normal node from the list using a \textit{CAS}.
It may fail to do so if another process is performing an operation instance,
leaving this extra auxiliary node in the list.
Executions exist in which, despite only two processes performing operation instances, 
the list becomes filled with an arbitrary number of consecutive auxiliary nodes.
This can make operation instances very expensive.
However, in any configuration in which no process is performing an operation instance, 
there are no extra auxiliary nodes in the list.
\medskip

Harris implemented a linearizable lock-free linked list.
His linked list consists of nodes which contain a \textit{key} and a \textit{next} field.
If a node is in the list, its \textit{next} field, which is a \textit{CAS} object, stores a pointer to the following node in the list,
along with a \textit{mark} bit which is stored in the lowest order bit of its \textit{next} field.
Before the node is removed from the list, this \textit{mark} bit becomes permanently set to 1.
Once this happens, the node's \textit{next} field cannot be changed for the rest of the execution
and the node is considered to be deleted.
We say a node is \textit{marked} if and only if the \textit{mark} bit of its \textit{next} field
is 1.

The elements in the linked list are sorted in strictly ascending order, that is,
for every two consecutive nodes, $A$ and $B$,
$A.key < B.key$.
Two sentinel nodes, \textit{head} and \textit{tail}, are always at the start and end 
of the list, respectively. 
At any time, the set represented by the data structure contains a key if and only if there 
exists an unmarked node containing the key that is reachable from \textit{head}.

When searching through the list, sequences of one or more marked nodes can be physically removed via a \textit{CAS} on 
the \textit{next} field of the unmarked node immediately preceding the sequence.
If the node a process is currently visiting becomes marked, it could be necessary for the process 
to restart its traversal from \textit{head}.
For example, a process trying to insert a node $B$ may first traverse the list to find two consecutive 
nodes $A$ and $C$, such that $A.key < B.key < C.key$.
Next, the process will update $B.next$ to point to $C$ and use a \textit{CAS} to attempt 
to update $A.next$ to point to $B$, thus inserting $B$ between $A$ and $C$.
If $A$ is marked, this \textit{CAS} will fail.
Since the process no longer has a pointer to an earlier unmarked node in the list, 
it can no longer reach the latest unmarked node whose key is less than $B.key$.
Therefore, the process will have to restart its traversal from \textit{head} to insert $B$ into the list.
A process trying to remove a node $B$ from the list could experience the same problem.
It may traverse to the node, $A$, which immediately precedes $B$ and use a \textit{CAS}
to attempt to remove $B$. However, if $A$ is marked, this \textit{CAS} will 
fail.
Since $A$ is marked, the process must restart its traversal from \textit{head} to 
locate an unmarked node which precedes $B$ so that it can remove $B$ from the list.
An insert or delete instance could be forced to restart from \textit{head} for every delete 
that overlaps with it.
The worst-case amortized step complexity of Harris's list is 
$\Omega(mc)$, where $m$ is 
the number of elements in the list at the beginning of the execution and $c$ is the maximum number of processes 
that perform operation instances at the same time as each other in the execution.

\medskip
Fomitchev and Ruppert implemented a linearizable lock-free linked list 
based on \textit{CAS} with amortized step complexity $O(c + m)$.
The design of this structure combines some new ideas with the \textit{backlinks} of Valois and \textit{mark} bits of Harris.
It ensures that processes may recover gracefully upon failing a \textit{CAS} step on a node and discovering the node is now \textit{marked}.
Every node in the list stores a \textit{key}, a \textit{successor} field and a \textit{backlink} field. 
The \textit{successor} field of a node is a \textit{CAS} object that stores three subfields:
\begin{itemize}
    \item $next$, which stores a pointer to a node,
    \item $mark$, a bit which is initially 0, and
    \item $flag$, a bit which is initially 0.
\end{itemize}
If a node is in the linked list, its \textit{next} field stores a pointer to the node which follows it in the linked list.
The \textit{mark} and \textit{flag} bits are used to help coordinate the removal of nodes from the list.
We say that a node is \textit{marked} if its \textit{mark} bit is 1
and it is \textit{unmarked} otherwise.
Similarly, we say that a node is \textit{flagged} if its \textit{flag} bit is 1 
and it is \textit{unflagged} otherwise.
A node cannot be marked and flagged at the same time.
If a node, $A$, is flagged, this indicates that the following node, $B$, is being removed from the linked list.
When a process encounters $A$ in this state, it will help remove $B$ from the linked list.
In the following figures, \textcolor{orange}{F} denotes a flagged node and \textcolor{red}{M} denotes a marked node.
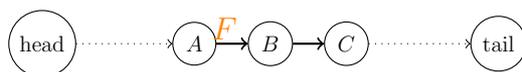
\begin{figure}[H]
        \centering
        \begin{tikzpicture}
            \node[vertex](head) at (0,0){head};
            \node[vertex](A) at (2,0){$A$};
            \node[vertex](B) at (3,0){$B$};
            \node[vertex](C) at (4,0){$C$};
            \node[vertex](tail) at (6,0){tail};

            \draw[dottedEdge](head) to (A);
            \draw[edge](A) to (B);
            \draw[edge](B) to (C);
            \draw[dottedEdge](C) to (tail);

            \node[color=orange](flag) at (2.4,0.2){$F$}; 
        \end{tikzpicture}
    \end{figure}
\noindent First, the process will set the backlink of $B$ to point to $A$.
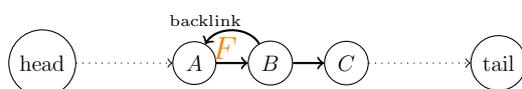
\begin{figure}[H]
        \centering
        \begin{tikzpicture}
            \node[vertex](head) at (0,0){head};
            \node[vertex](A) at (2,0){$A$};
            \node[vertex](B) at (3,0){$B$};
            \node[vertex](C) at (4,0){$C$};
            \node[vertex](tail) at (6,0){tail};

            \draw[dottedEdge](head) to (A);
            \draw[edge](A) to (B);
            \draw[edge](B) to (C);
            \draw[dottedEdge](C) to (tail);

            \node[color=orange](flag) at (2.4,0.2){$F$}; 
            \draw[edge](B) to [bend right=60] (A) node[node font=\tiny, above=4pt]{backlink}; 
        \end{tikzpicture}
    \end{figure}
    \noindent If $B$ is unmarked and unflagged, the process performs a \textit{CAS} to try to mark $B$
    leaving $B.next$ unchanged.
    If $B$ is flagged, the process will try to help to remove the node following $B$ 
    and update $B$ to being unflagged.
    The process continues performing these steps until $B$ is marked.
    In the following figure, because $B$ is unflagged and unmarked, the process performs a \textit{CAS} to mark $B$. 
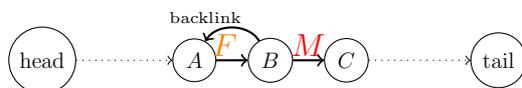
    \begin{figure}[H]
        \centering
        \begin{tikzpicture}
            \node[vertex](head) at (0,0){head};
            \node[vertex](A) at (2,0){$A$};
            \node[vertex](B) at (3,0){$B$};
            \node[vertex](C) at (4,0){$C$};
            \node[vertex](tail) at (6,0){tail};

            \draw[dottedEdge](head) to (A);
            \draw[edge](A) to (B);
            \draw[edge](B) to (C);
            \draw[dottedEdge](C) to (tail);

            \node[color=orange](flag) at (2.4,0.2){$F$}; 
            \draw[edge](B) to [bend right=60] (A) node[node font=\tiny, above=4pt]{backlink}; 
            \node[color=red](mark) at (3.5,0.2){$M$}; 
        \end{tikzpicture}
    \end{figure}
\noindent Finally, the \textit{successor} field of $A$ is updated via \textit{CAS} so that $A$ is unflagged and pointing to
    the node $C$ that immediately follows $B$.
    This completes the removal of $B$ from the list.
    
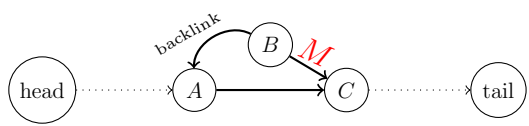
\begin{figure}[H]
        \centering
        \begin{tikzpicture}
            \node[vertex](head) at (0,0){head};
            \node[vertex](A) at (2,0){$A$};
            \node[vertex](B) at (3,0.6){$B$};
            \node[vertex](C) at (4,0){$C$};
            \node[vertex](tail) at (6,0){tail};

            \draw[dottedEdge](head) to (A);
            \draw[edge](A) to (C);
            \draw[edge](B) to (C);
            \draw[dottedEdge](C) to (tail);
            \draw[edge](B) to [bend right=60] (A) node[node font=\tiny, rotate=30, above=9pt]{backlink}; 
            \node[color=red,rotate=-25](mark) at (3.6,0.5){$M$}; 
        \end{tikzpicture}
    \end{figure} 
    
Since the \textit{backlink} of a node is always set to a preceding node before the node is marked,
processes that are currently visiting a marked node can always follow a path backwards to an earlier unmarked node in the list.
This means they do not need to restart their traversal from \textit{head},
helping reduce the step complexity as compared to Harris's list.
In addition to helping when encountering a node that is flagged, a process 
will also help when it encounters a node, $mNode$, that is marked and still in the list.
First, the process will read $mNode$'s \textit{backlink} to obtain a 
pointer to the node, \textit{pNode}, that precedes $mNode$ in the list.
Then the process will perform a \textit{CAS} on \textit{pNode}'s \textit{successor},
updating it so that it is unflagged and pointing to the node which follows $mNode$.
This completes the removal of $mNode$ from the list.

Two sentinel nodes, \textit{head} and \textit{tail}, are always at the start and end 
of the list respectively. 
At any time, the set represented by the data structure contains a key if and only if there 
exists an unmarked node containing the key that is reachable from \textit{head}.
The elements in the linked list are sorted in strictly ascending order, that is,
for every two consecutive nodes, $A$ and $B$,
$A.key < B.key$.

To search the set for a key, $x$, a process, \textit{p}, will traverse to find two consecutive unmarked 
nodes, $A$ and $C$, such that $A.key \le x < C.key$, returning whether $A.key = x$.

The steps to insert a key are virtually the same as those for Harris's list.
To insert a key $x$ into the set, a process, \textit{p}, will create a node $B$ which stores key $x$,
and then search for two consecutive unmarked nodes, $A$ and $C$, such that $A.key \le x < C.key$.
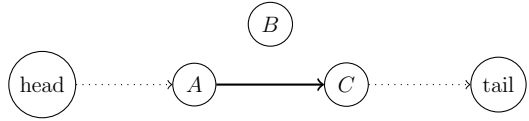
\begin{figure}[H]
        \centering
        \begin{tikzpicture}
            \node[vertex](head) at (0,0){head};
            \node[vertex](A) at (2,0){$A$};
            \node[vertex](B) at (3,0.8){$B$};
            \node[vertex](C) at (4,0){$C$};
            \node[vertex](tail) at (6,0){tail};

            \draw[dottedEdge](head) to (A);
            \draw[edge](A) to (C);
            \draw[dottedEdge](C) to (tail);
        \end{tikzpicture}
\end{figure}
\noindent If $A.key = x$ then $x$ is already in the set, so $p$ ends its operation instance.
If $A.key < x$, $p$ first updates $B$'s \textit{successor} field to point to $C$, as shown in the figure.
    
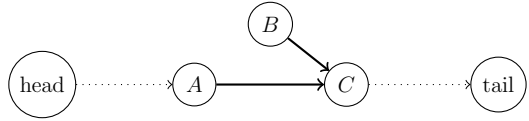
\begin{figure}[H]
        \centering
        \begin{tikzpicture}
            \node[vertex](head) at (0,0){head};
            \node[vertex](A) at (2,0){$A$};
            \node[vertex](B) at (3,0.8){$B$};
            \node[vertex](C) at (4,0){$C$};
            \node[vertex](tail) at (6,0){tail};

            \draw[dottedEdge](head) to (A);
            \draw[edge](A) to (C);
            \draw[edge](B) to (C);
            \draw[dottedEdge](C) to (tail);
        
        \end{tikzpicture}
    \end{figure}
\noindent Then $p$ performs a \textit{CAS} on $A$'s \textit{successor}, updating it from pointing to $C$ to pointing to $B$,
provided $A$ is unflagged and unmarked.
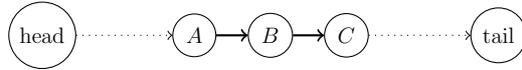
\begin{figure}[H]
    \centering
    \begin{tikzpicture}
        \node[vertex](head) at (0,0){head};
        \node[vertex](A) at (2,0){$A$};
        \node[vertex](B) at (3,0){$B$};
        \node[vertex](C) at (4,0){$C$};
        \node[vertex](tail) at (6,0){tail};

        \draw[dottedEdge](head) to (A);
        \draw[edge](A) to (B);
        \draw[edge](B) to (C);
        \draw[dottedEdge](C) to (tail);
    \end{tikzpicture}
\end{figure}
\noindent If this \textit{CAS} is unsuccessful, $p$ will resume its search from $A$ 
to find two consecutive nodes whose keys are at most $x$ and greater than $x$, respectively,
so that it may find the correct place to insert $B$.

To remove a key $x$ from the set, a process, \textit{p}, will search for two consecutive unmarked nodes, 
$A$ and $B$, such that $A.key < x \le B.key$.
If $B.key \neq x$, then $x$ is not in the set, so $p$ ends its operation instance.
If $B.key = x$, $p$ will perform a \textit{CAS} on $A$'s \textit{successor} field, trying to update 
it so that $A$ is flagged and still pointing to $B$.
Assuming this \textit{CAS} is successful, $p$ performs helping to remove $B$ as we previously described.
If this \textit{CAS} is not successful, $p$ will resume its search 
from $A$ find two consecutive nodes whose keys are less than $x$ and at least $x$, respectively,
so that it may find and remove a node containing $x$.

\medskip

A \textit{Destination} object stores a value and supports \textit{read()}, \textit{write(v)} and \textit{copy(ptr)} operations.
A \textit{copy(ptr)} operation atomically copies the contents of a shared object into the \textit{Destination} object,
where $ptr$ is a pointer to the shared object whose contents will be copied.
Blelloch and Wei\cite{DBLP:conf/wdag/BlellochW20} give a wait-free implementation of single-writer \textit{Destination} objects 
from \textit{CAS} objects with $O(1)$ step complexity.
For every single-writer \textit{Destination} object,
only one process may perform \textit{write} and \textit{copy} operation instances on the object,
whereas all processes may perform \textit{read} instances on it.
They implement a single-writer \textit{Destination} object from a weak version of an \textit{LL/SC} object and a register.
The weak \textit{LL/SC} object supports a \textit{wLL} operation instead of an \textit{LL} operation.
A \textit{wLL} operation is the same as an \textit{LL} operation, except that an instance of \textit{wLL}
may return $\bot$ if it overlaps with a successful \textit{SC} instance,
in which case we say this instance has failed.
A process that performs a \textit{wLL} that fails must perform a successful \textit{wLL}
before it may perform a successful \textit{SC} instance.

In their implementation, every single-writer \textit{Destination} object, $O$, 
is represented by a weak \textit{LL/SC} object, \textit{data}, and a register, \textit{old}.
Let $p$ be the process that can perform \textit{copy} and \textit{write} instances on $O$.
Two fields, \textit{val} and \textit{copying}, are stored in \textit{data}.
The \textit{copying} field stores a single bit, which is initially 0.
When \textit{copying} is equal to 0, \textit{val} stores the 
current value of $O$.
When \textit{copying} is equal to 1, \textit{val} stores a pointer 
to a shared object whose contents are being copied into $O$ by $p$.
When $p$ performs an instance of \textit{write} or \textit{copy} on $O$,
it first writes the current value of $O$ into \textit{old}
before the value of $O$ is changed.

When any process, $q$, wants to read $O$, it first performs an instance of \textit{wLL} on \textit{data}.
If this \textit{wLL} is unsuccessful, $q$ will perform a second instance of \textit{wLL} on \textit{data},
reading and returning the value of \textit{old} if this second \textit{wLL} is unsuccessful.
Suppose either $q$'s first \textit{wLL} is successful, or it fails but $q$'s second \textit{wLL} is 
successful.
Let $v'$ and $c'$ be the values of \textit{val} and \textit{copying}, respectively,
that $q$ obtained from \textit{data}.
If $c'$ is equal to 0, then $v'$ was the value of $O$ when $q$ performed 
its successful \textit{wLL} of \textit{data}, so $q$ returns $v'$.
Suppose $c'$ is equal to 1, so $v'$ holds a pointer to a shared
object whose value is being copied into $O$ by $p$.
In this case, $q$ will first read the value, $z$, of 
this shared object.
Next, $q$ performs an \textit{SC} on \textit{data}, 
trying to update \textit{val} to $z$ and \textit{copying} to 0.
If this \textit{SC} is successful, $q$ returns $z$.
Suppose this \textit{SC} is not successful. 
In this case, some other process finished copying the value of the shared object into $O$.
Therefore $q$ will perform another \textit{wLL} on \textit{data}.
If this \textit{wLL} is unsuccessful or it is successful but \textit{copying} was equal to 1, $q$
reads and returns the value of \textit{old}.
If this \textit{wLL} is successful but \textit{copying} is equal to 0, 
$q$ returns the value of \textit{val} it obtained during this \textit{wLL} instance.

To copy the contents of a shared object pointed to by $ptr$ into $O$, 
process $p$ will first perform an instance of \textit{wLL} on \textit{data}.
Next, $p$ performs an \textit{SC} instance on \textit{data} to update \textit{val} to 
\textit{ptr} and \textit{copying} to 1.
Then $p$ reads the contents, $v$, of the shared object pointed to by \textit{ptr}.
Now $p$ performs another \textit{wLL} instance on \textit{data} and, 
if successful, $p$ performs an \textit{SC} instance on \textit{data} to 
try to update \textit{val} to $v$ and \textit{copying} to 0.
If $p$'s second \textit{wLL} instance is unsuccessful, or it succeeds but $p$'s second \textit{SC} instance 
is unsuccessful, some other process finished copying the contents of the shared object 
into $O$, so $p$ returns.
Otherwise, $p$ finishes copying $v$ into $O$.
The first instance of \textit{wLL} and \textit{SC} by $p$ are guaranteed to succeed.
This is because \textit{copying} is equal to 1 until $p$ has finished performing 
these instances and other processes may only perform successful \textit{SC} instances 
when \textit{copying} is equal to 1.

To write a value, $v$, into $O$, process $p$ will first perform an instance of \textit{wLL} on \textit{data}.
Then $p$ performs an \textit{SC} instance on \textit{data} to update \textit{val} to $v$, 
leaving \textit{copying} still equal to 0.
Both the \textit{wLL} and \textit{SC} instances by $p$ are guaranteed to succeed,
since \textit{copying} remains equal to 0 while $p$ is writing $v$ to $O$.

\chapter{Wait-Free Bounded Min Registers} \label{minRegChapter}
An $X$-bounded min register, $m$, is a shared object that stores a non-negative integer value in the range 0 to $X-1$.
Initially the value of $m$ is $X-1$.
It supports $\mathit{minRead()}$ and $\mathit{minWrite(v)}$ operations, where $v \in \{0, \dots, X-1\}$.
A $\mathit{minWrite(v)}$ of \textit{m} will overwrite the value of $m$ with $v$ if $v$ is less than $m$'s current value.
A $\mathit{minRead()}$ of \textit{m} returns the current value of $m$.
We say that $m$ is $X$-bounded since it can hold $X$ distinct values during an execution.

In this section, we present an existing implementation of bounded min registers due to Aspnes, Attiya, and Censor-Hillel\cite{10.1145/2108242.2108244}.
We then give an implementation of a wait-free $(b+1)$-bounded min register from a $b$-bit AND object.
Finally, we describe an implementation of bounded min registers which combines both these approaches.

\section{Aspnes, Attiya and Censor-Hillel's Bounded Min Register} \label{aspnesMinReg}
Aspnes, Attiya and Censor-Hillel\cite{10.1145/2108242.2108244} show inductively how to produce a linearizable, wait-free $X$-bounded max register
from $X-1$ single-bit registers. We present an analogous version of their implementation for min registers instead.

A 2-bounded min register $m$ may be implemented from a single bit register $r$, which has initial value 1.
To perform an instance of a \textit{minRead()} operation on $m$, a process reads the value of $r$ and returns it.
To perform an instance of a \textit{minWrite(0)} operation on $m$, a process writes 0 to $r$.
A \textit{minWrite(1)} operation on $m$ does nothing. Every operation instance performs at most a single step, 
thus this min register is wait-free. Instances of operations which perform a step may be linearized 
in the order their steps are performed, and those that do not perform a step may be ordered arbitrarily.

Suppose we have two wait-free and linearizable min registers $\mathit{left}$ and $\mathit{right}$, 
which are $L$-bounded and $R$-bounded respectively, and a 2-bounded min register $\mathit{switch}$.
Then we can implement an $(L+R)$-bounded min register, $m$, as follows.

To perform a $\mathit{minRead()}$ of $m$, a process, \textit{p}, performs a \textit{minRead()} of \textit{switch}.
If the value was 1, $p$ performs a $\mathit{minRead()}$ on $\mathit{right}$ and returns the result plus $L$. 
If the value was 0, $p$ performs a $\mathit{minRead()}$ on $\mathit{left}$ and returns the result.

To perform a $\mathit{minWrite(v)}$ of $m$ when $v$ is at least $L$, 
$p$ performs a \textit{minRead()} of $\mathit{switch}$. 
If the value was 1, then $p$ performs an instance of $\mathit{minWrite(v - L)}$ on $\mathit{right}$.
If the value was 0, then $p$ returns.
To perform a $\mathit{minWrite(v)}$ instance on $m$ when $v$ is less than $L$,
$p$ performs a $\mathit{minWrite(v)}$ instance on $\mathit{left}$ and then performs a \textit{minWrite(0)} instance on 
\textit{switch}.

This figure shows their implementation of a 16-bounded min register, $m$, 
which is implemented from two 8-bounded min registers, \textit{left} and \textit{right}.
\begin{figure}[H]            
    \centering
    \begin{tikzpicture}
        \node(leftLabel) at (-2,-1){left};
        \node(rightLabel) at (2,-1){right};
        \node(switchLabel) at (-1,0){switch};
        \node[vertex](switch) at (0,0){1};
        \node[subtree](left) at (-1.5,-1){7};
        \node[subtree](right) at (1.5,-1){7};

        \draw[dottedEdge](switch) -- (left);
        \draw[edge](switch) -- (right);
    \end{tikzpicture}
\end{figure}
\noindent In this figure, a process performs a \textit{minWrite(12)} instance on $m$. 
It performs a \textit{minRead} of \textit{switch} and since \textit{switch=1}, 
it performs a \textit{minWrite(4)} instance on \textit{right}.
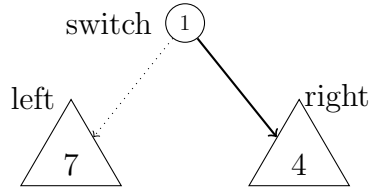
\begin{figure}[H]
    \centering
    \begin{tikzpicture}
        \node(leftLabel) at (-2,-1){left};
        \node(rightLabel) at (2,-1){right};
        \node(switchLabel) at (-1,0){switch};
        \node[vertex](switch) at (0,0){1};
        \node[subtree](left) at (-1.5,-1){7};
        \node[subtree](right) at (1.5,-1){4};

        \draw[dottedEdge](switch) -- (left);
        \draw[edge](switch) -- (right);
    \end{tikzpicture}
\end{figure}
\noindent In this next figure, another process performs a \textit{minWrite(2)} instance on $m$.
It performs a \textit{minWrite(2)} instance on \textit{left}, then performs a \textit{minWrite(0)} instance on \textit{switch}.
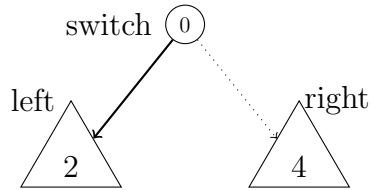
\begin{figure}[H]
    \centering
    \begin{tikzpicture}
        \node(leftLabel) at (-2,-1){left};
        \node(rightLabel) at (2,-1){right};
        \node(switchLabel) at (-1,0){switch};
        \node[vertex](switch) at (0,0){0};
        \node[subtree](left) at (-1.5,-1){2};
        \node[subtree](right) at (1.5,-1){4};

        \draw[edge](switch) -- (left);
        \draw[dottedEdge](switch) -- (right);
    \end{tikzpicture}
\end{figure}

Their implementation of an $X$-bounded min register, $m$, where $X > 1$ 
is a binary tree of height $\ceil{\log_2 X} - 1$ 
whose leaves are 2-bounded min registers. 
Each operation instance on $m$ performs at most $\ceil{\log_2 X}$ steps.
Therefore $m$ is wait-free.
The proof that $m$ is linearizable is more involved.
In Section $\ref{bigMinReg}$, we present an implementation of a wait-free bounded min register which generalizes the min register discussed in this section 
to a $k$-ary tree, along with a proof of its linearizability. 
Since $m$ is a special case of the bounded min register discussed in that section, 
the proof of linearizability also applies to $m$.

\section{A Wait-Free Bounded Min Register from an AND object} \label{smallMinReg}
Using a $b$-bit AND object, $r$, we can implement a $(b+1)$-bounded min register, $m$, that 
only requires one step per operation instance.
The value of $m$ is represented by the value of $r$ in unary.
When the value of $m$ is $i$, $r$ contains $\mathit{0^{b-i}1^i}$,
that is, $(b-i)$ 0-bits followed by $i$ 1-bits.
Initially, every bit of $r$ is 1, indicating $m$ has the value $b$. 

To perform a $\mathit{minRead()}$ of $m$, a process will read $r$, 
and return the number of consecutive ones that were in $r$ 
starting from the right.
%
To perform a $\mathit{minWrite(v)}$ of $m$, a process will 
set the leftmost $b-v$ bits of $r$ to 0.
It will do this by performing an \textit{AND} of $r$ 
with $0^{b-v}1^v$. If the bits of $r$ are numbered $b-1$ to $0$ from left to right,
this step will set bits $b-1$, $b-2$, $\dots$, $v$ to 0.
If $v$ is greater than zero, bits $v-1$ to $0$ of $r$ are left unchanged.
This figure shows the initial state of a 9-bounded min register, $m$, which is 
implemented using an 8-bit \textit{AND} object.
\begin{figure}[H]
    \centering
    \begin{tikzpicture}[
        node distance=0pt,
            start chain = A going right,
            X/.style = {rectangle, draw,
                        minimum width=2ex, minimum height=3ex,
                        outer sep=0pt, on chain}]
        \foreach \i in {1,1,1,1,1,1,1,1}
        \node[X] {\i};
    \end{tikzpicture}
\end{figure}
This figure shows the contents of this \textit{AND} object
following a \textit{minWrite(5)} instance on $m$.
    \begin{figure}[H]
        \centering
        \begin{tikzpicture}[
        node distance=0pt,
         start chain = A going right,
            X/.style = {rectangle, draw,
                        minimum width=2ex, minimum height=3ex,
                        outer sep=0pt, on chain}]
        \foreach \i in {0,0,0,1,1,1,1,1}
        \node[X] {\i};
        \end{tikzpicture}
    \end{figure}


This implementation is wait-free as each operation instance performs a single step.
The instances of $\mathit{minRead()}$ and $\mathit{minWrite(v)}$ operations 
are linearized in the order that their steps are performed.

On modern 64 bit systems, this implementation provides a simple and efficient 
wait-free 65-bounded min register. Such machines support an instruction called \textit{tzcnt}, 
which returns the number of consecutive 0's starting from the right of a local register.
It is straightforward to use this instruction to instead obtain the number of consecutive 1's starting from the right.

\section{Combining Both Approaches for Building Bounded Min Registers} \label{bigMinReg}
The min register described in Section \ref{aspnesMinReg} used \textit{switch} as a 2-bounded min register.
A $k$-bounded min register implemented from a $\mathit{(k-1)}$-bit AND object, as discussed in Section \ref{smallMinReg}, 
can be used to implement an
$X$-bounded min register based on a $k$-ary tree of height $\ceil{\log_k X} - 1$ instead of $\ceil{\log_2 X} - 1$, for any $k \ge 2$.

Suppose we can implement a $Y$-bounded min register that is wait-free and linearizable.
Then we can implement a $kY$-bounded min register, $m$, from a wait-free and linearizable $k$-bounded min register \textit{switch}
and $k$ wait-free and linearizable $Y$-bounded min registers, denoted $T_i$, for $0 \le i < k$.

To perform a \textit{minRead()} on $m$, a process, \textit{p}, will first 
use $\mathit{minRead()}$ to get the current value $i$ of \textit{switch}.
Then $p$ will output $iY$ plus the value of the min register $T_i$ which $p$ obtains from a \textit{minRead()} instance.
\begin{figure}[!ht]
    \begin{algorithmic}[1]
        \State $minRead()$
        \Indent
            \State $\mathit{i \gets switch.minRead()}$ 
            \State \Return $\mathit{T_i.minRead() + iY}$
        \EndIndent 
    \alglinenoNew{alg3}
    \alglinenoPush{alg3}
    \end{algorithmic}
\end{figure}

To perform a $\mathit{minWrite(iY + j)}$ on $m$, where $0 \le j < Y$, $p$ uses a \textit{minRead()}
instance to get the current value $d$ of $\mathit{switch}$. 
Then the current value of $m$ is at least $dY$ and is less than $(d+1)Y$.
If $\mathit{d < i}$, then $v$ is 
larger than the current value of $m$, so $p$ just returns.
If $d \ge i$, $p$ performs $\mathit{minWrite(j)}$ on $\mathit{T_i}$.
If $d > i$, then $p$ also performs \textit{minWrite(i)} on \textit{switch}.
\begin{figure}[!ht]
    \begin{algorithmic}[1]
    \alglinenoPop{alg3}
        \State $minWrite(iY + j)$, where $0 \le j < Y$
        \Indent
            \State $\mathit{d \gets switch.minRead()}$
            \If{$d \ge i$}
                \State $\mathit{T_i.minWrite(j)}$
                \If{$d > i$}
                    \textit{switch.minWrite(i)}
                \EndIf
            \EndIf
        \EndIndent   
    \alglinenoPush{alg3}
    \end{algorithmic}
\end{figure}

We will now show that $m$ is linearizable.
\begin{theorem}
    If \textit{switch} is a linearizable $k$-bounded min register and every $T_i$ is a linearizable $s$-bounded min register, then
    $m$ is a linearizable $kY$-bounded min register.
\end{theorem}
\begin{proof}
    Let $e$ be any execution. 
    Since \textit{switch} and every $\mathit{T_i}$ are linearizable,
    we can treat all instances of operations performed on them in $e$ as if they were atomic.
    First, we will partition all complete instances of \textit{m.minWrite(v)} and 
    \textit{m.minRead()} and some incomplete instances of \textit{m.minWrite(v)} in $e$ into $k+1$ sets.
    We then construct a sequence $o$ which contains all instances in these sets.

    \begin{itemize}
        \item For every $i \in \{0, \dots, k-1\}$, $\mathit{S_i}$ is the set containing:
        \begin{itemize}
            \item All complete $\mathit{m.minRead()}$ instances whose \textit{minRead()} of \textit{switch} returns $i$, and,
            \item all $\mathit{m.minWrite(iY + j)}$ instances which perform $\mathit{T_i.minWrite(j)}$, for some $j \in \{0, \dots, Y-1\}$.
        \end{itemize}
        \item $\mathit{S_{<}}$ is the set of all complete $\mathit{m.minWrite(iY + j)}$ instances whose \textit{minRead()}
        of \textit{switch} returns an integer less than $i$, for $i \in \{1, \dots, k-1\}$ and $j \in \{0, \dots, Y-1\}$.
    \end{itemize}
    
    Observe that every complete instance of a \textit{minRead()} operation on $m$ in $e$ 
    is contained in $S_i$ for exactly one value of $i$.
    Every complete instance, $I$, of a \textit{minWrite(iY + j)} operation on $m$ is contained in exactly one of $\mathit{S_{<}}$ or $\mathit{S_i}$:
    The instance $I$ begins with
    \textit{minRead()} to obtain the value $d$ of \textit{switch}.
    If $d < i$, then, by definition, $I$ is in $\mathit{S_{<}}$. 
    Otherwise, $d \geq i$ and $I$ performs an instance of \textit{minWrite(j)} on $\mathit{T_i}$,
    so $I \in \mathit{S_i}$.

    An incomplete instance of \textit{minWrite($iY+j$)} on $m$ that performs $\mathit{T_i.minWrite(j)}$
    is not in $S_{<}$, since it read a value $d \geq i$ when it performed \textit{minRead()} on \textit{switch}.
    Hence it is only in $S_i$.
    Thus an instance of an operation on $m$ is in at most one set.
    \\\\ \indent
    We construct a sequence $o$ of the instances in these sets, using the following rules:
    \begin{enumerate}
        \item We order all instances in $S_{i'}$ before those in $S_i$, for every $i$ and $i'$ such that $0 \le i < i' < k$.
        \item For every $i$, we order the instances in $S_i$ by the order in which they access $T_i$.
\end{enumerate}
Note that this orders all instances in  $\cup\{S_i \ |\ 0 \leq i < k\}$.
The next two rules insert the instances in $\mathit{S_{<}}$ into this sequence.
Any instance $B$ of $\mathit{m.minWrite(iY + j)}$ in $\mathit{S_{<}}$ receives a value $d < i$
from its \textit{minRead()} instance of \textit{switch}.
Therefore, before $B$ starts, some process must have performed \textit{switch.minWrite(d)}, 
which completes an instance of \textit{m.minWrite(dY + e)}, for some $0 \le e < Y$.
\begin{enumerate}
\setcounter{enumi}{2}
        \item For any instance $B$ of $\mathit{m.minWrite(iY + j)}$ in $\mathit{S_{<}}$, where $0 < i < k$,
        we order $B$ after the first (complete) instance $A \in S_z$ of \textit{m.minWrite} that performs \textit{switch.minWrite(z)}, for some $z < i$.  
        \item If there is no ordering between two instances implied by the above rules, 
        then we order them
        according to the order in which they perform \textit{switch.minRead()}.
    \end{enumerate}
    
    Now, we prove that if an instance $A$ finished before an instance $B$ started in $e$, 
    then $A$ precedes $B$ in $o$.
    First, suppose $A \in S_i$ for some $0 \leq i < k$.
    Then, either $A$'s instance of \textit{switch.minRead()} returned $i$ or $A$ performed \textit{switch.minWrite(i)}.
    Therefore, by the linearizability of \textit{switch}, $B$'s instance of \textit{switch.minRead()} output at most $i$.
    So $B$ cannot be in $S_h$ for $h > i \ge 0$.
    \begin{itemize}
    \item If $B \in S_i$, then by rule 2, $A$ precedes $B$ in $S_0$ since $A$'s access to $\mathit{T_i}$ precedes $B$'s.
    \item If $B \in S_h$ for $0 \le h < i$, then by rule 1, $A$ precedes $B$ in $o$.
    \item If $B \in S_{<}$ and rule 3 applies, $A$ precedes $B$ in $o$. 
    \item If $B \in S_{<}$ and rule 3 does not apply, then, by rule 4,  $A$ precedes $B$ in $o$, since $A$'s instance of \textit{switch.minRead()} precedes $B$'s.
    \end{itemize}

    Now, suppose $A$ is an instance of $\mathit{m.minWrite(iY + j)}$ in $\mathit{S_{<}}$, for some $0 \le j < Y$.
    Then rules 1 and 2 do not apply.
    We show that rule 3 does not apply. Consider the first instance $F$ of \textit{m.minWrite} that performs \textit{switch.minWrite(z)} 
    before \textit{A} starts, for some $z < i$. Observe that $F$ performs its last step before \textit{A} starts.
    By our assumption, $A$ finished before $B$ started and, therefore, $F \neq B$.
    So rule 3 does imply an ordering between $A$ and $B$ in $o$.
    Therefore, by rule 4, $A$ precedes $B$ in $o$.

    Finally, we will prove that every instance \textit{I} in $e$ receives the same response as when
    the instances in $o$ are performed sequentially.
    In particular, if an instance \textit{I} of \textit{m.minRead()} in $e$ returns $t$, we prove the following two properties:
    \begin{enumerate}
        \item If $t < kY - 1$, there was an instance \textit{Z} of \textit{m.minWrite(t)} ordered before \textit{I} in $o$.
        \item Every instance \textit{W} of $\mathit{m.minWrite(t')}$ with $0 \le t' < t$ occurs later than \textit{I} in $o$.
    \end{enumerate}

    Suppose that $t = iY + j$, where $0 \le j < Y$. Then \textit{I}'s instance, $I'$, of \textit{switch.minRead()} returns $i$ and 
    \textit{I}'s instance of $\mathit{T_i.minRead()}$ returns $j$.
\\\\ \indent
    We prove the first property. 
    Assume that $t < kY - 1$.
    If $j < Y-1$, 
    there must have been an instance $Z'$ of $\mathit{T_i.minWrite(j)}$ performed before $I'$,
    which was performed by an instance 
    $Z$ of $\mathit{m.minWrite(iY + j)}$.
    If $j = Y-1$, then $i < k-1$. 
    Therefore, before $I$ performs \textit{switch.minRead()}, there was an instance of
    \textit{switch.minWrite(i)} performed by some instance $Z$ of \textit{m.minWrite(iY + x)}, where $0 \le x < Y$. 
    Instance $Z$ performed its instance $Z'$ of $\mathit{T_i.minWrite(x)}$
    before \textit{I} performed \textit{I'}.
    Since 
    \textit{I'} returns $Y-1$, 
    and $T_i$ is a $Y$-bounded min register,
    $x$ must be equal to $Y-1$. 
    In either case, $Z$ is an instance of $\mathit{m.minWrite(t)}$ in $S_i$.
    Since $I \in S_i$ and $\mathit{Z'}$ precedes $\mathit{I'}$, rule 2 implies that $Z$ precedes $I$ in $o$.

    Now, we prove the second property. 
    Consider an \textit{W} of \textit{m.minWrite(t')}, where $0 \leq t' < t$.
    Then $t' = i'Y + j'$, where $0 \le i' \le i$.
    
    First, suppose that \textit{W} is in $S_{i'}$.
    If $i' < i$, then rule 1 implies that \textit{I} precedes \textit{W} in $o$.
    If $i' = i$, then $j' < j$. Since 
    $I'$ outputs $j$, $W$ must have performed $\mathit{T_i.minWrite(j')}$ after $I'$.
    Since both \textit{I} and \textit{W} are in $\mathit{S_i}$, rule 2 implies that \textit{I} precedes \textit{W} in $o$.
    
    Otherwise, \textit{W} is in $S_{<}$.
    By rule 3, \textit{W} is ordered after the first (complete) instance $A$ of \textit{m.minWrite} in $\mathit{S_z}$ that performs \textit{switch.minWrite(z)},
    for some $0 \le z < i'$. Since $z < i$, by rule 1 this implies that \textit{I} precedes \textit{A} in $o$.
    Since \textit{A} precedes \textit{W} in $o$, by transitivity \textit{I} precedes \textit{W} in $o$.
\end{proof}
The resulting min register $m$ is wait-free since instances of \textit{minRead} and \textit{minWrite} 
operation instances perform at most a constant number of instances on bounded min registers that are wait-free.

This construction takes the form of a perfect $k$-ary tree of height $\log_k kY$ in which the root is \textit{switch}
and the $i$-th subtree is $\mathit{T_i}$ for $0 \le i < k$.
The internal nodes and leaves of the tree are $k$-bounded min registers.
Suppose one implements a $(k^x)$-bounded min register $m$ using the techniques described in this section, for 
any non-negative integer $y$.
We will prove a result concerning the step complexity of $m$.

\begin{theorem}
    Let $m$ be a $(k^x)$-bounded min register which is implemented using $k$-ary tree of height $x-1$, as described in this section, for any $x \ge 1$.
    An instance of a \textit{minRead} operation on $m$ performs exactly $x$ steps,
and an instance of a \textit{minWrite} operation on $m$ performs at most $2x-1$ steps.
\end{theorem}
\begin{proof}
As described in Section \ref{smallMinReg}, our implementation of a $k$-bounded min register 
requires exactly 1 step for a \textit{minRead} or \textit{minWrite} operation.
Thus the property is true for a $k$-bounded min register based on a $k$-ary tree of height 0.

Suppose, by induction on $g$ for some $g \ge 0$, that our implementation of a $(k^{g})$-bounded min register based on a $k$-ary tree of height $g-1$
requires exactly $g$ steps for a \textit{minRead}
and at most $2(g)-1$ steps for a \textit{minWrite}.
Then we argue that a $(k^{g+1})$-bounded min register $m$ based on a $k$-ary tree of height $g$ using this implementation requires exactly $g+1$
steps for a \textit{minRead} and at most $2(g+1) - 1$ steps for a \textit{minWrite}.

A \textit{minRead} operation instance on $m$ performs a single step
to perform a \textit{minRead} of \textit{switch},
then it performs a \textit{minRead} of $\mathit{T_i}$ for some $i$.
By our induction hypothesis, since $\mathit{T_i}$ is a $(k^{g})$-bounded min register implemented using our construction, 
exactly $g$ steps are required to perform 
a \textit{minRead} of $\mathit{T_i}$. Thus a \textit{minRead} of \textit{m} requires exactly $g+1$ steps.

A \textit{minWrite} operation instance on $m$ performs at most 
two instances on \textit{switch}
and a \textit{minWrite} instance once $\mathit{T_i}$. 
A \textit{minRead} or \textit{minWrite} of \textit{switch} ends after exactly one step.
By our induction hypothesis, since $\mathit{T_i}$ is a $(k^{g})$-bounded min register implemented using our construction,
at most $2(g)-1$ steps are required to perform a \textit{minWrite} of $\mathit{T_i}$.
Thus a \textit{minWrite} of \textit{m} requires at most $2 + 2(g) - 1 = 2(g+1) - 1$ steps.
\end{proof}

\chapter{An Implementation of Ko's Lock-Free Binary Trie}
\label{trieImplementationChapter}
In this chapter, we describe our implementation of Ko's lock-free binary trie\cite{jeremy_ko_trie,DBLP:journals/corr/abs-2405-06208}.
Ko's data structure is a lock-free implementation of a dynamic ordered set.
The worst-case step complexity of searches on his data structure is $O(1)$,
whereas the worst-case amortized step complexity of insert, search and predecessor
is $O(\log u + c^2)$, where $u$ is the number of keys in the universe and 
$c$ is a measure of the contention.

The first section describes our implementation of Ko's relaxed binary trie,
which is a major component of the lock-free binary trie.
The second section introduces the components used in our implementation of Ko's lock-free binary trie 
and concerns how its algorithms extend those of the relaxed binary trie.
The third section discusses the lock-free linked lists used by 
the lock-free binary trie in more detail.
The final section concerns the memory reclamation techniques used in our 
implementation.

The source code of our implementation is available at \url{https://github.com/Jakjm/LockFreeDynamicSetExperiments}\footnote{The relevant code is in 
the DynamicSets/Trie directory along with the files in the DynamicSets directory.}.

\section{Relaxed Binary Trie}
\label{relaxedTrie}
Ko's lock-free binary trie is based on a relaxed binary trie, which at any time stores a subset $S$ of keys from the universe $U = \{0, \dots, 2^k - 1\}$ for some fixed positive integer $k$.
The relaxed binary trie is wait-free. It supports search in $O(1)$ steps and insert, 
remove and \textit{relaxedPredecessor} in $O(k)$ steps.
The \textit{relaxedPredecessor} operation is a non-linearizable relaxation of a predecessor operation.
In particular, an instance of \textit{relaxedPredecessor} may return $\bot$ or an incorrect result if inserts or removes occur while 
the instance 
is in progress. If no concurrent insert or remove occurs, the instance always returns a correct result.
The implementation of the relaxed binary trie makes use of dynamically allocated \textit{Update\-Nodes} which
are composed of several shared objects.

When a process wants to insert a key $x \in U$ into the binary trie or remove $x$, it uses an \textit{Update\-Node} to try to announce this to other processes.
Key $x$ is stored in the \textit{key} field of the Update\-Node.
Each Update\-Node has an immutable \textit{type}, which is either \textit{INS} or \textit{DEL}.
An Update\-Node with \textit{type} \textit{INS} is also called an \textit{InsertNode}
and is used when a process wants to insert a key into the Trie.
An Update\-Node with \textit{type} \textit{DEL} is also called a \textit{DelNode}
and is used when a process wants to remove a key from the Trie.
The \textit{state} field of an Update\-Node is initially \textit{inactive}
and it may only be changed to \textit{active}.
The \textit{latestNext} field of an Update\-Node is a \textit{SWAP} object that stores either $\bot$ 
or a pointer to an Update\-Node with the same \textit{key}, but opposite \textit{type}.
In the implementation, for every key $x \in U$, there is a LatestList which is a linked list of Update\-Nodes
which have the same \textit{key}.
Every LatestList has length at least one and at most two and its last UpdateNode is \textit{active}.
If it has length two then the types of the two UpdateNodes are different.
There is an array, \textit{latest}, of $2^k$ \textit{CAS} objects,
such that $latest[x]$
points to the first Update\-Node in the LatestList for \textit{key} $x$.
The \textit{latestNext} of this Update\-Node stores either $\bot$ or a pointer to the second Update\-Node 
in this LatestList.
We say that the first \textit{active} Update\-Node in the LatestList is the \textit{first active} Update\-Node for key $x$.
For every key $x \in U$, $x$ is in $S$ if and only if the \textit{first active}
Update\-Node for key $x$ is an InsertNode.
If an operation instance, $op$, successfully inserts an Update\-Node into a LatestList, 
we say that this Update\-Node is \textit{owned} by $op$.
In our implementation, a process that fails to insert an Update\-Node into a LatestList 
will reuse the Update\-Node and try to insert it into a LatestList during its next instance of an update operation of the same type. 
The \textit{key} of an Update\-Node may be changed if it is being reused. 
However, once inserted into a LatestList, the Update\-Node's \textit{key} will remain fixed for the rest of the execution.

The relaxed binary trie consists of a perfect binary tree of height $k$,
in which each internal node is a \textit{TrieNode} and the leaves are the elements of $latest$,
in order from left to right.
We often refer to $\mathit{latest[x]}$ as \textit{leaf x}, for $0 \le x < 2^k$.
In the relaxed binary trie, every node has an \textit{interpreted bit}, which replaces the
bit it stores in a sequential binary trie.
The \textit{interpreted bit} of a leaf is 1 if and only if the first active UpdateNode 
in its LatestList is an InsertNode.
We will discuss how the \textit{interpreted bits} of TrieNodes are determined shortly.
The implementations of operations for the relaxed binary trie
are very similar to their implementations
for a sequential binary trie, except that they use the \textit{interpreted bits}.
A relaxedPredecessor operation instance may return an incorrect result or $\bot$ 
when there are concurrent updates by other process to the \textit{interprted bits} of 
nodes that it accesses.




Every \textit{TrieNode} contains a \textit{dNodePtr} field, which is a \textit{CAS} object that stores a pointer to a DelNode in (or was previously in) a LatestList of a leaf in the TrieNode's subtree.
If an InsertNode is the \textit{first active} UpdateNode in this LatestList, the TrieNode has \textit{interpreted bit} 1.
Every DelNode contains two additional fields, \textit{upper0Boundary} and \textit{lower1Boundary},
which are used to help determine the \textit{interpreted bits} of TrieNodes.
The \textit{upper0Boundary} field of a DelNode is a register storing an integer between 0 and $k$, 
and is non-decreasing throughout the execution.
If the \textit{upper0Boundary} field of a DelNode is changed, this is done by the operation instance 
that owns the DelNode.
The \textit{lower1Boundary} field of a DelNode is a $(k+2)$-bounded min register,
initially storing the value $k + 1$.
This field is only modified by insert operation instances performed on the 
relaxed binary trie. 
Section \ref{smallMinReg} describes the implementation of a bounded min-register that 
we use.
Consider a DelNode, $dNode$, that is the \textit{first active} UpdateNode in some LatestList.
Every TrieNode that points to $dNode$ (or a DelNode with the same key) and has height at least $dNode.lower1Boundary$ or greater than $dNode.upper0Boundary$,
has \textit{interpreted bit} 1.
Otherwise, the TrieNode has \textit{interpreted bit} 0. 
In our implementation, the TrieNodes are stored in an array 
of length $2^k - 1$, such that the $i$-th TrieNode from the left 
at depth $d$ is at index $2^d + i$, for $0 \le d < k$ and $0 \le i < 2^d$.

In the initial configuration, every LatestList in the relaxed trie stores a single DelNode whose
$upper0Boundary$ is $k$, so it represents the empty set.
Since every TrieNode in the relaxed trie points to a DelNode in the LatestList of one of the leaves in its subtree, 
the \textit{interpreted bit} of every TrieNode is 0.

As we will discuss later in more detail, any process, $p$, that is inserting key $x$ into the relaxed binary trie will first 
try to set the \textit{interpreted bit} of leaf $x$ to 1.
If successful, while the \textit{interpreted bit} of leaf $x$ remains equal to 1, 
$p$ will visit every ancestor of the leaf in the relaxed binary trie from lowest 
to highest and set the \textit{interpreted bit} of this ancestor to 1.
When process $p$ is removing key $x$ from the relaxed trie, it will 
first try to set the \textit{interpreted bit} of leaf $x$ to 0.
If successful, while the \textit{interpreted bit} of leaf $x$ remains equal to 0, 
$p$ will visit the ancestors of leaf $x$ in the relaxed binary trie from lowest to highest.
For every ancestor $p$ visits whose children have their \textit{interpreted bits} equal to 0, it will try to set the \textit{interpreted bit} of this ancestor to 0.

Every DelNode contains an additional field, \textit{stop}, which is initially \textit{False}
but may be updated to \textit{True} by a remove operation instance.
When this happens, the remove operation instance which owns the 
DelNode will stop trying to set the \textit{interpreted bits} of TrieNodes to 0.
Every InsertNode contains two additional fields, \textit{target} and \textit{targetKey}, store a pointer to a DelNode and the key of this DelNode, respectively.
Initially \textit{target} is equal to $\bot$ whereas \textit{targetKey} is 
initially equal to $-1$.
Consider an insert operation instance, \textit{iOp}, that owns an InsertNode, 
\textit{iNode}.
Suppose \textit{iOp} traverses the relaxed binary trie and encounters a TrieNode 
that points to a DelNode with key $k$, and that \textit{dNode} is a DelNode that is the \textit{first active} UpdateNode for $k$.
In this case, \textit{iOp} will set \textit{iNode.targetKey} to \textit{k} 
and \textit{iNode.target} to \textit{dNode}.
Suppose a remove instance, \textit{dOp}, inserts a DelNode ahead of $iNode$ in the LatestList for $iNode.key$ and 
then reads a pointer to $dNode$ from $iNode.target$ and $k$ from $iNode.targetKey$.
If $dNode$ is still the \textit{first active} UpdateNode for $k$, \textit{dOp} will set  
\textit{dNode.stop} to \textit{True}.
It is possible that the operation instance that owns \textit{dNode}, $dOp'$,
is still in progress and trying to set the \textit{interpreted bits} of 
TrieNodes that are ancestors of leaf $k$ to 0.
By setting \textit{dNode.stop} to \textit{True}, \textit{dOp} informs \textit{dOp'}
that it should stop updating these \textit{interpreted bits} to 1.

To search for a key, $x$, in the relaxed binary trie, a process 
simply returns whether the \textit{first active} UpdateNode 
for key $x$ is an InsertNode.

To insert a key, $x$, into the relaxed binary trie, a process, \textit{p}, will first read a pointer to the \textit{first active}
UpdateNode, \textit{uNode}, in the LatestList for key $x$.
If \textit{uNode} is an InsertNode, then $p$ will return \textit{False} since $x$ is already in the relaxed binary trie.
If \textit{uNode} is a DelNode, then $p$ will try to insert an InsertNode, $iNode$, into into the LatestList 
for key $x$ before \textit{uNode}. 
First, $p$ sets $iNode.latestNext$ to point to \textit{uNode} and then it uses a \textit{CAS} to 
attempt to update $latest[x]$ to point to $iNode$.
If this \textit{CAS} is unsuccessful because another process inserted an 
InsertNode ahead of \textit{uNode} into the LatestList, $p$ will update this 
InsertNode to \textit{active} before returning \textit{False}.
If the \textit{CAS} is successful, $p$ first sets $iNode.state$ to $Active$ and $iNode.latestNext$ to $\bot$.
Next, $p$ will try to set the \textit{interpreted bits} of TrieNodes on the path to this 
LatestList to 1.
For every TrieNode, $tNode$ on the path from the LatestList for key $x$ to the root, 
$p$ will read the \textit{key}, $z$, of the DelNode pointed to by $tNode$.
Suppose the \textit{first active} UpdateNode for $z$ is a DelNode, $dNode$, 
whose $upper0Boundary$ 
is at least $tNode.height$ and whose $lower1Boundary$ is greater than $tNode.height$.
This means the \textit{interpreted bit} of $tNode$ is 0 and that $p$ needs to 
try to set it to 1.
The \textit{targetKey} of $iNode$ is set to $z$ and its \textit{target} field is updated to point to $dNode$.
Remove operation instances that read $z$ from \textit{targetKey} and \textit{dNode} from \textit{target} can set \textit{dNode.stop} to \textit{True}, 
informing the operation instance that owns \textit{dNode} that it should stop setting 
the \textit{interpreted bits} of TrieNodes to 0.
If $iNode$ is still the \textit{first active} UpdateNode for key $x$, 
$p$ performs $minWrite(t.height)$ on $dNode$'s $lower1Boundary$, 
setting the \textit{interpreted bit} of $tNode$ to 1.
If $iNode$ is no longer the \textit{first active} UpdateNode for key $x$, 
key $x$ has been removed.
Therefore, $p$ should no longer set the \textit{interpreted bits} of ancestors of leaf $x$ to 1,
so $p$ returns \textit{True}, indicating that it successfully inserted key $x$.
Otherwise, once $p$ has visited all ancestors of leaf $x$ and set their 
\textit{interpreted bits} to 1, it returns \textit{True}.


The steps to remove a key, $x$, from the relaxed trie are initially analogous to 
the steps to insert $x$.
First, a process, \textit{p}, attempts to insert a DelNode, $dNode$, with key $x$ into the LatestList for key $x$, provided an InsertNode, $iNode$, was the \textit{first active} UpdateNode for key $x$ (returning \textit{False} otherwise).
If $p$ fails to insert $dNode$ because another process inserted a DelNode
into this LatestList ahead of $iNode$, $p$ will help this DelNode become \textit{active} before returning \textit{False}.
After successfully inserting $dNode$, if $iNode.target \neq \bot$ and $iNode.target$ is the \textit{first active}
UpdateNode for $iNode.targetKey$, $p$ sets $iNode.target.stop$ to $True$.
This is done to help inform the remove operation instance that owns the DelNode pointed to 
by \textit{iNode.target} that it should no longer set the \textit{interpreted bits} of 
the ancestors of leaf $x$ to 0.
Following this, the \textit{interpreted bits} of TrieNodes leading to the LatestList for key $x$ may need to be set to 0,
so $p$ will start visiting the TrieNodes on the path from this LatestList to the root.
For every TrieNode, $tNode$, on this path, $p$ will first check if:
\begin{itemize}
    \item one or both of $tNode$'s children has \textit{interpreted bit} 1,
    \item $dNode.stop = True$,
    \item $dNode$ is no longer the \textit{first active} node for key $x$, or
    \item $dNode.lower1Boundary \neq k + 1$.
\end{itemize}
If so, either there is a node in $tNode$'s subtrie whose \textit{interpreted bit} is 1
or this was the case at some point since $dNode$ was the \textit{first active} Update\-Node for for $x$.
In this case,  $p$ stops setting the \textit{interpreted bits} of $tNode$ and its ancestors to 0 and returns \textit{True}.
If not, $p$ will next perform a \textit{CAS} to try to update $tNode$ 
to point to $dNode$.
If this \textit{CAS} fails, $p$ will check those conditions 
again and try another \textit{CAS}, returning upon a second failed \textit{CAS}.
Upon a successful \textit{CAS} to update $tNode$ to point to $dNode$, $p$ will set $dNode.upper0Boundary$ to the height of $tNode$, provided 
the \textit{interpreted bits} of $tNode$'s children remain equal to 0.
Once $p$ visits all of the TrieNodes on the path to \textit{latest[x]}
it returns \textit{True}.

To perform a \textit{relaxedPredecessor} instance for key $x$, a process 
performs the algorithm for a sequential predecessor operation on the relaxed
binary trie, using the \textit{interpreted bits}. 
If the process manages to complete the sequential predecessor algorithm 
using the \textit{interpreted bits}, it returns the result of this sequential algorithm.
However, it may be impossible for the process to complete the sequential 
predecessor algorithm.
Suppose the process finds an ancestor of leaf $x$ whose 
left child, $L$, has \textit{interpreted bit} 1 and is not an ancestor of leaf $x$.
The process will traverse the right most path of nodes whose \textit{interpreted bits} 
are 1 in $L$'s subtrie.
While traversing down from $L$, the process may encounter a node whose children
both have \textit{interpreted bit} 0. 
If this occurs, instances of remove which overlap with the \textit{relaxedPredecessor} instance must have turned off their \textit{interpreted bits}.
In this case, the process will return $\bot$ since it did not manage to find 
the \textit{predecessor} of $x$.

\section{Lock-Free Binary Trie}
\label{lockFreeTrie}
Now we discuss our implementation of Ko's lock-free binary trie.
The lock-free Trie extends the relaxed binary trie and is a linearizable 
implementation of a dynamic ordered set with $U = \{0, \dots, 2^k - 1\}$.
The implementation of search is identical to its implementation in the relaxed 
binary trie.
The implementations of insert, remove and predecessor for the lock-free binary trie 
add additional steps to their implementations in the relaxed binary trie.

The lock-free binary trie has two linked lists called the \textit{UALL} (Update Announcement Linked List) and \textit{RUALL} (Reverse Update Announcement Linked List). 
As their names suggest, these lists are used to announce 
update operation instances that are being performed on the Trie, by storing 
UpdateNodes owned by those instances.
We discuss the implementations of the \textit{UALL} and \textit{RUALL} in Section \ref{uallImplementation} and Section \ref{ruallSection}.
When an UpdateNode, \textit{uNode}, is inserted into the LatestList by an operation instance, $uOp$,
$uOp$ will insert it into both the \textit{UALL} and \textit{RUALL} before setting \textit{uNode.state} to \textit{active} and \textit{uNode.latestNext} to $\bot$. 
Suppose an operation instance, $hOp$, of the same key and of the same type as 
$uOp$ tries to insert an UpdateNode into the LatestList ahead of an UpdateNode,
\textit{prevNode}, but it fails because \textit{uNode} was inserted first.
If \textit{uNode.latestNext} does not point to \textit{prevNode}, 
$hOp$ returns \textit{False}.
If \textit{uNode.latestNext} points to \textit{prevNode},
$hOp$ will help \textit{uNode} become \textit{active} before returning \textit{False}.
In particular, $hOp$ inserts \textit{uNode} into the \textit{UALL} and \textit{RUALL},
sets \textit{uNode.state} to \textit{active} and \textit{uNode.latestNext} to $\bot$ before returning.
Before $uOp$ finishes, it removes \textit{uNode} from the \textit{UALL} and \textit{RUALL}.
Update\-Nodes have other fields used in the algorithms for the \textit{UALL} and \textit{RUALL}, which are discussed in 
the sections concerning them.
A process performing a predecessor instance will traverse the \textit{UALL} and \textit{RUALL} to collect information about ongoing updates to the Trie.
This information helps the process determine the correct predecessor of its key.

When a process, $p$, performs a predecessor instance on the Trie for key $x$, it creates a dynamically allocated \textit{PredecessorNode}, \textit{pNode},
which stores $x$ in its immutable \textit{key} field.
To announce its operation instance, $p$ inserts \textit{pNode} into a linked list 
called the \textit{PALL} (Predecessor Announcement Linked List),
which stores PredecessorNodes of predecessor instances that 
are in progress.
The implementation of the \textit{PALL} is discussed in more detail 
in Section \ref{pallImplementation}.
PredecessorNodes have additional fields used in the algorithms for the \textit{PALL},
which is discussed in that section.
An operation instance which creates a PredecessorNode will remove 
the PredecessorNode from the \textit{PALL} before it ends.
Every PredecessorNode contains a \textit{ruallPosition} field.
When an operation instance, $pOp$, traverses the \textit{RUALL}, the \textit{ruallPosition}
of a PredecessorNode created by $pOp$ stores a pointer to the UpdateNode in the \textit{RUALL} that $pOp$ is currently visiting.
We discuss the \textit{ruallPosition} field in more detail in Section \ref{ruallSection}.

Ko's Trie also makes use of dynamically allocated \textit{NotifyNodes}.
A NotifyNode is used to notify a predecessor operation instance 
about an update operation instance which inserted or removed a key from the Trie.
Every NotifyNode contains a pointer, \textit{updateNode}, to the 
UpdateNode owned by this update operation instance.
The NotifyNode also stores the \textit{key} of this UpdateNode.
The \textit{updateNodeMax} field of a NotifyNode stores either $-1$ or 
a key in the universe.
In addition, a NotifyNode has a \textit{notifyThreshold} field, 
which is either $\infty$, a key in the universe, or $-\infty$.

Every PredecessorNode contains a \textit{notifyList}, which is a linked list 
of NotifyNodes.
NotifyNodes that are in the \textit{notifyList} of a PredecessorNode are 
never removed from this \textit{notifyList}.
Consider some update operation instance, $uOp$, which owns an UpdateNode, 
\textit{uNode}.
After updating \textit{uNode} to \textit{active} and finishing to update 
\textit{interpreted bits} in the relaxed binary trie, 
$uOp$ will notify predecessor instances about itself, provided \textit{uNode} is still the \textit{first active} UpdateNode for its \textit{key}.
It does this by inserting NotifyNodes pointing to \textit{uNode} into the \textit{notifyLists} of the PredecessorNodes 
belonging to these predecessor instances.
First, $uOp$ will traverse the \textit{UALL}.
It creates a local list, \textit{iKeys}, that contains the \textit{key} of every InsertNode 
that \textit{uOp} encounters in the \textit{UALL} that is the \textit{first active} UpdateNode for its \textit{key}.
The keys in \textit{iKeys} are stored in ascending order.
Next, $uOp$ will traverse the \textit{PALL} while \textit{uNode} remains the \textit{first active} UpdateNode for its \textit{key}.
For every PredecessorNode, \textit{pNode}, that $uOp$ encounters in the \textit{PALL},
$uOp$ will create a NotifyNode, \textit{nNode}, such that \textit{nNode.updateNode} points to \textit{uNode} and \textit{nNode.key} is equal to \textit{uNode.key}.
Now, $uOp$ performs a binary search to find the largest \textit{key} 
less than \textit{pNode.key} in \textit{iKeys}.
If such a key exists, \textit{nNode.updateNodeMax} is set to it, 
otherwise \textit{nNode.updateNodeMax} is set to -1.
In addition, $uOp$ reads the pointer to an UpdateNode in \textit{pNode.ruallPosition} and sets \textit{nNode.notifyThreshold} 
to the \textit{key} of this Update\-Node.
Finally, $uOp$ will attempt to insert \textit{nNode} as the first node in  \textit{pNode}'s \textit{notifyList} using a \textit{CAS}.
If this \textit{CAS} fails and \textit{uNode} is still the \textit{first active}
UpdateNode for its \textit{key}, $uOp$ will try again until it is either successful 
or \textit{uNode} is no longer the \textit{first active} UpdateNode for its \textit{key}.
In our implementation, if a process that fails to insert a NotifyNode
into a \textit{notifyList} because its UpdateNode is no longer the \textit{first active} UpdateNode for its \textit{key}, the process may reuse the NotifyNode during a 
later update operation instance.
The fields of a NotifyNode may be changed if it is being reused, however
once inserted into a \textit{notifyList}, the fields of a NotifyNode become immutable.
A predecessor instance can use the information in the \textit{notifyList} of 
its PredecessorNode to help it determine the correct predecessor for a key.

Consider an InsertNode, $iNode$, that is owned by an insert operation instance, $iOp$.
Immediately before a remove operation instance, $rOp$, attempts to insert a DelNode ahead of 
$iNode$ in the LatestList for $iNode.key$, $rOp$ will check if $iNode$ has been removed from the \textit{UALL}.
If not, it is possible that $iOp$ has not finished informing operation instances about $iNode$.
So, in this case, $rOp$ will notify predecessor operation instances about $iNode$ 
as though it was $iOp$, according to the steps we described in the previous paragraph.
In particular, first $rOp$ will traverse the \textit{UALL}, storing the keys of InsertNodes it encounters.
Then, while $iNode$ remains the \textit{first active} UpdateNode for its \textit{key}, 
$rOp$ will traverse the \textit{PALL} and try to insert a NotifyNode into the \textit{notifyList}
of every PredecessorNode it encounters.

A process performing a \textit{remove(x)} operation instance, \textit{dOp}, on the Trie may perform up to two \textit{predecessor(x)} operations, called \textit{embedded predecessor} operations.
Every DelNode contains additional fields, \textit{delPredNode}, \textit{delPred} and \textit{delPred2}, which contain information about these embedded predecessor instances.
Before \textit{dOp} inserts a DelNode, $dNode$, into the LatestList for key $x$,
\textit{dOp} will perform its first \textit{embedded predecessor} instance, $pOp$.
The key returned by $pOp$ is stored in \textit{dNode.delPred} and the 
a pointer to the PredecessorNode that $pOp$ created is stored in \textit{dNode.delPredNode}.
Immediately after $dNode$ is inserted into the LatestList and becomes \textit{active}, 
\textit{dOp} performs a second \textit{embedded predecessor} instance, \textit{pOp'}.
The key returned by \textit{pOp'} is stored in \textit{dNode.delPred2}.

Now that we have described how searches, inserts and removals work 
in the lock-free binary trie, we will describe how predecessor operations work.

A \textit{predecessor(x)} instance, $pOp$, starts by creating a PredecessorNode, \textit{pNode}, with \textit{key} $x$ and inserting \textit{pNode} into the \textit{PALL}.
Next, $pOp$ traverses \textit{PALL} starting from \textit{pNode},
storing pointers to PredecessorNodes that it encounters.
Then $pOp$ traverses the \textit{RUALL}. 
A local set called $\mathit{I_{ruall}}$ stores pointers to InsertNodes that 
$pOp$ encounters that are the \textit{first active} UpdateNodes for their \textit{keys}.
Another local set, $\mathit{D_{ruall}}$, stores pointers to DelNodes that $pOp$ encounters that are the \textit{first active} UpdateNodes for their \textit{keys}.
Now, $pOp$ performs a \textit{relaxedPredecessor(x)} instance.
Let $r$ be the result of this instance.

Following this, $pOp$ traverses the \textit{UALL}.
A local variable called $\mathit{iMax_{uall}}$ stores the maximum \textit{key} 
less than $x$ of an InsertNode encountered by $pOp$ which is also the \textit{first active} UpdateNode for its \textit{key}.
If $pOp$ does not encounter such an InsertNode while traversing the \textit{UALL}, 
$\mathit{iMax_{uall}}$ has value $-1$, otherwise $\mathit{iMax_{uall}}$
is a key that was in the data structure during $pOp$.
Another local variable, $\mathit{dMax_{uall}}$, stores the maximum \textit{key}
less than $x$ of a DelNode, $dNode$, such that:
\begin{itemize}
    \item $pOp$ encounters $dNode$ while traversing the \textit{UALL},
    \item $dNode$ is not in $\mathit{D_{ruall}}$, and
    \item $dNode$ was the \textit{first active} UpdateNode for its key.
\end{itemize}
If $pOp$ does not encounter such a DelNode while traversing the \textit{UALL},
$\mathit{dMax_{uall}}$ has value $-1$, otherwise $\mathit{dMax_{uall}}$
is a key that was in the data structure during $pOp$.

Subsequently, $pOp$ will traverse \textit{pNode}'s \textit{notifyList}.
Consider the set, \textit{pNotify}, of NotifyNodes with key less than $x$ that $pOp$ encounters.
The local variable $\mathit{iMax_{notify}}$ holds the maximum among -1 and:
\begin{itemize}
    \item the \textit{keys} of every NotifyNode, \textit{nNode}, that is in \textit{pNotify} 
    such that $\mathit{nNode.notifyThreshold} \le \mathit{nNode.key} < \mathit{x}$ and 
    \textit{nNode.updateNode} points to an InsertNode.
    \item the \textit{updateNodeMax} of every NotifyNode, $nNode'$, that is 
    in \textit{pNotify} such that such that $\mathit{nNode'.notifyThreshold = -\infty}$, $\mathit{nNode'.key} < \mathit{x}$ and 
    $\mathit{nNode'.updateNode}$ is not in $\mathit{I_{ruall}}$ or $\mathit{D_{ruall}}$.
\end{itemize}
The local variable $\mathit{dMax_{notify}}$ holds the maximum among -1 and
the \textit{keys} of every NotifyNode, \textit{nNode}, that is in \textit{pNotify}
such that $nNode.notifyThreshold < nNode.key < x$ and $nNode.updateNode$ points to a DelNode.

Suppose either $\mathit{D_{ruall}}$ is empty or the \textit{relaxedPredecessor(x)} instance by $pOp$ produces a result, that is, $r \neq \bot$.
In this case, $pOp$ returns the maximum \textit{key} among $r$,   
$\mathit{iMax_{uall}}$, $\mathit{dMax_{uall}}$, $\mathit{iMax_{notify}}$
and $\mathit{dMax_{notify}}$.

Suppose $\mathit{D_{ruall}}$ is not empty and $r = \bot$.
First, $pOp$ checks if there is a DelNode in $D_{ruall}$ whose \textit{delPredNode} 
points to a PredecessorNode that $pOp$ encountered while traversing the \textit{PALL}.
If so, $pOp$ will traverse the \textit{notifyList} of the latest such PredecessorNode it encountered in the \textit{PALL}.
For each NotifyNode with \textit{key} less than $x$ encountered, 
$pOp$ inserts a pointer to the UpdateNode pointed to by the NotifyNode into a local list, $L_1$.
Next, $pOp$ traverses the \textit{notifyList} of its own PredecessorNode, \textit{pNode}.
For each NotifyNode with \textit{key} less than $x$ encountered,
$pOp$ inserts the UpdateNode it points to into a local set, $L_{rem}$.
Moreover, if the NotifyNode's \textit{notifyThreshold} is at least its \textit{key}
the UpdateNode it points to is inserted into a local list, $L_2$.
A local list, $L'$, consists of the elements in $L_1$ that are not in $L_{rem}$
in reverse order, followed by the elements in $L_2$ in reverse order.
The local list, $L$, consists of the elements in $L'$ without any DelNode 
that is not the last UpdateNode in $L'$ with its \textit{key}.
Let $G$ be a directed graph on the set of integers, such that there is a 
directed edge between two vertices, $u$ and $v$, if there is 
some DelNode in $L$ whose \textit{key} is $u$ and whose \textit{delNode2} is $v$.
Let $V'$ be the set of \textit{delPreds} of DelNodes in $D_{ruall}$
and \textit{keys} of InsertNodes in $L$.
Let $R$ be the set containing every sink in $G$ which is reachable from a \textit{key} in $V'$ and is not the \textit{key} of some DelNode in $D_{ruall}$.
Let $r'$ be the maximum element in $R$.
Finally, $pOp$ returns the maximum \textit{key} among $r'$,   
$\mathit{iMax_{uall}}$, $\mathit{dMax_{uall}}$, $\mathit{iMax_{notify}}$
and $\mathit{dMax_{notify}}$.

Ko provides a proof that his algorithm is linearizable 
and provides bounds on the worst-case amortized step complexity of 
operation instances on the Trie.
The worst-case step complexity of search instances on the Trie is $O(1)$.
The following definitions are necessary to explain the step complexity bounds for 
insert, delete and predecessor instances.
Let $op$ be some operation instance.
The \textit{point contention} of $op$, $\dot{c}(op)$, is the 
maximum number of processes that were performing an operation 
instance during a configuration in which $op$ is in progress.
The \textit{interval contention} of $op$, $\overline{c}(op)$,
is the number of operation instances that overlap $op$.
The \textit{overlapping interval contention} of $op$, $\widetilde{c}(op)$,
is the maximum interval contention of all operation instances that overlap
$op$.
Ko proves that the worst-case amortized step complexity of an insert instance, $iOp$, 
is $O(k + \dot{c}(iOp)^2)$ whereas the worst-case amortized step complextiy 
of a predecessor or remove instance, $rOp$, is $O(k + \dot{c}(rOp)^2 +\widetilde{c}(rOp)$).

As previously mentioned, the following section discusses the implementation 
of the \textit{UALL}, \textit{RUALL} and \textit{PALL}, as well as 
the \textit{ruallPosition} field of PredecessorNodes.
In Section \ref{trieMemoryReclamation}, we introduce how UpdateNodes, 
PredecessorNode and NotifyNodes are reclaimed by our implementation.

Aside from adding memory reclamation, our implementation has some other 
notable differences from Ko's implementation which we discuss here:

As we will discuss in Section \ref{uallImplementation}, our implementation 
does not permit UpdateNodes that have been inserted and removed from the \textit{UALL} (or \textit{RUALL}) to be reinserted.
By contrast, in Ko's implementation the \textit{UALL} and \textit{RUALL} permit
reinsertions. 
In both of our implementations, an update operation instance 
may help an UpdateNode it does not own become \textit{active} by 
inserting it into the \textit{UALL} and \textit{RUALL}.
However, in Ko's implementation, if the instance that owns 
the UpdateNode has already removed it from the \textit{UALL} and \textit{RUALL},
this will reinsert the UpdateNode into these lists.
In this case, the helping instance must remove the UpdateNode 
from the \textit{UALL} and \textit{RUALL} once it recognizes that 
it erroneously reinserted the UpdateNode.
In our implementation, by contrast, the helping instance 
cannot reinsert the UpdateNode into the \textit{UALL} and \textit{RUALL}.
Therefore, processes helping an UpdateNode become \textit{active} 
never need to remove the UpdateNode from the \textit{UALL} and \textit{RUALL}.

The \textit{updateNodeMax} field of a NotifyNode is a pointer to an UpdateNode 
in Ko's implementation whereas in our implementation the \textit{updateNodeMax} field stores the \textit{key} of this UpdateNode.
This is slightly more efficient since a process reading the \textit{updateNodeMax}
field does not need to do another access to read the \textit{key} of the UpdateNode.

The \textit{targetKey} field of a InsertNode is not in Ko's implementation.
The \textit{targetKey} field is always assigned to the \textit{key} of the DelNode 
pointed to by the \textit{target} field of an InsertNode.
Consider an InsertNode in our implementation such that its \textit{target} 
field points to a DelNode, \textit{dNode}.
Processes can use the \textit{targetKey} field of the InsertNode 
to ensure that \textit{dNode} is the \textit{first active} UpdateNode 
for \textit{targetKey} before setting $\mathit{dNode.stop}$  to \textit{True}.
As we will discuss in our memory reclamation section, an UpdateNode 
is not reclaimed before every process has been quiescent (not performing an operation) since the UpdateNode was removed from the LatestList.
By checking \textit{dNode} was in the LatestList, this ensures \textit{dNode} 
has not yet been reclaimed when the process sets $\mathit{dNode.stop}$ to \textit{True}.

\section{Lock-Free Linked Lists}
\label{lockFreeLinkedLists}
In this section we discuss the specification and implementation of lock-free linked lists
used in our implementation of Ko's lock-free binary trie. 
These include the \textit{UALL}, \textit{RUALL} and \textit{PALL}.

\subsection{UALL}
\label{uallImplementation}
The \textit{UALL} is a linked list that is part of Ko's data structure. 
It contains Update\-Nodes and may contain Insert\-Desc\-Nodes, 
which we discuss later.
Update\-Nodes in the list are stored in non-decreasing order by their keys, 
and Update\-Nodes with the same key in the list are sorted from least recently to most recently inserted. 
Two Update\-Nodes, \textit{head} and \textit{tail}, are used as sentinel nodes.
The first node of the list is always \textit{head}, and it has key $-\infty$,
whereas \textit{tail} is always at the end of the list and it has key $\infty$.
The \textit{UALL} supports the following operations, where \textit{uNode} is an Update\-Node:
\begin{itemize}
    \item \textit{insert(uNode)}: If \textit{uNode} has never been in the \textit{UALL}, 
    \textit{uNode} is inserted into the \textit{UALL} after any Update\-Nodes whose keys are less than or equal to \textit{uNode}'s key, 
    but before any Update\-Node with a larger key. Otherwise this operation has no effect. 
    \item \textit{remove(uNode)}: If \textit{uNode} is in the \textit{UALL}, \textit{uNode} is removed from the \textit{UALL}.
    Otherwise, this operation has no effect. 
    \item \textit{first()}: Returns a pointer to the Update\-Node immediately following \textit{head} in the \textit{UALL}, or $\bot$ if \textit{tail} immediately follows \textit{head}.
    \item \textit{readNextUpdate\-Node(uNode)}: This operation can only be performed if \textit{uNode} has been inserted into the \textit{UALL}.
    If \textit{uNode} is in the \textit{UALL}, this operation returns a pointer to the Update\-Node 
    immediately following \textit{uNode} in the \textit{UALL}, or $\bot$ if \textit{tail} immediately follows \textit{uNode}.
    If \textit{uNode} is not in the \textit{UALL}, this operation returns a pointer to the Update\-Node that 
    was immediately following \textit{uNode} in the \textit{UALL} when \textit{uNode} was removed, or $\bot$ if \textit{tail} was immediately
    following \textit{uNode} when it was removed.
\end{itemize}

In existing implementations of lock-free singly linked lists,
a node is inserted by a single process, without help from other processes.
In contrast, more than one process may attempt to insert the same Update\-Node into the \textit{UALL}.
Allowing more than one process to insert the same node causes problems for existing lock-free linked list algorithms.
For example, consider the following execution involving Fomitchev and Ruppert's implementation of a linked list. 
Two processes, $p$ and $q$, are trying to insert the same Update\-Node, $B$, into an empty linked list
which contains only the sentinel Update\-Nodes, $head$ and $tail$.
\begin{figure}[H]
    \centering
    \begin{tikzpicture}
        \node[vertex](head) at (0,0){head};
        \node[vertex](B) at (2,0.8){$B$};
        \node[vertex](tail) at (4,0){tail};

        \draw[edge](head) to (tail);
    \end{tikzpicture}
\end{figure}

\noindent Suppose $q$ performs its insertion of $B$ until it has updated \textit{B} to point to $tail$. 
However, $q$ goes to sleep just before performing a \textit{CAS} to update $head$.
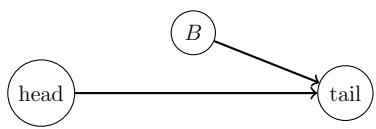
\begin{figure}[H]
    \centering
    \begin{tikzpicture}
        \node[vertex](head) at (0,0){head};
        \node[vertex](B) at (2,0.8){$B$};
        \node[vertex](tail) at (4,0){tail};

        \draw[edge](head) to (tail);
        \draw[edge](B) to (tail);
    \end{tikzpicture}
\end{figure}

\noindent Next, $p$ performs its insertion of $B$, which ends when it performs a 
\textit{CAS} that updates $head$ so that it points to $B$.
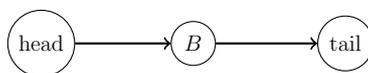
\begin{figure}[H]
    \centering
    \begin{tikzpicture}
        \node[vertex](head) at (0,0){head};
        \node[vertex](B) at (2,0){$B$};
        \node[vertex](tail) at (4,0){tail};

        \draw[edge](head) to (B);
        \draw[edge](B) to (tail);
    \end{tikzpicture}
\end{figure}
\noindent Then $p$ performs a remove operation to remove $B$ from the list.
It flags $head$ and marks $B$.
In the figures, \textcolor{orange}{F} denotes a flagged Update\-Node and \textcolor{red}{M} denotes a marked Update\-Node.
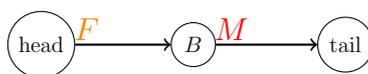
\begin{figure}[H]
    \centering
    \begin{tikzpicture}
        \node[vertex](head) at (0,0){head};
        \node[vertex](B) at (2,0){$B$};
        \node[vertex](tail) at (4,0){tail};
        
        \draw[edge](head) to (B);
        \draw[edge](B) to (tail);
        2.4,0.2
        \node[color=orange](mark) at (0.6,0.2){$F$};
        \node[color=red](mark) at (2.5,0.2){$M$};
    \end{tikzpicture}
\end{figure}

\noindent Now $p$ performs a \textit{CAS} on \textit{head} so that it is unflagged and points to $tail$,
completing the removal of $B$.
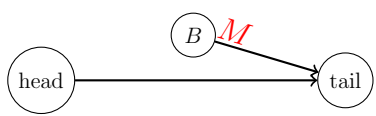
\begin{figure}[H]
        \centering
        \begin{tikzpicture}
            \node[vertex](head) at (0,0){head};
            \node[vertex](B) at (2,0.6){$B$};
            \node[vertex](tail) at (4,0){tail};
    
            \draw[edge](head) to (tail);
            \draw[edge](B) to (tail);
            \node[color=red,rotate=-17](mark) at (2.55,0.65){$M$}; 
        \end{tikzpicture}
\end{figure}

\noindent Then $q$ wakes up and performs its $CAS$, which successfully reinserts $B$ between $head$ and $tail$.
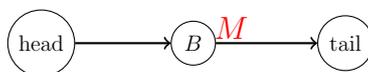
\begin{figure}[H]
    \centering
    \begin{tikzpicture}
        \node[vertex](head) at (0,0){head};
        \node[vertex](B) at (2,0){$B$};
        \node[vertex](tail) at (4,0){tail};

        \draw[edge](head) to (B);
        \draw[edge](B) to (tail);

        \node[color=red](mark) at (2.5,0.2){$M$}; 
    \end{tikzpicture}
\end{figure}
This violates our requirement for the \textit{UALL} that Update\-Nodes should 
not be reinserted following their removal.
Moreover, it also violates an important invariant of Fomitchev and Ruppert's algorithm
that a marked node in the list is immediately preceded by a flagged node.
Following the \textit{CAS} by $q$, $B$ is marked and in the list, but its predecessor, $head$, 
is not flagged.
If an \textit{LL/SC} object is used for the successor field rather than a \textit{CAS}
object, $q$ would be able to see that $B$ has already been inserted, 
and this would circumvent the \textit{ABA problem}.
Our implementation instead makes use of unbounded sequence numbers that are used when inserting an Update\-Node
to resolve this problem.

The \textit{UALL} is a variant of Fomitchev and Ruppert's linked list
which uses \textit{Insert\-Desc\-Nodes} to coordinate updates to an Update\-Node 
which is being inserted and the Update\-Node which will immediately precede it once it is inserted.
When trying to insert an Update\-Node into the \textit{UALL}, 
a process will first insert an Insert\-Desc\-Node into the \textit{UALL} to provide information
to other processes about its ongoing insertion of the Update\-Node.
Any process that encounters the Insert\-Desc\-Node in the \textit{UALL} will use this information 
to help the insert operation. If the Update\-Node has never been inserted, 
it will be inserted in place of the Insert\-Desc\-Node; otherwise, the 
Insert\-Desc\-Node is simply removed.

Every Update\-Node $n$ has a \textit{next} field, a \textit{state} field and a \textit{backlink} field.
If $n$ is in the \textit{UALL}, \textit{n.next} indicates the node immediately following $n$ in the \textit{UALL}.
The \textit{state} of an Update\-Node can either be \textit{Normal}, \textit{InsFlag}, \textit{DelFlag} or \textit{Marked}.
It is initially \textit{Normal}.
If \textit{n.state} is \textit{Normal}, then \textit{n.next} stores either $\bot$ or a pointer to an Update\-Node.
If \textit{n.state} is \textit{InsFlag}, this indicates that $n$ is in the \textit{UALL} and
it is immediately followed by an Insert\-Desc\-Node.
The \textit{DelFlag} \textit{state} has similar semantics to a flag in Fomitchev and Ruppert's 
linked list, whereas the \textit{Marked} \textit{state} has similar semantics to a mark in their list.
If \textit{n.state} is \textit{DelFlag}, this indicates that $n$ is in the \textit{UALL} and
\textit{n.next} stores a pointer to an Update\-Node, $n'$, that is being removed from the \textit{UALL}.
Before $n'$ is removed from the \textit{UALL}, $\mathit{n'.state}$ is set to \textit{Marked} and does not change for the rest of the execution.
We say that an Update\-Node is marked if its \textit{state} is \textit{Marked}; otherwise it is unmarked.
Once $n'$ is marked, $\mathit{n'.next}$ continues to point to the same Update\-Node for the rest of the 
execution.
While a marked Update\-Node is in the \textit{UALL}, it is immediately preceded by an Update\-Node with \textit{state} \textit{DelFlag}.
Before $n'$ is marked, $\mathit{n'.backlink}$ is set to $n$, which allows processes 
to remove $n'$ and resume their traversal of the \textit{UALL} from $n$ once $n'$ is removed.

The \textit{next} and \textit{state} fields of Update\-Nodes are stored inside a single \textit{CAS} object, \textit{NextState}.
An efficient way to implement this on a real machine is to store 
the \textit{NextState} field in a single word. 
On most modern machines, the address of a node object in memory will be a multiple of 8,
leaving the lowest 3 bits unused.
In our implementation, a node's \textit{next} field is stored in the highest 61 bits of a 64-bit word,
whereas the \textit{state} field is stored in the lowest 3 bits.
We use the notation $\twoField{next}{state}$ to denote values which are stored by this \textit{CAS} object.

In addition to helping with insertions, processes operating on the \textit{UALL} may help remove Update\-Nodes from the \textit{UALL} when they encounter an Update\-Node whose \textit{state} is \textit{Marked} or \textit{DelFlag}.
The \textit{helpMarked(prev, delNode)} function is used by a process to attempt to remove a marked Update\-Node, \textit{delNode},
from the \textit{UALL}.
The Update\-Node, \textit{prev}, immediately precedes \textit{delNode} in the 
\textit{UALL} and has \textit{state} \textit{DelFlag}, or there existed 
some time that this was the case.
If \textit{delNode} is still in the \textit{UALL}, it is removed,
\textit{prev.state} is updated from \textit{DelFlag} to \textit{Normal},
and \textit{prev.next} is updated to the Update\-Node that immediately followed \textit{delNode} before it was removed. 
The $next$ and $state$ fields of $prev$ are returned.

To perform an instance of \textit{helpMarked}, a process, \textit{p}, reads a pointer, \textit{delNext}, to the Update\-Node which 
immediately follows (or followed) \textit{delNode} in the \textit{UALL}.
Next, $p$ will perform a \textit{CAS} on \textit{prev.NextState} to attempt to update \textit{prev.next} from \textit{delNode} to \textit{delNext}
and \textit{prev.state} from \textit{DelFlag} to \textit{Normal}.
If successful, this removes the Update\-Node pointed to by $delNode$ from the \textit{UALL}.
Finally, $p$ will return the $next$ and $state$ fields of $prev$. 
\begin{figure}[H] \label{helpMarked}
    \begin{algorithmic}[1]
        \alglinenoNew{alg1}
        \State $helpMarked(prev, delNode)$ \Comment{Helps physically remove $delNode$}
        \Indent
            \State $\twoField{delNext}{state'} \gets delNode.NextState.read()$
            \State $\twoField{next'}{state'} \gets prev.NextState.CAS(\twoField{delNode}{DelFlag}, \twoField{delNext}{Normal})$\label{pruneCAS}
            \If{$\twoField{next'}{state'} = \twoField{delNode}{DelFlag}$} 
                \State \Return $\twoField{delNext}{Normal}$\Comment{CAS successfully removed $delNode$.}\label{helpRemRet1}
            \Else \
                \Return $\twoField{next'}{state'}$\Comment{CAS failed; \textit{delNode} was already removed.} \label{helpRemRet2}
            \EndIf
        \EndIndent 
        \alglinenoPush{alg1}
    \end{algorithmic}
\end{figure}

Instead of allocating and reclaiming Insert\-Desc\-Nodes after each insertion and removal, in our implementation,
every process has its own Insert\-Desc\-Node which it uses to insert Update\-Nodes into the \textit{UALL} throughout the execution.
Arbel-Raviv and Brown\cite{arbelraviv_et_al:LIPIcs.DISC.2017.4} proposed a technique for reusing immutable descriptor nodes
using sequence numbers, which is what we use in our implementation.
Let $p_i$ be the process with id $i$, for any $i$.
There is a shared array \textit{insDesc} which contains the Insert\-Desc\-Node of every process,
such that $p_i$'s Insert\-Desc\-Node is denoted $\mathit{insDesc[i]}$.
Every Insert\-Desc\-Node contains three fields, \textit{new\-Node}, \textit{next} field and \textit{seqNum}.
Any process may read the fields of $\mathit{insDesc[i]}$, but only $p_i$ may change them.
The fields of $\mathit{insDesc[i]}$ are only changed when $\mathit{insDesc[i]}$ is not in the \textit{UALL}.
If $\mathit{insDesc[i]}$ is in the \textit{UALL}, it comes right after an Update\-Node, its \textit{next} field stores a pointer to the Update\-Node that immediately follows  
$\mathit{insDesc[i]}$ in the \textit{UALL} and its \textit{new\-Node} field stores a pointer to an Update\-Node that $p_i$ is trying to insert.
The \textit{seqNum} field stores a non-negative integer which serves as a sequence number, initially holding the value 0.
When $\mathit{insDesc[i]}$ is removed from the \textit{UALL}, $p_i$ increments $\mathit{insDesc[i].seqNum}$,
so $\mathit{insDesc[i]}$ has a unique sequence number every time it is inserted into the \textit{UALL}.
Moreover, until $p_i$ has incremented $\mathit{insDesc[i].seqNum}$, $p_i$ does not change 
$\mathit{insDesc[i].newNode}$ or $\mathit{insDesc[i].next}$.
If $\mathit{insDesc[i]}$ is in the \textit{UALL}, it is immediately preceded by exactly one Update\-Node, $A$, in the \textit{UALL} such that: 
\begin{itemize}
    \item \textit{A.state} is \textit{InsFlag}, and
    \item \textit{A.next} stores two subfields, $seqNum$ and $pid$, 
    such that $seqNum$ is equal to $\mathit{insDesc[i].seqNum}$ and $pid$ is equal to $i$.
\end{itemize}
The contents of the \textit{next} field of an Update\-Node whose \textit{state} is \textit{InsFlag} are denoted $\twoField{seqNum}{pid}$.
As $\mathit{insDesc[i]}$ has a unique \textit{seqNum} every time it is inserted,
the contents of $A.NextState$ will be unique as well.
This ensures that \textit{CAS} steps on $A.NextState$ by processes still trying to help an earlier insertion performed by $p_i$
will not be successful, as $A.NextState$ will be different.

The storage of both sequence numbers and process ids in the \textit{next} field of Update\-Nodes might appear worrisome, 
given that in physical implementations words have a fixed width.
In our implementation of Update\-Nodes on a 64-bit machine, for example, 
the \textit{next} field is stored within the highest 61 bits of a word.
Could the sequence number continue to grow during an execution and exhaust the available space in the \textit{next} field?
While the answer is actually yes, it would take an enormous amount of time.
Assuming there can be at most $65536 = 2^{16}$ processes, 16 bits can be used to store the process id in the \textit{next} field,
leaving 45 bits for the sequence number.
This leaves $2^{45} - 1$ as the largest sequence number available.
Assuming further that each process can perform at most 1 million insertions per second into the \textit{UALL},
it would take an execution lasting over a year for a process to use up all of the available sequence numbers.
This is sufficient for our purposes. 
However, if an application needed to be able to run with more processes, or for longer than a year,
many machines support \textit{DWCAS} (double-width \textit{CAS}) operations.
A \textit{DWCAS} operation is a \textit{CAS} operation that is applied to two consecutive words in memory.
On such a machine, we could instead store the \textit{next} field and \textit{state} fields of an UpdateNode in a 128-bit word. 
Allowing even $2^{32}$ (over 4 billion) distinct process ids, 
there would be 93 bits remaining to store the sequence number. If every process performs at most
at most 1 trillion insert operation instances per second, it would take over 314 million years to exhaust
all of the available sequence numbers. It is hard to imagine that this would not be sufficient.

The \textit{helpInsert(prev, seq, x)} function can be used by a process which encounters an Update\-Node, $prev$, 
such that \textit{prev.state} is \textit{InsFlag} and \textit{prev.next} is $\mathit{\twoField{seq}{x}}$, or this was the case previously.
The process will try to will try to help complete the insert instance being performed by $p_x$.
If $\mathit{insDesc[x]}$ is still in the \textit{UALL} and its \textit{seqNum} is still equal to \textit{seq}, it is removed from the \textit{UALL},
\textit{prev}'s $state$ is updated from \textit{InsFlag} to \textit{Normal},
and if the Update\-Node pointed to by $\mathit{insDesc[x].new\-Node}$ has never been inserted, it is inserted in the place of $\mathit{insDesc[x]}$.
The contents of the $next$ and $state$ fields of $prev$ following the function are returned.

To perform an instance of \textit{helpInsert(prev, seq, x)}, a process, \textit{p}, first reads the \textit{next} field, \textit{new\-Node} and \textit{seqNum} fields of $\mathit{insDesc[x]}$.
Let $new\-Node'$ and $next'$ be equal to the values which $p$ read from the \textit{new\-Node} and \textit{next} fields of $\mathit{insDesc[x]}$. 
If its \textit{seqNum} is not equal to $seq$ then $p$ simply reads and returns the \textit{next} field and \textit{state} fields of \textit{prev}.
If its \textit{seqNum} is equal to $seq$ then $new\-Node'$ is a pointer to an Update\-Node that $p_x$ is (or was) trying to insert into the \textit{UALL}.
A \textit{CAS} is performed by $p$ in an attempt to update $\mathit{new\-Node'.next}$ from $\bot$ to $\mathit{next'}$.
Suppose this \textit{CAS} fails because the \textit{state} of $new\-Node'$ is \textit{Marked}.
Then $new\-Node'$ has already been inserted into the \textit{UALL}.
In this case, $p$ performs a \textit{CAS} to attempt to update \textit{prev.next} from $\twoField{seq}{x}$ to $next'$,
and \textit{prev.state} from \textit{InsFlag} to \textit{Normal}, returning the values of $prev.next$ and $prev.state$ following the \textit{CAS}.
If the \textit{CAS} was successful then $\mathit{insDesc[x]}$ was removed from the \textit{UALL} without inserting $new\-Node'$.
Suppose instead the \textit{state} of $new\-Node'$ was not \textit{Marked}.
Then $p$ performs a \textit{CAS} on \textit{prev.NextState} to attempt to update \textit{prev.next} from $\twoField{seq}{x}$ to $new\-Node$ and 
\textit{prev.state} from \textit{InsFlag} to \textit{Normal}, returning the $next$ and $state$ fields of \textit{prev} following the \textit{CAS}.
If the \textit{CAS} was successful then \textit{new\-Node} has been inserted into the \textit{UALL} into the place of $\mathit{insDesc[x]}$.
\begin{figure}[H] \label{helpInsert}
    \begin{algorithmic}[1]
        \alglinenoPop{alg1}
        \State $helpInsert(prev, seq, x)$
        \Indent
            \State $new\-Node' \gets insDesc[x].new\-Node.read()$
            \State $next' \gets insDesc[x].next.read()$
            \If{$insDesc[x].seqNum.read() \neq seq$} \label{seqNumChangedHI}
                \State \Return $prev.NextState.read()$ \label{returnEarlyHI} \Comment{Insert op by process $x$ already completed.}
            \EndIf 
            \State $\twoField{n'}{state'} \gets new\-Node'.NextState.CAS(\twoField{\bot}{Normal}, \twoField{next'}{Normal})$\label{initCAS}
            \If{$state' = Marked$} \label{checkMarkedHI}
                \State $newNext \gets next'$ \label{noInsertHI} \Comment{Should replace $insDesc[x]$ with next Update\-Node}
            \Else
                \State $newNext \gets new\-Node'$ \label{insertHI} \Comment{Should replace $insDesc[x]$ with $new\-Node'$}
            \EndIf
            \State $\twoField{next'}{state'} \gets prev.NextState.CAS(\twoField{\twoField{seq}{x}}{InsFlag}, \twoField{newNext}{Normal})$\label{helpInsCAS}
            \If{$\twoField{next'}{state'} = \twoField{\twoField{seq}{x}}{InsFlag}$}
                \State \Return $\twoField{newNext}{Normal}$\Comment{CAS succeeded}\label{casSuccHI}
            \Else \ 
                \Return $\twoField{next'}{state'}$ \label{casFailHI}\Comment{CAS failed}
            \EndIf
        \EndIndent
        \alglinenoPush{alg1}
    \end{algorithmic}
\end{figure}

An instance of the \textit{helpRemove(prev, delNode)} function is used by a process when it encounters an Update\-Node, $prev$, 
such that \textit{prev.state} is \textit{DelFlag} and it immediately precedes another UpdateNode, $delNode$, in the \textit{UALL},
or this was the case previously.
First, \textit{delNode.backlink} is set to $prev$. This allows processes that are traversing the \textit{UALL} to efficiently resume their search 
after \textit{delNode} has been removed from the \textit{UALL}.
Next, if unmarked, $delNode$ is marked. 
Finally, if $delNode$ has not already been removed, \textit{prev.next} is updated to point to the Update\-Node which follows 
$delNode$ in the \textit{UALL} and \textit{prev.state} is updated from \textit{DelFlag} to \textit{Normal}.
The contents of the $next$ and $state$ fields of $prev$ following the instance are returned.

To perform an instance of \textit{helpRemove(prev, delNode)}, 
a process, \textit{p}, will first write \textit{prev} into \textit{delNode.backlink}.
Next, $p$ reads the \textit{next} field and \textit{state} fields of 
\textit{delNode} by reading \textit{delNode.NextState}.
While \textit{delNode.state} is not \textit{Marked}, $p$ will perform one of the following:
\begin{itemize}
\item If \textit{delNode.state} is \textit{Normal}, then the process performs a \textit{CAS} on \textit{delNode.NextState} to try to update \textit{delNode.state}
from \textit{Normal} to \textit{Marked}, leaving \textit{delNode.next} unchanged.
\item If \textit{delNode.state} is \textit{DelFlag}, then $p$ uses $helpRemove$ recursively to attempt to help remove the Update\-Node immediately following $delNode$ 
from the \textit{UALL} and update \textit{delNode.state} back to \textit{Normal}.
\item If \textit{delNode.state} is \textit{InsFlag}, then $p$ uses $helpInsert$ to attempt to help insert an Update\-Node immediately following \textit{delNode}
and update \textit{delNode.state} back to \textit{Normal}.
\end{itemize}
Once $delNode$ is marked, \textit{helpMarked(prev, delNode)} is called by $p$ to remove \textit{delNode} from the \textit{UALL}, 
and $p$ returns the $next$ and $state$ fields of $prev$.
\begin{figure}[H] \label{helpRemove}
    \begin{algorithmic}[1]
        \alglinenoPop{alg1}
        \State $helpRemove(prev, delNode)$ \Comment{Helps mark and remove $delNode$}
        \Indent
            \State $delNode.backlink.write(prev)$ \label{setBacklink}
            \State $\twoField{next'}{state'} \gets delNode.NextState.read()$\label{helpRemoveNextInit}
            \While{$state' \neq Marked$}\Comment{Keep trying to mark $delNode$.}\label{helpRemoveLoop}
                \If{$state' = InsFlag$}\Comment{Help with pending insertion.}
                    \State $\twoField{seq}{j} \gets next'$ \Comment{Get $seqNum$ and $pid$ from $next'$.}
                    \State $\twoField{next'}{state'} \gets helpInsert(delNode, seq, j)$\label{helpRemoveHI}
                \ElsIf{$state' = DelFlag$}\Comment{Help with pending deletion.}
                    \State $\twoField{next'}{state'} \gets helpRemove(delNode, succ)$ \label{helpRemoveHR}
                \Else \Comment{$state'$ is $Normal$.}
                    \State $\twoField{oldNext}{state'} \gets delNode.NextState.CAS(\twoField{next'}{Normal}, \twoField{next'}{Marked})$\label{markCAS}
                    \If{$\twoField{oldNext}{state'} = \twoField{next'}{Normal}$} 
                        \Break \Comment{The CAS succeeded in marking $delNode$.}\label{helpRemoveExit}
                    \EndIf
                    \State $next' \gets oldNext$
                \EndIf
            \EndWhile
            \State \Return $helpMarked(prev, delNode)$ \label{helpRemoveHM} \Comment{Return \textit{next} field and \textit{state} of $prev$.}
        \EndIndent 
        \alglinenoPush{alg1}
    \end{algorithmic}
\end{figure}

Suppose $p_i$ wants to insert an Update\-Node, $B$, into the \textit{UALL}.
First, $p_i$ searches for the latest Update\-Node in the \textit{UALL}, $A$, such that $A.key \le B.key$.
If $p_i$ encounters $B$ while searching for $A$ or $p_i$ sees that $B$ is marked, it returns.
While traversing the \textit{UALL}, if $p_i$ encounters an Update\-Node, \textit{uNode},
whose \textit{state} is not Normal, $p_i$ will perform helping:
\begin{itemize}
    \item If \textit{uNode} has \textit{state} \textit{DelFlag}, $p_i$ will use 
    \textit{helpRemove} to help remove the UpdateNode immediately following \textit{uNode} from the \textit{UALL}, restoring \textit{uNode.state} to \textit{Normal}.
    \item If \textit{uNode} has \textit{state} \textit{Marked}, $p_i$ will use 
    \textit{helpMarked} to attempt to remove \textit{uNode} from the \textit{UALL}, 
    before resuming its traversal from \textit{uNode.backlink}.
    \item If \textit{uNode} has \textit{state} \textit{InsFlag}, $p_i$ will use 
    \textit{helpInsert} to help insert an UpdateNode immediately following
    \textit{uNode}, restoring \textit{uNode.state} to \textit{Normal}.
\end{itemize}
If $p_i$ locates $A$, and \textit{A.state} is not \textit{Normal}, then $p_i$ helps with either an insertion or a removal, trying to return \textit{A.state} to \textit{Normal}. 
If $A$ is marked, $p_i$ will attempt to remove $A$ from the \textit{UALL} and resume its search for the latest Update\-Node 
in the \textit{UALL} whose key is less than or equal to $B$'s key.
If \textit{A.state} is \textit{Normal}, then $p_i$ attempts to insert $\mathit{insDesc[i]}$
into the \textit{UALL} between $A$ and the Update\-Node $C$ that immediately follows $A$.
It sets the \textit{new\-Node} of $\mathit{insDesc[i]}$ to point to $B$, the \textit{next} field of $\mathit{insDesc[i]}$ to 
point to $C$. Then, $p_i$ performs a \textit{CAS} to attempt to update \textit{A.state} from \textit{Normal} to \textit{InsFlag} and \textit{A.next} from pointing to \textit{C} to 
containing $\twoField{seq}{i}$, where $seq$ is equal to the \textit{seqNum} of $\mathit{insDesc[i]}$.
If this \textit{CAS} is successful, then $p_i$ performs a \textit{helpInsert(A, seq, i)} instance before incrementing the \textit{seqNum} of $\mathit{insDesc[i]}$ and returning.
If the \textit{CAS} is not successful, then $p_i$ resumes its search for the furthest Update\-Node in the \textit{UALL} whose key is less than or equal to $B$'s key.

\begin{figure}[!ht] \label{insert}
    \begin{algorithmic}[1]
        \alglinenoPop{alg1}
        \State $insert(uNode)$ by $p_i$
        \Indent
            \If{$uNode.NextState.read().state = Marked$} 
                \Return \Comment{\textit{uNode} marked.} \label{insertNodeMarked1} 
            \EndIf
            \State $insDesc[i].new\-Node.write(uNode)$ \label{setDescNewNode}\Comment{Using $insDesc[i]$ to try to insert \textit{uNode}.}
            \State $seq_i \gets insDesc[i].seqNum.read()$ \label{readSeqNum}
            \State $curr \gets head$\label{currInitIns} \Comment{Start searching from \textit{head}.}
            \State $\twoField{next'}{state'} \gets curr.NextState.read()$\label{insertNextInit}
            \Loop \label{insertLoop} 
                \If{$state' = Normal$}
                    \If{$next'.key \le uNode.key$}\label{insertPrecCheck}
                        \If{$next' = uNode$}
                            \Return \Comment{\textit{uNode} already in \textit{UALL}.}
                        \EndIf
                        \State $curr \gets next'$ \label{insertAdvance} \Comment{Continue search from following Update\-Node.}
                        \State $\twoField{next'}{state'} \gets curr.NextState.read()$\label{insertAdvanceUpdate}
                    \Else
                        \If{$uNode.NextState.read().state = Marked$}
                            \State \Return\label{insertNodeMarked2}\Comment{\textit{uNode} marked.}
                        \EndIf
                        \State $insDesc[i].next.write(next')$ \label{setDescNext}
                        \State $\twoField{oldNext}{state'} \gets curr.NextState.CAS(\twoField{next'}{Normal}, \twoField{\twoField{seq_i}{i}}{InsFlag})$\label{insFlagCAS}
                        \If{$\twoField{oldNext}{state'} = \twoField{next'}{Normal}$}
                            \State $helpInsert(curr, seq_i, i)$\label{insertHI1}\Comment{CAS succeeded, finish insertion.}
                            \State $insDesc[i].seqNum.write(seq_i + 1)$\label{incrSeqNum} \Comment{Increment seq num}
                            \State \Return \label{insertNodeInserted} 
                        \EndIf
                        \State $next' \gets oldNext$
                    \EndIf
                \ElsIf{$state' = InsFlag$}
                    \State $\twoField{seq}{j} \gets next'$ \Comment{Get $seqNum$ and $pid$ from $next'$.}
                    \State $\twoField{next'}{state'} \gets helpInsert(curr, seq, j)$\label{insertHI2}\Comment{Help insertion.}
                \ElsIf{$next' = uNode$} 
                    \Return \Comment{\textit{uNode} already in the \textit{UALL}.}
                \ElsIf{$state' = DelFlag$}
                    \State $\twoField{next'}{state'} \gets helpRemove(curr, next')$\label{insertHR} \Comment{Help remove $next'$.}
                \Else \Comment{$curr$ is Marked.}
                    \State $prev \gets curr.backlink.read()$\label{insertBL}
                    \State $\twoField{next'}{state'} \gets prev.NextState.read()$\label{insertBLUpdate}
                    \If{$\twoField{next'}{state'} = \twoField{curr}{DelFlag}$} 
                        \State $\twoField{next'}{state'} \gets helpMarked(prev, curr)$\Comment{Help remove $curr$.}\label{insertHM}
                    \EndIf
                    \State $curr \gets prev$\label{insertCurrSetToPrev}\Comment{Continue searching from $curr$'s backlink.} 
                \EndIf 
            \EndLoop
        \EndIndent
        \alglinenoPush{alg1}
    \end{algorithmic}
\end{figure}

In the following figure, we show how we denote the fields of Update\-Nodes and Insert\-Desc\-Nodes which are relevant to the \textit{UALL}.
On the left is an Update\-Node, and on the right is an Insert\-Desc\-Node.
\begin{figure}[H]
\centering
\begin{tikzpicture}
     \node[updateNode](head) at (0,0) {\nodepart{one} \textit{next}  \nodepart{two} \textit{state} \nodepart{three} \textit{backlink} \nodepart{four} \textit{key} };
     \node[descNode] (bp) at (4,0) {\nodepart{one} \textit{next} \nodepart{two} \textit{new\-Node} \nodepart{three} \textit{seqNum}};
\end{tikzpicture}
\end{figure}
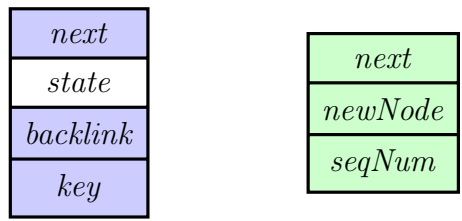
Figure \ref{NewInsProcFig} depicts the insertion of an Update\-Node, $B$, by process $p_i$.
Figure \ref{PreventReinsFig} shows how the reinsertion of an Update\-Node, $B'$, is avoided if $B'$ was previously inserted and removed from the \textit{UALL}.
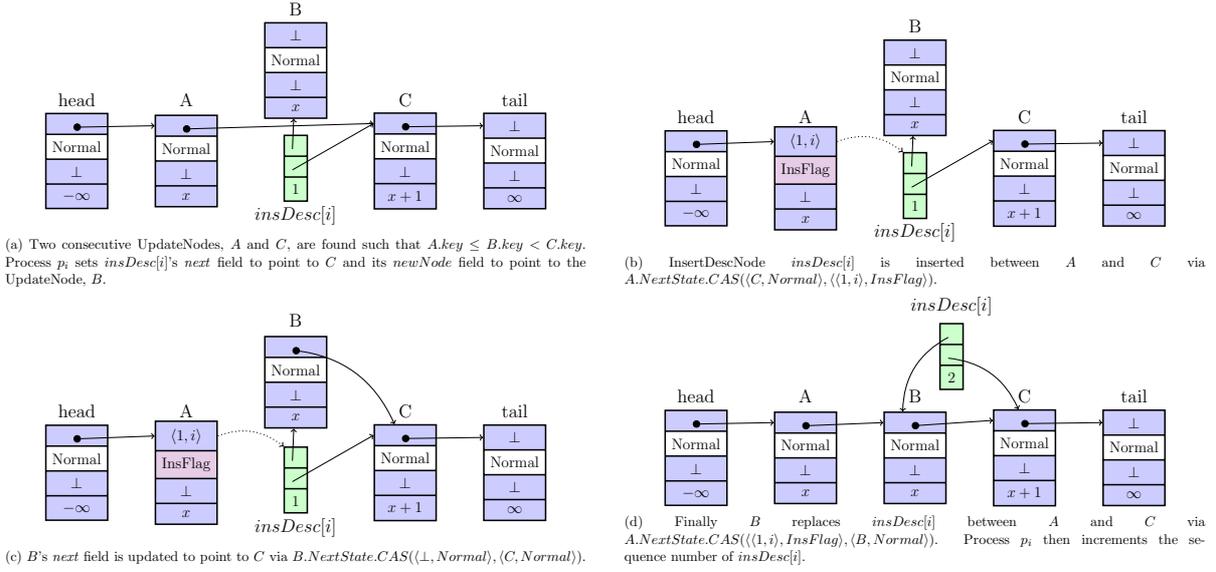
\begin{figure*}[!ht] 
    \scalebox{0.48}{
    \begin{subfigure}[b]{1\columnwidth}
    \centering

    \begin{tikzpicture}

    \node[updateNode](head) at (0,0) {\nodepart{two} Normal \nodepart{three} $\bot$ \nodepart{four} $-\infty$  };
    \node[updateNode] (A) at (3,0) {\nodepart{two} Normal \nodepart{three} $\bot$ \nodepart{four} $x$};
    \node[updateNode] (B) at (6,2.5) {\nodepart{one} $\bot$ \nodepart{two} Normal \nodepart{three} $\bot$ \nodepart{four} $x$};
    \node[descNode] (bp) at (6,-0.2) {\nodepart{three} 1};
    \node[updateNode] (C) at (9,0) {\nodepart{two} Normal \nodepart{three} $\bot$ \nodepart{four} $x+1$};
    \node[updateNode] (tail) at (12,0){\nodepart{one} $\bot$ \nodepart{two} Normal \nodepart{three} $\bot$ \nodepart{four} $\infty$};

    \node at (0,1.7) {\large head};
    \node at (3,1.7) {\large A};
    \node at (6,4.2) {\large B};
    \node at (6,-1.5) {\large $insDesc[i]$};
    \node at (9,1.7) {\large C};
    \node at (12,1.7){\large tail};

    \path[thick, *->] (head.one) edge (A.one west);
    \path[thick, *->] (A.one) edge (C.one west);
    \path[thick, ->] (bp.one) edge  (B);
    \path[thick, ->] (bp.two) edge (C.one west);
    \path[thick, *->] (C.one) edge  (tail.one west);
    \end{tikzpicture}
    \caption{Two consecutive Update\-Nodes, $A$ and $C$, are found such that $A.key \le B.key < C.key$. Process $p_i$
    sets $insDesc[i]$'s \textit{next} field to point to $C$ and its $new\-Node$ field to point to the Update\-Node, $B$.}
    \end{subfigure}
    }
    \hfill
    \scalebox{0.48}{
    \begin{subfigure}[b]{1\columnwidth}
    \centering

    \begin{tikzpicture}

    \node[updateNode](head) at (0,0) {\nodepart{two} Normal \nodepart{three} $\bot$ \nodepart{four} $-\infty$  };
    \node[insFlagNode] (A) at (3,0) {\nodepart{one} $\twoField{1}{i}$ \nodepart{two} InsFlag \nodepart{three} $\bot$ \nodepart{four} $x$};
    \node[updateNode] (B) at (6,2.5) { \nodepart{one} $\bot$ \nodepart{two} Normal \nodepart{three} $\bot$ \nodepart{four} $x$};
    \node[descNode] (bp) at (6,-0.2) { \nodepart{three} 1};
    \node[updateNode] (C) at (9,0) {\nodepart{two} Normal \nodepart{three} $\bot$ \nodepart{four} $x+1$};
    \node[updateNode] (tail) at (12,0){\nodepart{one} $\bot$ \nodepart{two} Normal \nodepart{three} $\bot$ \nodepart{four} $\infty$};

    \node at (0,1.7) {\large head};
    \node at (3,1.7) {\large A};
    \node at (6,-1.5) {\large $insDesc[i]$};
    \node at (6,4.2) {\large B};
    \node at (9,1.7) {\large C};
    \node at (12,1.7){\large tail};

    \path[thick, *->] (head.one) edge (A.one west);
    \path[thick, ->, dotted, bend left] (A.one east) edge (bp.north west);
    \path[thick, ->] (bp.one) edge  (B);
    \path[thick, ->] (bp.two) edge (C.one west);
    \path[thick, *->] (C.one) edge (tail.one west);

    \end{tikzpicture}
    \caption{Insert\-Desc\-Node $insDesc[i]$ is inserted between $A$ and $C$ 
    via $A.NextState.CAS(\twoField{C}{Normal}, \twoField{\twoField{1}{i}}{InsFlag})$.}
    \end{subfigure}
    }
    \scalebox{0.48}{
    \begin{subfigure}[b]{1\columnwidth}
    \centering
    \begin{tikzpicture}
    \node[updateNode](head) at (0,0) {\nodepart{two} Normal \nodepart{three} $\bot$ \nodepart{four} $-\infty$  };
    \node[insFlagNode] (A) at (3,0) {\nodepart{one} $\twoField{1}{i}$ \nodepart{two} InsFlag \nodepart{three} $\bot$ \nodepart{four} $x$};
    \node[updateNode] (B) at (6,2.5) {\nodepart{two} Normal \nodepart{three} $\bot$ \nodepart{four} $x$};
    \node[descNode] (bp) at (6,-0.2) {\nodepart{three} 1};
    \node[updateNode] (C) at (9,0) {\nodepart{two} Normal \nodepart{three} $\bot$ \nodepart{four} $x+1$};
    \node[updateNode] (tail) at (12,0){\nodepart{one} $\bot$ \nodepart{two} Normal \nodepart{three} $\bot$ \nodepart{four} $\infty$};

    \node at (0,1.7) {\large head};
    \node at (3,1.7) {\large A};
    \node at (6,-1.5) {\large $insDesc[i]$};
    \node at (6,4.2) {\large B};
    \node at (9,1.7) {\large C};
    \node at (12,1.7){\large tail};

    \path[thick, *->] (head.one) edge (A.one west);
    \path[thick, ->, dotted, bend left] (A.one east) edge (bp.north west);
    \path[thick, *->, bend left] (B.one) edge (C);
    \path[thick, ->] (bp.one) edge  (B);
    \path[thick, ->] (bp.two) edge (C.one west);
    \path[thick, *->] (C.one) edge  (tail.one west);

    \end{tikzpicture}
    \caption{$B$'s \textit{next} field is updated to point to $C$ via $B.NextState.CAS(\twoField{\bot}{Normal}, \twoField{C}{Normal})$.}
    \end{subfigure}
    }
    \hfill
    \scalebox{0.48}{
    \begin{subfigure}[b]{1\columnwidth}
    \centering

    \begin{tikzpicture}

    \node[updateNode](head) at (0,0) {\nodepart{two} Normal \nodepart{three} $\bot$ \nodepart{four} $-\infty$  };
    \node[updateNode] (A) at (3,0) {\nodepart{two} Normal \nodepart{three} $\bot$ \nodepart{four} $x$};
    \node[updateNode] (B) at (6,0) {\nodepart{two} Normal \nodepart{three} $\bot$ \nodepart{four} $x$};
    \node[descNode] (bp) at (7,2.8) {\nodepart{three} 2};
    \node[updateNode] (C) at (9,0) {\nodepart{two} Normal \nodepart{three} $\bot$ \nodepart{four} $x+1$};
    \node[updateNode] (tail) at (12,0){\nodepart{one} $\bot$ \nodepart{two} Normal \nodepart{three} $\bot$ \nodepart{four} $\infty$};

    \node at (0,1.7) {\large head};
    \node at (3,1.7) {\large A};
    \node at (7,4.2) {\large $insDesc[i]$};
    \node at (6,1.7) {\large B};
    \node at (9,1.7) {\large C};
    \node at (12,1.7){\large tail};

    \path[thick, *->] (head.one) edge (A.one west);
    \path[thick, *->] (A.one) edge (B.one west);
    \path[thick, *->] (B.one) edge (C.one west);
    \path[thick, ->, bend right] (bp.one) edge  (B);
    \path[thick, ->, bend left] (bp.two) edge (C);
    \path[thick, *->] (C.one) edge  (tail.one west);

    \end{tikzpicture}
    \caption{Finally $B$ replaces $insDesc[i]$ between $A$ and $C$ via $A.NextState.CAS(\twoField{\twoField{1}{i}}{InsFlag}, \twoField{B}{Normal})$.
    Process $p_i$ then increments the sequence number of $insDesc[i]$.}
    \end{subfigure}
    }
\caption{The insertion of an Update\-Node, $B$, between two consecutive Update\-Nodes, $A$ and $C$, in the algorithm.}
\label{NewInsProcFig}
\end{figure*}
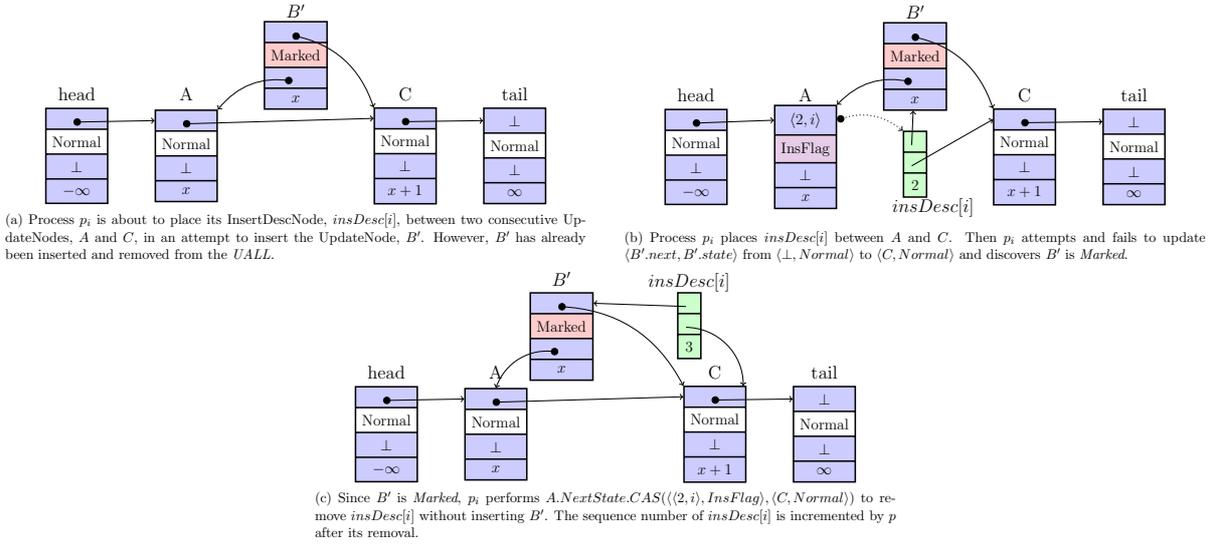
\begin{figure}[!ht]
    \centering
    \scalebox{0.48}{
    \begin{subfigure}[b]{1\columnwidth}
        \centering

        \begin{tikzpicture}
        \node[updateNode](head) at (0,0) {\nodepart{two} Normal \nodepart{three} $\bot$ \nodepart{four} $-\infty$  };
        \node[updateNode] (A) at (3,0) {\nodepart{two} Normal \nodepart{three} $\bot$  \nodepart{four} $x$};
        \node[markedNode] (B) at (6,2.5) {\nodepart{two} Marked \nodepart{four} $x$};
        \node[updateNode] (C) at (9,0) { \nodepart{two} Normal \nodepart{three} $\bot$ \nodepart{four} $x+1$};
        \node[updateNode] (tail) at (12,0){\nodepart{one} $\bot$ \nodepart{two} Normal \nodepart{three} $\bot$ \nodepart{four} $\infty$};
    
        \node at (0,1.7) {\large head};
        \node at (3,1.7) {\large A};
        \node at (6,4) {\large $B'$};
        \node at (9,1.7) {\large C};
        \node at (12,1.7){\large tail};

        \path[thick, *->] (head.one) edge (A.one west);
        \path[thick, *->] (A.one) edge (C.one west);
        \path[thick, *->, bend left] (B.one) edge (C.north west);
        \path[thick, *->] (C.one) edge  (tail.one west);
        \path[thick, *->, bend right] (B.three) edge (A.north east); 
    
        \end{tikzpicture}
    \caption{Process $p_i$ is about to place its Insert\-Desc\-Node, $insDesc[i]$, between two consecutive UpdateNodes, $A$ and $C$, in an attempt to insert the Update\-Node, $B'$.
    However, $B'$ has already been inserted and removed from the \textit{UALL}.}
    \end{subfigure}
    }
    \hfill
    \scalebox{0.48}{
    \begin{subfigure}[b]{1\columnwidth}
    \centering
        \begin{tikzpicture}
       \node[updateNode](head) at (0,0) {\nodepart{two} Normal \nodepart{three} $\bot$ \nodepart{four} $-\infty$  };
        \node[insFlagNode] (A) at (3,0) {\nodepart{one} $\twoField{2}{i}$ \nodepart{two} InsFlag \nodepart{three} $\bot$ \nodepart{four} $x$};
        \node[markedNode] (B) at (6,2.5) {\nodepart{two} Marked \nodepart{four} $x$};
        \node[descNode] (bq) at (6,-0.2) {\nodepart{three} 2};
        \node[updateNode] (C) at (9,0) {\nodepart{two} Normal \nodepart{three} $\bot$ \nodepart{four} $x+1$};
        \node[updateNode] (tail) at (12,0){\nodepart{one} $\bot$ \nodepart{two} Normal \nodepart{three} $\bot$ \nodepart{four} $\infty$};
    
        \node at (0,1.7) {\large head};
        \node at (3,1.7) {\large A};
        \node at (6,4) {\large $B'$};
        \node at (6.5,-1.4) {\large $insDesc[i]$};
        \node at (9,1.7) {\large C};
        \node at (12,1.7){\large tail};

        \path[thick, *->] (head.one) edge (A.one west);
        \path[thick, *->, dotted, bend left] (A.one east) edge (bq.north west);
        \path[thick, ->] (bq.one) edge  (B);
        \path[thick, ->] (bq.two) edge (C.one west);
        \path[thick, *->, bend left] (B.one) edge (C.north west);
        \path[thick, *->, bend right] (B.three) edge (A.north east); 
        \path[thick, *->] (C.one) edge  (tail.one west);

        \end{tikzpicture}
    \caption{Process $p_i$ places $insDesc[i]$ between $A$ and $C$. Then $p_i$ attempts and fails to 
    update $\twoField{B'.next}{B'.state}$ from $\twoField{\bot}{Normal}$ to $\twoField{C}{Normal}$ and discovers $B'$ is \textit{Marked}.}
    \end{subfigure}
    }
    \scalebox{0.48}{
    \begin{subfigure}[b]{1\textwidth}
        \centering
        \begin{tikzpicture}
    
        \node[updateNode](head) at (0,0) {\nodepart{two} Normal \nodepart{three} $\bot$ \nodepart{four} $-\infty$  };
        \node[updateNode] (A) at (3,0) {\nodepart{two} Normal \nodepart{three} $\bot$ \nodepart{four} $x$};
        \node[markedNode] (B) at (4.8,2.7) {\nodepart{two} Marked \nodepart{four} $x$};
        \node[descNode] (bq) at (8.3,3) {\nodepart{three} 3};
        \node[updateNode] (C) at (9,0) {\nodepart{two} Normal \nodepart{three} $\bot$ \nodepart{four} $x+1$};
        \node[updateNode] (tail) at (12,0){\nodepart{one} $\bot$ \nodepart{two} Normal \nodepart{three} $\bot$ \nodepart{four} $\infty$};
    
        \node at (0,1.7) {\large head};
        \node at (3,1.7) {\large A};
        \node at (4.8,4.3) {\large $B'$};
        \node at (8.3,4.2) {\large $insDesc[i]$};
        \node at (9,1.7) {\large C};
        \node at (12,1.7){\large tail};

        \path[thick, *->] (head.one) edge (A.one west);
        \path[thick, *->] (A.one) edge (C.one west);
        \path[thick, ->] (bq.one) edge  (B.one east);
        \path[thick, ->, bend left=45] (bq.two) edge (C);
        \path[thick, *->, bend left] (B.one) edge (C.north west);
        \path[thick, *->, bend right=40] (B.three) edge (A.one north); 
        \path[thick, *->] (C.one) edge  (tail.one west);

        \end{tikzpicture}
    \caption{Since $B'$ is \textit{Marked}, $p_i$ performs $A.NextState.CAS(\twoField{\twoField{2}{i}}{InsFlag}, \twoField{C}{Normal})$ 
    to remove $insDesc[i]$ without inserting $B'$. The sequence number of $insDesc[i]$ is incremented by $p$ after its removal.}
    \end{subfigure}
    }
\caption{The reinsertion of an Update\-Node, $B$, after it has been marked and removed is prevented.}
\label{PreventReinsFig}
\end{figure}

Checking that \textit{uNode} is not \textit{Marked} before an Insert\-Desc\-Node 
is inserted is not necessary for the correctness of the algorithm.
Processes will check that \textit{uNode} is not marked before attempting to 
insert \textit{uNode} in the place of the Insert\-Desc\-Node.
However, placing an Insert\-Desc\-Node in the \textit{UALL} for a marked Update\-Node will only serve to 
slow down the processes that encounter the Insert\-Desc\-Node, which will need to remove the Insert\-Desc\-Node using a CAS.
Therefore attempting to avoid the unnecessary insertion of an Insert\-Desc\-Node for a marked Update\-Node 
seems reasonable.


To perform an instance of $\textit{remove(uNode)}$, a process, $p$, will traverse the \textit{UALL} starting from \textit{head}
in search of \textit{uNode}.
If \textit{uNode} is not discovered before encountering an Update\-Node with key greater than \textit{uNode.key} while searching through the \textit{UALL}, then $p$ will return, as \textit{uNode} was not in the \textit{UALL}.
Just as with insertion, if $p$ encounters an Update\-Node whose \textit{state} is \textit{InsFlag} or \textit{DelFlag} while searching through the \textit{UALL}, 
$p$ will use $helpInsert$ or $helpRemove$ respectively to attempt to help with the pending \textit{UALL} operation and update the \textit{state} of the Update\-Node
back to \textit{Normal}.
If $p$ encounters a marked Update\-Node, $curr$, $p$ will attempt to use $helpMarked$ to remove $curr$ from the \textit{UALL}, 
before resuming its search from the Update\-Node pointed to by $curr$'s \textit{backlink}.
When \textit{uNode} is encountered, $p$ will try to update the \textit{state} of 
the Update\-Node preceding \textit{uNode} from \textit{Normal} to \textit{DelFlag} via \textit{CAS}.
If this \textit{CAS} is successful, $p$ will use \textit{helpRemove} to mark and
remove \textit{uNode} from the \textit{UALL} before returning.
If the \textit{CAS} is not successful, $p$ will resume its search for \textit{uNode} in 
the \textit{UALL}.
\begin{figure}[!ht] \label{remove}
    \begin{algorithmic}[1]
        \alglinenoPop{alg1}
        \State $remove(uNode)$ 
        \Indent
            \State $curr \gets head$\label{currInitRem} \Comment{Start searching from \textit{head}.}
            \State $\twoField{next'}{state'} \gets curr.NextState.read()$ \label{removeNextInit}
            \Loop \label{removeLoop}
                \If{$state' = Normal$}
                    \If{$next'.key > uNode.key$}
                        \Return\label{removeReturnFalse1} \Comment{\textit{uNode} was not in the \textit{UALL}.}
                    \EndIf
                    \If{$next' \neq uNode$}
                        \State $curr \gets next'$ \label{removeAdvance} \Comment{Continue search from following Update\-Node.}
                        \State $\twoField{next'}{state'} \gets curr.NextState.read()$\label{removeAdvanceUpdate}
                    \Else
                        \State $\twoField{next'}{state'} \gets curr.NextState.CAS(\twoField{uNode}{Normal}, \twoField{uNode}{DelFlag})$\label{delFlagCAS}
                        \If{$\twoField{next'}{state'} = \twoField{uNode}{Normal}$}
                            \State $helpRemove(curr, uNode)$ \Comment{CAS succeeded, complete removal.}\label{removeHR1}
                            \State \Return \label{removeReturnTrue} 
                        \EndIf
                    \EndIf
                \ElsIf{$state' = InsFlag$}
                    \State $\twoField{seq}{j} \gets next'$ \Comment{Get $seqNum$ and $pid$ from $next'$.}
                    \State $\twoField{next'}{state'} \gets helpInsert(curr, seq, j)$\label{removeHI}\Comment{Help insertion.}
                \ElsIf{$next'.key > uNode.key$}
                    \Return\label{removeReturnFalse2} \Comment{\textit{uNode} was not in the \textit{UALL}.}
                \ElsIf{$state' = DelFlag$}
                    \State $\twoField{next'}{state'} \gets helpRemove(curr, next')$\label{removeHR2} \Comment{Help remove \textit{next} field.}
                    \If{$next' = uNode$}
                        \Return\label{removeReturnFalse3} \Comment{\textit{uNode} was removed.}
                    \EndIf
                \Else \Comment{$curr$ is Marked.}
                    \State $prev \gets curr.backlink.read()$
                    \State $\twoField{next'}{state'} \gets prev.NextState.read()$
                    \If{$\twoField{next'}{state'} = \twoField{curr}{DelFlag}$}
                        \State $\twoField{next'}{state'} \gets helpMarked(prev, curr)$\label{removeHM}\Comment{Help remove $curr$.}
                    \EndIf
                    \State $curr \gets prev$\label{removeCurrSetToPrev}\Comment{Continue searching from $curr$'s backlink.}
                \EndIf
            \EndLoop
        \EndIndent
        \alglinenoPush{alg1}
    \end{algorithmic}
\end{figure}

Figure \ref{NodeRemoveFig} depicts the process of removing an Update\-Node $B$ from the \textit{UALL}.
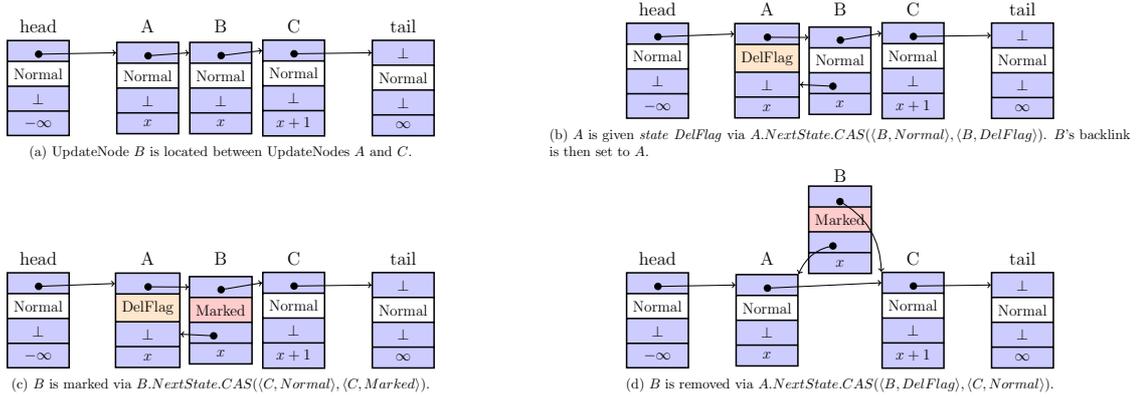
\begin{figure}[!ht] 
    \centering
    \scalebox{0.48}{
    \begin{subfigure}[b]{1\textwidth}
    \centering
    \begin{tikzpicture}

    \node[updateNode] at (0,0) {\nodepart{two} Normal \nodepart{three} $\bot$ \nodepart{four} $-\infty$  };
    \node[updateNode] (A) at (3,0) { \nodepart{two} Normal \nodepart{three} $\bot$ \nodepart{four} $x$};
    \node[updateNode] (B) at (5,0) {\nodepart{two} Normal \nodepart{three} $\bot$ \nodepart{four} $x$};
    \node[updateNode] (C) at (7,0) { \nodepart{two} Normal \nodepart{three} $\bot$ \nodepart{four} $x+1$};
    \node[updateNode] (tail) at (10,0){\nodepart{one} $\bot$ \nodepart{two} Normal \nodepart{three} $\bot$ \nodepart{four} $\infty$};

    \node at (0,1.7) {\large head};
    \node at (3,1.7) {\large A};
    \node at (5,1.7) {\large B};
    \node at (7,1.7) {\large C};
    \node at (10,1.7){\large tail};
    
    \path[thick, *->] (head.one) edge (A.one west);
    \path[thick, *->] (A.one) edge (B.one west);
    \path[thick, *->] (B.one) edge (C.one west);
    \path[thick, *->] (C.one) edge  (tail.one west);

    \end{tikzpicture}
    \caption{Update\-Node $B$ is located between Update\-Nodes $A$ and $C$.}
    \end{subfigure}
    }
    \hfill
    \scalebox{0.48}{
    \begin{subfigure}[b]{1\textwidth}
        \centering
        \begin{tikzpicture}
        \node[updateNode](head) at (0,0) {\nodepart{two} Normal \nodepart{three} $\bot$ \nodepart{four} $-\infty$  };
        \node[delFlagNode] (A) at (3,0) {\nodepart{two} DelFlag \nodepart{three} $\bot$ \nodepart{four} $x$};
        \node[updateNode] (B) at (5,0) {\nodepart{two} Normal \nodepart{four} $x$};
        \node[updateNode] (C) at (7,0) {\nodepart{two} Normal \nodepart{three} $\bot$ \nodepart{four} $x+1$};
        \node[updateNode] (tail) at (10,0){\nodepart{one} $\bot$ \nodepart{two} Normal \nodepart{three} $\bot$ \nodepart{four} $\infty$};
    
        \node at (0,1.7) {\large head};
        \node at (3,1.7) {\large A};
        \node at (5,1.7) {\large B};
        \node at (7,1.7) {\large C};
        \node at (10,1.7){\large tail};

        \path[thick, *->] (head.one) edge (A.one west);
        \path[thick, *->] (A.one) edge (B.one west);
        \path[thick, *->] (B.one) edge (C.one west);
        \path[thick, *->] (C.one) edge  (tail.one west);
        \path[thick, *->] (B.three) edge (A.three east); 
    
        \end{tikzpicture}
    
    \caption{$A$ is given \textit{state} \textit{DelFlag} via $A.NextState.CAS(\twoField{B}{Normal}, \twoField{B}{DelFlag})$. $B$'s backlink is then set to $A$.}
    \end{subfigure}
    }
    \scalebox{0.48}{
    \begin{subfigure}[b]{\textwidth}
        \centering
        \begin{tikzpicture}
        \node[updateNode](head) at (0,0) {\nodepart{two} Normal \nodepart{three} $\bot$ \nodepart{four} $-\infty$  };
        \node[delFlagNode] (A) at (3,0) {\nodepart{two} DelFlag \nodepart{three} $\bot$ \nodepart{four} $x$};
        \node[markedNode] (B) at (5,0) {\nodepart{two} Marked \nodepart{four} $x$};
        \node[updateNode] (C) at (7,0) {\nodepart{two} Normal \nodepart{three} $\bot$ \nodepart{four} $x+1$};
        \node[updateNode] (tail) at (10,0){\nodepart{one} $\bot$ \nodepart{two} Normal \nodepart{three} $\bot$ \nodepart{four} $\infty$};
    
        \node at (0,1.7) {\large head};
        \node at (3,1.7) {\large A};
        \node at (5,1.7) {\large B};
        \node at (7,1.7) {\large C};
        \node at (10,1.7){\large tail};

        \path[thick, *->] (head.one) edge (A.one west);
        \path[thick, *->] (A.one) edge (B.one west);
        \path[thick, *->] (B.one) edge (C.one west);
        \path[thick, *->] (C.one) edge  (tail.one west);
        \path[thick, *->] (B.three) edge (A.three east); 
    
        \end{tikzpicture}
    \caption{$B$ is marked via $B.NextState.CAS(\twoField{C}{Normal}, \twoField{C}{Marked})$.}
    \end{subfigure}
    }
    \hfill
    \scalebox{0.48}{
    \begin{subfigure}[b]{\textwidth}
        \centering
        \begin{tikzpicture}
    
        \node[updateNode](head) at (0,0) {\nodepart{two} Normal \nodepart{three} $\bot$ \nodepart{four} $-\infty$  };
        \node[updateNode] (A) at (3,0) { \nodepart{two} Normal  \nodepart{three} $\bot$ \nodepart{four} $x$};
        \node[markedNode] (B) at (5,2.5) {\nodepart{two} Marked \nodepart{four} $x$ };
        \node[updateNode] (C) at (7,0) { \nodepart{two} Normal \nodepart{three} $\bot$ \nodepart{four} $x+1$};
        \node[updateNode] (tail) at (10,0){\nodepart{one} $\bot$ \nodepart{two} Normal \nodepart{three} $\bot$ \nodepart{four} $\infty$};
    
        \node at (0,1.7) {\large head};
        \node at (3,1.7) {\large A};
        \node at (5,4) {\large B};
        \node at (7,1.7) {\large C};
        \node at (10,1.7){\large tail};
    
        \path[thick, *->] (head.one) edge (A.one west);
        \path[thick, *->] (A.one) edge (C.one west);
        \path[thick, *->, bend left] (B.one) edge (C.north west);
        \path[thick, *->] (C.one) edge  (tail.one west);
        \path[thick, *->, bend right] (B.three) edge (A.north east); 
    
        \end{tikzpicture}
    \caption{$B$ is removed via $A.NextState.CAS(\twoField{B}{DelFlag}, \twoField{C}{Normal})$.}
    \end{subfigure}
    }
\caption{The removal of an Update\-Node $B$ between $A$ and $C$.}
\label{NodeRemoveFig}
\end{figure}

To perform an instance of \textit{readNextUpdate\-Node(uNode)}, a process, $p$, first reads the \textit{next} field and \textit{state} fields 
of \textit{uNode} via a read of \textit{uNode.NextState}. 
Suppose \textit{uNode}'s \textit{state} is not \textit{InsFlag}.
Then its \textit{next} field points to an Update\-Node, and $p$ returns a
pointer to this Update\-Node or $\bot$ if this next Update\-Node was \textit{tail}.
Now suppose \textit{uNode}'s \textit{state} is \textit{InsFlag}.
Then there is an Insert\-Desc\-Node $insDesc[q]$
with sequence number $seq$ immediately following \textit{uNode}, where $q$ is the process id of some process. 
Process $p$ will try to read the pointer to the Update\-Node, $next'$, that immediately follows $insDesc[q]$ in 
the \textit{UALL} before the \textit{seqNum} of $insDesc[q]$ changes from $seq$. 
If the sequence number has not changed, $p$ will return a pointer to $next'$ or 
$\bot$ if $next'$ is \textit{tail}.
If the sequence number has changed, $p$ will read the \textit{next} field and \textit{state} 
fields of \textit{uNode} again and start over.
\begin{figure}[H]
\begin{algorithmic}[1]
\alglinenoPop{alg1}
    \State $readNextUpdate\-Node(uNode)$
    \Indent
        \State $\twoField{next'}{state'} \gets uNode.NextState.read()$
        \While{$state' = InsFlag$} \Comment{There is an Insert\-Desc\-Node following \textit{uNode}.}
            \State $\twoField{seq}{j} \gets next'$ \Comment{Get $seqNum$ and $pid$ from $next'$.}
            \State $next' \gets insDesc[j].next.read()$ \Comment{Try to read \textit{next} field before \textit{seqNum} changes}
            \If {$insDesc[j].seqNum.read() = seq$} 
                \Break 
            \EndIf
            \State $\twoField{next'}{state'} \gets uNode.NextState.read()$ \Comment{Try again.}
        \EndWhile
        \If{$next' = tail$} 
            \Return $\bot$ \Comment{No Update\-Node following \textit{uNode}.}
        \Else \
             \Return $next'$ \Comment{$next'$ was a successor of \textit{uNode}.}
        \EndIf
    \EndIndent   
\alglinenoPush{alg1}
\end{algorithmic}
\end{figure}

To perform an instance of $first()$, a process, \textit{p}, will simply return the 
result of an instance of $readNextUpdate\-Node(head)$.
This returns $\bot$ if \textit{tail} is the UpdateNode that immediately follows \textit{head} in the \textit{UALL}, 
otherwise it returns a pointer to the Update\-Node that immediately follows \textit{head}.

\subsubsection{Correctness and Step Complexity Argument}
We argue why our implementation of the \textit{UALL} is correct.
Since our implementation extends Fomitchev and Ruppert's linked list
with a mechanism that allows many processes to help insert the same 
UpdateNode into the \textit{UALL}, we focus on arguing why this mechanism 
works.

Consider an operation instance, $iOp$, that is trying to insert an 
UpdateNode, \textit{uNode}, into the \textit{UALL}.
Suppose $iOp$ inserts an InsertDescNode, $iDesc$, that points to \textit{uNode} into the \textit{UALL}
between two consecutive UpdateNodes, $A$ and $C$, such that $A.key \le uNode.key < C.key$.
When an UpdateNode with the same key as \textit{uNode} is inserted into the \textit{UALL}, 
it is inserted after every UpdateNode in the \textit{UALL} 
of a lesser or equal key and it is immediately followed by an UpdateNode with a larger key.
Thus, until $iDesc$ is removed from the \textit{UALL}, the set of UpdateNodes 
in the \textit{UALL} with the same key as \textit{uNode} is a subset of those
that $iOp$ encountered before it inserted $iDesc$.
If $iOp$ had encountered \textit{uNode} while traversing to $A$ and $C$, it would 
have returned, so while $iDesc$ is in the \textit{UALL},
\textit{uNode} is not in the \textit{UALL}.
Our implementation is an extension of Fomitchev and Ruppert's linked 
list, so if \textit{uNode} was previously in the \textit{UALL} and 
was later removed, it would be marked.
When any process encounters $iDesc$ in the \textit{UALL}, 
it will check if \textit{uNode} is marked.
If \textit{uNode} is marked, the process will remove $iDesc$ 
from between $A$ and $C$ without inserting \textit{uNode}, provided $iDesc$ has not already been removed.
If \textit{uNode} is not marked, then \textit{uNode} was never previously in the 
\textit{UALL} and \textit{uNode} remains unmarked until it is inserted, 
that is, while $iDesc$ is still in the \textit{UALL}.
The process will perform a \textit{CAS} to insert \textit{uNode} 
in the place of $iDesc$ following $A$ in the \textit{UALL}.
Therefore UpdateNodes that were previously inserted 
and later removed from the \textit{UALL} are not 
reinserted.
Notice that if \textit{uNode} was inserted in the place of $iDesc$,
$iDesc$ was the first InsertDescNode pointing 
to \textit{uNode} that was inserted into the \textit{UALL}.
It remains between
$A$ and $C$ until it is removed and \textit{uNode} is inserted in its place.
Therefore \textit{uNode.next} only needs to be set once to point to $C$ 
before \textit{uNode} is inserted.
Once \textit{uNode.next} has been set, no process changes \textit{uNode.next} until \textit{uNode} 
is in the \textit{UALL}.
This ensures that when \textit{uNode} is inserted in the place of $iDesc$,
only \textit{uNode} is added to the \textit{UALL} and no UpdateNode
is removed.

We also argue that the amortized step complexity of insert and remove
remains the same as in Fomitchev and Ruppert's list.
Consider an insert or remove operation instance, $op$.
At most a constant number of steps is performed by $op$ for every UpdateNode
it encounters, plus a constant number of steps for every time it helps 
with the insertion or removal of an UpdateNode.
Notice that when an instance updates the \textit{state} of an UpdateNode 
to either \textit{InsFlag} or \textit{DelFlag}, it ends after a
constant number of steps (since it will finish its operation).
Every time $op$ helps with the insertion or removal of an UpdateNode, 
there is an insertion or removal instance that overlaps $op$ which 
accomplished its task and either finished its operation instance 
or will finish it if it performs a constant number of additional steps.
Therefore, the number of steps taken by $op$ is at most a constant times
the number of UpdateNodes in the \textit{UALL} when $op$ starts 
plus the number of insert and remove instances that overlap $op$.
Therefore the complexity of $op$ is $O(m(op) + \overline{c}(op))$,
where $m(op)$ is the number of UpdateNodes in the \textit{UALL} at the 
beginning of $op$ and $\overline{c}(op)$ is the interval contention of $op$.
Gibson and Gramoli\cite{10.1007/978-3-662-48653-5_14} showed that in any execution,
the total interval contention of all operation instances is at most twice 
the total of their point contention.
Ruppert\cite{Ruppert_2016} gives a simple proof of this result.
Therefore, the amortized complexity of insert and remove instances 
on the \textit{UALL} is $O(m(op) + \dot{c}(op))$.

\pagebreak

\subsection{\textit{RUALL} and the \textit{ruallPosition} Field of PredecessorNodes}
\label{ruallSection}
In this section we discuss the \textit{RUALL} and the \textit{ruallPosition} field of PredecessorNodes.
The \textit{RUALL} is a linked list that contains Update\-Nodes and may contain Insert\-Desc\-Nodes. 
It is the same as the \textit{UALL}, except the Update\-Nodes stored in the \textit{RUALL}
are sorted in non-increasing order by key rather than non-decreasing order. 
Update\-Nodes with the same key in the \textit{RUALL} 
are still ordered from least recently to most recently inserted.
Like the \textit{UALL}, the \textit{RUALL} has two Update\-Nodes used as sentinels,
\textit{rHead} and \textit{rTail}, such that \textit{rHead} has key $\infty$ and \textit{rTail} has key $-\infty$.
Every Update\-Node has \textit{rNext}, \textit{rState} and \textit{rBacklink} fields, which
are analogous to the \textit{next} field, \textit{state} and \textit{backlink} fields for the \textit{UALL}.
The \textit{rNext} and \textit{rState} fields are stored in a single \textit{CAS} 
object, called \textit{rNextState}.
Just as with the \textit{UALL}, 
each process has its own Insert\-Desc\-Node which it uses to insert Update\-Nodes into the \textit{RUALL}.
The fields of Insert\-Desc\-Nodes used for the \textit{RUALL} are exactly the same as those used for the \textit{UALL}.
Like the \textit{UALL}, the \textit{RUALL} supports the \textit{first}, \textit{readNextUpdateNode(uNode)}, \textit{insert(uNode)} and \textit{remove(uNode)} operations, where \textit{uNode} is a Update\-Node.
The \textit{first} and \textit{readNextUpdateNode(uNode)} operations are only performed by a process during a predecessor (or embedded predecessor) operation.
The implementations of \textit{insert(uNode)} and \textit{remove(uNode)} are symmetric to their implementations for the \textit{UALL}. 

The \textit{ruallPosition} field of a PredecessorNode, \textit{pNode}, stores a pointer to an Update\-Node
that is in the \textit{RUALL} or was in it previously.
Initially, \textit{ruallPosition} stores a pointer to \textit{rHead}.
Let $pOp$ be the operation instance that created \textit{pNode}.
While $pOp$ traverses the \textit{RUALL}, \textit{pNode.ruallPosition}
points to the Update\-Node, \textit{uNode}, that $pOp$ is currently visiting.
To advance to the Update\-Node which either immediately follows \textit{uNode} or immediately followed \textit{uNode} 
before it was removed from the \textit{RUALL}, $pOp$ will invoke $copyNext(uNode)$.
This operation atomically copies a pointer to the UpdateNode immediately following 
\textit{uNode} into \textit{pNode.ruallPosition} and returns a copy of this pointer.
Other processes may \textit{read} the UpdateNode pointer stored in \textit{ruallPosition}.
Our implementation is lock-free but not wait-free.
Operation instances trying to copy a pointer to the UpdateNode 
immediately following an UpdateNode, \textit{cur}, into the \textit{ruallPosition}
could take an unbounded number of steps if InsertDescNodes 
continue to be inserted and removed following \textit{cur} in the \textit{RUALL}.
The \textit{ruallPosition} fields can be implemented using Blelloch and Wei's single-writer Destination object \cite{DBLP:conf/wdag/BlellochW20}.
We present a simplified version of their implementation which is sufficient for our context.



As in the \textit{UALL}, the \textit{first()} operation
obtains a pointer to the Update\-Node, $vNode$, that immediately follows \textit{rHead} in the \textit{RUALL}.
However, it does this by performing an instance of \textit{copyNext(rHead)} on \textit{pNode.ruallPosition},
that atomically reads a pointer to $vNode$ from $rHead.rNext$ and copies it into \textit{pNode.ruallPosition}.
If $vNode$ is \textit{rTail}, then $\bot$ is returned. Otherwise, a pointer to $vNode$ is returned. 
The \textit{readNextUpdateNode(uNode)} operation obtains a pointer to the Update\-Node, $vNode$, that immediately follows \textit{uNode} in the \textit{RUALL} or, if \textit{uNode} is no longer in the \textit{RUALL}, 
that immediately followed \textit{uNode} when it was removed.
It does this by performing an instance of \textit{copyNext(uNode)} on \textit{pNode.ruallPosition}, 
that atomically obtains a pointer to $vNode$ and copies it into \textit{pNode.ruallPosition}.
If $vNode$ is \textit{rTail}, $\bot$ is returned. Otherwise, a pointer to $vNode$ is returned.

The \textit{ruallPosition} field of a PredecessorNode is represented using a single \textit{writeable CAS} object, \textit{data}, which stores two fields, \textit{ptr} and \textit{copying}.
The \textit{ptr} field stores a pointer to an Update\-Node, initially storing a pointer to \textit{rHead}.
The \textit{copying} field stores a single bit that indicates whether a \textit{copyNext} operation by $w$ is in progress.
Initially \textit{copying} is equal to 0. When \textit{copying} is equal to 0, no instance of \textit{copyNext} by $w$ is in progress and 
\textit{ptr} stores the current value of \textit{ruallPosition}.
When performing a \textit{copyNext(uNode)} operation, $w$ will update \textit{copying} to 1 and 
\textit{ptr} to \textit{uNode}.
If \textit{copying} is equal to 1, then \textit{ptr} stores a pointer to an Update\-Node, \textit{uNode}, such that $w$
is attempting to copy a pointer to the Update\-Node $u$ that immediately follows (or immediately followed) \textit{uNode} into the \textit{ruallPosition}.
A \textit{CAS} will be attempted by $w$ to try to update \textit{copying} from 1 to 0 and \textit{ptr} from pointing to \textit{uNode} to pointing to $u$.
If the \textit{CAS} is successful, this completes the atomic copy by 
If a process $q$ tries to read \textit{ruallPosition}, it first reads \textit{data} to obtain the values of \textit{ptr} and \textit{copying}. 
If \textit{copying} is equal to 0, then $q$ returns the contents of \textit{ptr}.
If \textit{copying} is equal to 1, then $q$ must help finish a \textit{copyNext(uNode)} operation that is being performed by $w$, where \textit{uNode} is equal to \textit{ptr}.
First, $q$ obtains a pointer to the Update\-Node $u'$ which immediately follows (or followed) \textit{uNode} in the \textit{RUALL}.
Next, $q$ will perform a \textit{CAS} to attempt to update \textit{ptr} to point to $u'$, \textit{copying} back to 0.
If the \textit{CAS} is successful, the \textit{read} returns a pointer to $u'$,
as $q$ has completed the atomic copy initiated by $w$.
If the \textit{CAS} is not successful, the \textit{read} returns the value of \textit{ptr} returned by the \textit{CAS}.

To perform a \textit{copyNext(uNode)} operation on the \textit{ruallPosition} field, $w$ will first write $\twoField{uNode}{1}$ to \textit{data}.
Then, $w$ will read the \textit{rNext} and \textit{rState} fields of \textit{uNode}.
If \textit{rState} is \textit{Normal}, \textit{DelFlag} or \textit{Marked},
then \textit{rNext} held a pointer to the Update\-Node \textit{uNode'} that follows (or followed) \textit{uNode} in the \textit{RUALL}.
If \textit{rState} is \textit{InsFlag}, then \textit{rNext} stores a process id $x$ and sequence number $seq$, 
and $w$ will try to read the \textit{next} field and \textit{seqNum} fields of the Insert\-Desc\-Node $insertDesc[x]$ before its \textit{seqNum}
has changed from $seq$. If it is not equal to $seq$, then $w$ will reread the \textit{rNext} and \textit{rState} fields of \textit{uNode},
and tries again to the read the Update\-Node 
If the sequence number of $insertDesc[x]$ is still equal to $seq$, then its \textit{next} field stored a pointer
to the Update\-Node \textit{uNode'} which immediately followed \textit{uNode} in the list.
Once $w$ has a pointer to \textit{uNode'}, 
$w$ will attempt a \textit{CAS} on $data$ to update \textit{ptr} from \textit{uNode} to \textit{uNode'} and 
\textit{copying} from 1 to 0.
If the \textit{CAS} is successful, then a pointer to \textit{uNode'} was successfully copied into the \textit{ruallPosition} field,
so $w$ returns \textit{uNode}.
If the \textit{CAS} is not successful, then some instance of \textit{read} atomically copied a pointer to an Update\-Node immediately following \textit{uNode}
into \textit{ptr}, so $w$ returns the value of \textit{ptr} following its \textit{CAS}.
\begin{figure}[H]
\begin{algorithmic}[1]
\alglinenoNew{alg20}
    \State $copyNext(uNode)$ 
    \Indent
        \State $data.write(\twoField{uNode}{1})$ \Comment{Signal start of atomic copy.}
        \State $\twoField{next'}{state'} \gets uNode.rNextState.read()$
        \While{$state' = InsFlag$} \Comment{There is an Insert\-Desc\-Node following \textit{uNode}.}
            \State $\twoField{seq}{j} \gets next'$ \Comment{Get $seqNum$ and $pid$ from $next'$.}
            \State $next' \gets insDesc[j].next.read()$ \Comment{Try to read \textit{next} field before \textit{seqNum} changes}
            \If {$insDesc[j].seqNum.read() = seq$} 
                \Break 
            \EndIf
            \State $\twoField{next'}{state'} \gets uNode.rNextState.read()$ \Comment{Try again.}
        \EndWhile
        \State $\twoField{newPtr}{c} \gets data.CAS(\twoField{uNode}{1}, \twoField{next'}{0})$
        \If{$newPtr = uNode$} 
            \Return $next'$ \Comment{\textit{CAS} was successful.}
        \Else \
             \Return $newPtr$ \Comment{Some other process completed atomic copy.}
        \EndIf
    \EndIndent   
\alglinenoPush{alg20}
\end{algorithmic}
\end{figure}

To perform a \textit{read} operation on the \textit{ruallPosition} of a PredecessorNode,
a process, \textit{p}, will read \textit{data} to obtain the values of the \textit{ptr} and \textit{copying} fields.
If \textit{copying} was equal to 0, then the value of \textit{ptr} stored the current value of the \textit{ruallPosition} field,
which $p$ returns.
Otherwise, \textit{ptr} stored a pointer to an Update\-Node $u$ such that a pointer to the Update\-Node $u'$ that immediately follows (or followed) $u$ 
in the \textit{RUALL} is being copied into the \textit{ruallPosition}. 
In this case, $p$ will obtain a pointer, $\mathit{next'}$ to this Update\-Node immediately following $u$ in the \textit{RUALL}.
Then $p$ performs a \textit{CAS} on $data$ to attempt to update \textit{ptr} from pointing to $u$ to pointing to $u'$ and 
\textit{copying} from 1 to 0.
If this \textit{CAS} is successful, then $p$ returns $\mathit{next'}$.
If the \textit{CAS} is unsuccessful, then the value of \textit{ptr} which was returned by the \textit{CAS} is returned by $p$.
\begin{figure}[H]
\begin{algorithmic}[1]
\alglinenoPop{alg20}
    \State $read()$
    \Indent
        \State $\twoField{ptr'}{copying'} \gets data.read()$
        \If{$copying' = 0$}
            \Return $ptr'$ \Comment{No atomic copy in progress.}
        \EndIf
        \State $uNode \gets ptr'$ \Comment{Copying pointer to Update\-Node immediately following \textit{uNode}.}
        \State $\twoField{next'}{state'} \gets uNode.rNextState.read()$
        \While{$state' = InsFlag$} \Comment{There is an Insert\-Desc\-Node following \textit{uNode}.}
            \State $\twoField{seq}{j} \gets next'$ \Comment{Get $seqNum$ and $pid$ from $next'$.}
            \State $next' \gets insDesc[j].next.read()$ \Comment{Try to read \textit{next} field before \textit{seqNum} changes}
            \If {$insDesc[j].seqNum.read() = seq$} 
                \Break 
            \EndIf
            \State $\twoField{next'}{state'} \gets uNode.rNextState.read()$ \Comment{Try again.}
        \EndWhile
        \State $\twoField{newPtr}{c} \gets data.CAS(\twoField{uNode}{1}, \twoField{next'}{0})$
        \If{$newPtr = uNode$} 
            \Return $next'$ \Comment{\textit{CAS} was successful.}
        \Else \
             \Return $newPtr$ \Comment{Some other process completed atomic copy.}
        \EndIf
    \EndIndent   
\alglinenoPush{alg20}
\end{algorithmic}
\end{figure}

We use \textit{writeable CAS} objects instead of weak \textit{LL/SC} objects in our implementation because we need not worry about the \textit{ABA problem}.
The UpdateNodes in the \textit{RUALL} are sorted in non-increasing order 
by \textit{key} and from least-recently to most-recently inserted.
It is never possible to follow \textit{rNext} fields starting 
from an UpdateNode and end up back at the same UpdateNode.
Therefore, a pointer to an UpdateNode is copied into the \textit{ruallPosition} of a
PredecessorNode at most once.
Another difference between our implementation and Blelloch and Wei's is that 
\textit{copyNext} is always copying a pointer to the UpdateNode immediately 
following the one currently pointed to by \textit{ruallPosition} into the object.
Unlike their implementation, we do not need to store a 
pointer to the shared object whose contents are being copied into the \textit{ruallPosition}.

\subsection{PALL}
\label{pallImplementation}
The \textit{PALL} is a linked list of PredecessorNodes. 
It contains two PredecessorNodes, \textit{pHead} and \textit{pTail}, which serve as sentinels.
Unlike the \textit{UALL} or \textit{RUALL}, the PredecessorNodes of the \textit{PALL}
are not sorted by their keys. Instead, PredecessorNodes in the list are ordered from most recently to least recently inserted.
It supports the following operations, where \textit{pNode} is a PredecessorNode:
\begin{itemize}
    \item \textit{insert(pNode)}: This operation inserts \textit{pNode} into the \textit{PALL} immediately following \textit{pHead}. 
    \item \textit{remove(pNode)}: This operation removes \textit{pNode} from the \textit{PALL}.
    \item \textit{first()}: Returns a pointer to the PredecessorNode immediately following \textit{pHead}, or $\bot$ if {pTail} immediately follows \textit{pHead}.
    \item \textit{readNext(pNode)}: This operation can only be performed if \textit{pNode} has been inserted into the \textit{PALL}.
    If \textit{pNode} is in the \textit{PALL}, this operation returns a pointer to the PredecessorNode immediately following \textit{pNode} in the \textit{PALL},
    or $\bot$ if \textit{pTail} immediately follows \textit{pNode}.
    If \textit{pNode} is not in the \textit{PALL}, this operation returns a pointer to the PredecessorNode that immediately followed \textit{pNode} in the \textit{PALL} when it was removed,
    or $\bot$ if \textit{pTail} was immediately following \textit{pNode} when it was removed.
\end{itemize}
A PredecessorNode may be inserted into the \textit{PALL} and removed from it exclusively by the 
process that created the PredecessorNode.
Once inserted and removed, a PredecessorNode remains permanently removed from the \textit{PALL}.
Since the \textit{PALL} does not need to support concurrent insertions of the same PredecessorNode by multiple processes, we use an unsorted variant of Fomitchev and Ruppert's linked list to implement it.
This is more efficient than our implementations of the \textit{UALL} and \textit{RUALL} which use 
helping to coordinate concurrent insertions of UpdateNodes.
Every PredecessorNode has a \textit{successor} field and a \textit{backlink}
field which is used in this implementation.

\section{Memory Reclamation}
\label{memoryReclamationSection}
An implementation of a shared data structure may make use of dynamically allocated \textit{records},
which are composed of one or more shared objects.
These shared objects are called the fields of the record.
A record is accessed when a process performs a step on one 
of its fields.
The Update\-Nodes, Predecessor\-Nodes and Notify\-Nodes (introduced in Section \ref{relaxedTrie} and Section \ref{lockFreeTrie}) are examples of records.
a process, \textit{p}, can request shared memory from the system for a record \textit{r}, at which point we 
say that \textit{r} has been \textit{allocated}. 
We say that it is \textit{safe to reclaim} a record \textit{r} when it is known that 
no process will subsequently attempt to access \textit{r}.
If this is the case, a single process may \textit{reclaim} \textit{r}.
This \textit{deallocates} the shared memory used by \textit{r} back into the system.

A \textit{memory reclamation scheme} helps processes determine when it is safe to reclaim 
the records of a shared data structure. 
Some memory reclamation schemes are called \textit{automatic}
as the algorithms for operations on the shared data structure do not need to be modified in order to reclaim memory.
An example of an automatic memory reclamation scheme is garbage collection.
Other memory reclamation schemes such as hazard pointers, epoch-based reclamation or reference counting 
are \textit{non-automatic}, as they require processes to actively participate in the reclamation of memory. In our implementation of Ko's Trie, we only use non-automatic memory reclamation techniques.

In this section, we present the memory reclamation techniques that are used to reclaim
records in our implementation of Ko's Trie. 
In subsection \ref{memoryrelatedwork}, \textit{epoch-based reclamation} and \textit{reference counting} are introduced,
which are used together to reclaim memory in our implementation.
In subsection \ref{debraProof}, a variant of epoch based reclamation is considered and it is proved that 
this variant can be used to safely reclaim memory when certain conditions are satisfied. 
Subsection \ref{trieMemoryReclamation} shows that our implementation satisfies these conditions.

\subsection{Memory Reclamation Schemes}
\label{memoryrelatedwork}
In this section, we introduce the lock-free memory reclamation schemes that are used to 
reclaim memory for the Trie: epoch-based reclamation and reference counting.

\subsubsection{Epoch-based reclamation}
To use traditional epoch-based reclamation, an algorithm must satisfy two requirements. 
First, processes must be able to determine when a record has been permanently 
removed from the data structure.
Second, when any process ends an operation instance on the data structure, it must not 
keep pointers to records that it encountered in the data structure during the instance.
If a record $r$ is permanently removed from the data structure, processes that are still performing 
operation instances when $r$ is removed may still be able to access $r$.
However, once those processes have stopped performing operation instances, it is safe to reclaim $r$.
We call a process quiescent if it is not performing an operation instance on the data structure.
Epoch-based reclamation provides processes with a mechanism for checking if 
every other process has been quiescent since a record has been permanently removed from the data structure,
so they know when it safe to reclaim the record. 

A shared \textit{limbo bag} stores pointers to records that have been permanently 
removed from the data structure.
When any process $p$ permanently removes a record $r$ from the data structure during an operation instance, 
it inserts a pointer to $r$ into the limbo bag.
An execution is divided into segments called \textit{epochs}, which are used to determine when other 
processes have been quiescent.
A \textit{CAS} object \textit{E} stores the current epoch which is initially 0.
Before $p$ begins any operation instance, $I$, on the data structure, it
reads the current epoch $e$ from \textit{E} and announces it in shared memory.
We say that instance $I$ announces $e$.
Next, $p$ will read the announcements of every other process.
If every process has announced $e$, $p$ will perform a \textit{CAS} to try to increment 
$E$ and, if successful, 
will reclaim 
any records that were placed into the limbo bag during instances which announced epoch $e - 2$, where $e \ge 2$.

Consider any record, $r$, that is permanently removed from the data structure during an instance \textit{I} which announces $e$. 
Let $p$ be the process that performs \textit{I} and inserts $r$ into the limbo bag.
At the end of $I$, the value of $E$ can be at most $e + 1$, since $p$ last announced epoch $e$.
Prior to performing an instance $I'$ that is concurrent with \textit{I}, the process $q$ that performs $I'$
announces epoch at most $e + 1$.
At the end of $I'$, the value of $E$ can be at most $e + 2$, since $q$ last announced at most epoch $e + 1$.
Hence, when $E$ is incremented from $e + 2$ to $e + 3$, $I$ and every instance that is concurrent 
with $I$ has ended. 
Therefore it is safe for the process that updates \textit{E} from $e+2$ to $e+3$ to reclaim $r$.

Epoch-based reclamation is not fault tolerant. 
If a single process does not perform any operation instances, or crashes while performing an instance on the data structure,
it will not continue to make announcements. 
This will eventually prevent processes that are performing operations on the data structure from updating the epoch and reclaiming memory.

Brown introduced DEBRA\cite{10.1145/2767386.2767436,DBLP:journals/corr/abs-1712-01044}, which is a distributed form of epoch based reclamation with many improvements.

First, DEBRA ensures that only processes that are performing instances of operations on the data structure
may prevent memory from being reclaimed.
When a process announces the epoch, it will also announce that it is no longer quiescent.
When it finishes an instance, it announces that it is now quiescent.
If a process has announced that it is quiescent, 
it can be treated as though it has announced the current epoch. 
However, if even one process starts an operation instance that it never finishes,
this will still prevent other processes from reclaiming memory.

Second, the scanning of announcements of other processes is amortized over many operation instances.
a process, \textit{p}, will read the announcements of a constant number of processes before each instance of an operation it performs.
Once $p$ has
verified that every other process has either
been quiescent since the current epoch began or announced the current epoch, $p$
may attempt to increment $E$ via \textit{CAS}.
While DEBRA consumes fewer steps per operation, it reclaims memory at a slower pace, since a process 
must perform $\Omega(N)$ operation instances before it attempts to update the epoch.

Finally, rather than using a shared limbo bag, each process uses a sequence of local limbo bags in which they store pointers to the records that 
they remove from the data structure. 
This eliminates the contention that is incurred when using a shared limbo bag.
Typically implementations of DEBRA use 3 limbo bags for each process, 
$\mathit{bag_0}$, $\mathit{bag_1}$ and $\mathit{bag_2}$. 
Initially a process stores pointers to records that it removes from the data structure in $\mathit{bag_0}$. When process $p$ 
updates its announcement to a different epoch, it starts using
the next limbo bag in the sequence (modulo 3) to insert records from the data structure. 
If this next limbo bag is non-empty, 
the records in the limbo bag are first reclaimed.

Brown proved that, for algorithms that use DEBRA with 3 bags, it is safe for an instance 
\textit{I} to place a record $r$ into a limbo bag if $r$
is permanently removed from the data structure by the end of $I$. When $r$ is reclaimed,
the epoch has been incremented at least 3 times since the end of $I$.
Every process must have been quiescent since a new epoch started after the end of 
$I$, which means that every instance that overlaps \textit{I} has terminated.
This ensures that no other processes hold a local pointer to $r$ when it is reclaimed.
In subsection \ref{debraProof}, we will consider DEBRA with more limbo bags,
which is used to reclaim memory used by our implementation of Ko's Trie.

Kim, Brown and Singh\cite{10.1145/3627535.3638491} discovered certain negative performance problems when certain memory allocation algorithms are paired with epoch-based reclamation schemes in which processes may reclaim a large number of records all at once.
They propose instead having each process store pointers to threads it is ready 
to reclaim and reclaiming a small number of these records before each data structure operation, an approach they call \textit{amortized freeing}.
Their experimental results reveal that this approach can significantly improve the performance
and scalability of many lock-free sets with some memory allocation algorithms.
However, they found \textit{amortized freeing} actually reduced 
the performance of algorithms when using the \textit{mimalloc} memory 
allocation algorithm.
Since this is the memory allocation algorithm we use to test our implementation 
of Ko's Trie in later sections, our implementation does not use \textit{amortized freeing}.

\subsubsection{Reference Counting}
Reference counting is a technique in which each record $r$ 
is given an additional field that stores the number of references to $r$
that are stored in shared memory and the local memory of processes.
When the reference count of $r$ drops to zero, it is safe to reclaim $r$.
In some implementations, the number of references to a record $r$ that exist in shared memory 
and local memory is precisely maintained in any configuration.
In other implementations, this reference count field only serves as an upper bound on 
the number of pointers to a record $r$ in shared memory.
Once the count is reduced to zero, no other pointers to $r$ will exist and 
$r$ may be reclaimed.
Reference counting may be used in combination with other memory reclamation schemes.
As we will discuss later in this section, our implementation uses DEBRA with 5 limbo bags, 
and a limited form of reference counting is used to determine when it is safe to insert a DelNode into a limbo bag.

\subsection{b-Bag DEBRA}
\label{debraProof}
We describe a simplified version of DEBRA\cite{10.1145/2767386.2767436,DBLP:journals/corr/abs-1712-01044} in which every process has $b$ local limbo bags, where $b \ge 3$.
Like DEBRA, the current epoch is stored in a \textit{CAS} object $E$, which initially holds the value 0.
Every process announces in shared memory the latest epoch that it has read and whether it is currently quiescent.

\subsubsection{High Level Description}
Before $p_i$ performs each instance of an operation on the data structure, it
calls the \textit{startOp} function,
in which it reads an announcement in shared memory from one other process
and updates its own announcement.
First, $p_i$ will read the current epoch from $E$.
If $p_i$ has never announced this epoch, it will perform an instance of \textit{rotateAndReclaim} 
to switch to using the limbo bag that it least recently used, after reclaiming the records it contains.
Next $p_i$ will read the information announced by some other process $q$.
If it observes $q$ has announced the current epoch or is quiescent, then $p_i$ will read the information announced by another process
before it next performs an operation.
If not, $p_i$ will read the information announced by the same process $q$ before its next operation.
If $p_i$ has verified that every process has announced the current epoch or been quiescent since $E$ was equal to the current epoch, 
it will try to increment $E$ using a \textit{CAS}.
Finally, $p_i$ announces the current epoch and that it is no longer quiescent.

Recall that two operation instances, $A$ and $B$, \emph{overlap} if $A$ does not end before $B$ starts and $B$ does not end before $A$ starts.
The \textit{contention graph of an execution} is an undirected graph whose nodes consist of the instances of 
operations that start during the execution, where there is an edge between every two instances that overlap.
This is used to describe the conditions under which a 
process may place a pointer to a record into a limbo bag.
For example, if the execution consists of these operation instances: 
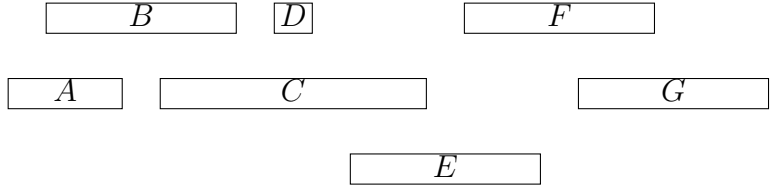
\begin{figure}[H]
    \centering
    \begin{tikzpicture}
    \begin{ganttchart}[inline, time slot format=simple,canvas/.style=%
    {draw=none}]{0}{22}
        \ganttbar{$B$}{2}{6} \ganttbar{$D$}{8}{8}  \ganttbar{$F$}{13}{17}\\
        \ganttbar{$A$}{1}{3} \ganttbar{$C$}{5}{11}  \ganttbar{$G$}{16}{20}\\
         \ganttbar{$E$}{10}{14}\\
    \end{ganttchart}
    \end{tikzpicture}
\end{figure}

\noindent Then its contention graph is:
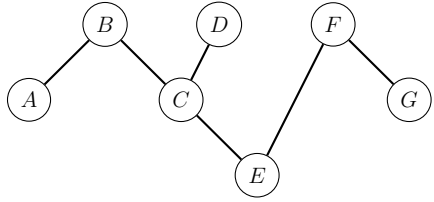
\begin{figure}[H]
    \centering
    \begin{tikzpicture}
        \node[vertex](a) at (0,0){$A$};
        \node[vertex](b) at (1,1){$B$};
        \node[vertex](c) at (2,0){$C$};
        \node[vertex](d) at (2.5,1){$D$};
        \node[vertex](e) at (3,-1){$E$};
        \node[vertex](f) at (4,1){$F$};
        \node[vertex](g) at (5,0){$G$};

        \draw[uEdge](a) -- (b);
        \draw[uEdge](b) -- (c);
        \draw[uEdge](c) -- (d);
        \draw[uEdge](c) -- (e);
        \draw[uEdge](e) -- (f);
        \draw[uEdge](f) -- (g);
    \end{tikzpicture}
\end{figure}

Consider an operation instance $I$ by process $p_i$ in which it encounters a record $r$.
Suppose that $p_i$ knows that, if $r$ is accessed by operation instance $I'$ after the end of $I$, 
then $I'$ is at most distance $d$ from $I$ in the contention graph.
If the number of bags $b$ is at least $d + 2$, $p_i$ may insert $r$ into a limbo bag,
by performing an instance of \textit{reclaimLater(r)},
provided no process has already inserted $r$ into a limbo bag.
We will prove that after $r$ is reclaimed, no process will access $r$.

After $p_i$ finishes each instance of an operation, it calls the \textit{endOp} function
to announce that it is now quiescent.

\subsubsection{Detailed Implementation}
Every process $p_i$ has its own announcement register, $announce_i$, containing two fields, $epoch_i$ and $quiescent_i$, which it uses to announce epochs and whether it is quiescent, respectively.
The value of $quiescent_i$ is initially is equal to \textit{true}. 
It is set to \textit{false} before every instance of an operation on the data structure by $p_i$,
and set back to \textit{true} when $p_i$ completes the instance.
The value of $epoch_i$ is initially 0.

The $j$-th local limbo bag belonging to process $p_i$ is denoted $limboBag_i[j]$, for $0 \le j < b$.
A limbo bag supports the $insert(r)$ and $remove()$ operations.
The $insert(r)$ function inserts a pointer $r$ to a record into the limbo bag.
The $remove()$ function returns $\bot$ if the limbo bag is empty, otherwise a 
record is removed from the limbo bag and a pointer to the record is returned.
In our implementation, every limbo bag is implemented using an dynamically allocated array-based stack.
The local variable $e_i$ is used by $p_i$ to keep track of the epoch it last announced.
If $p_i$ has not performed an instance of \textit{startOp}, $e_i = \bot$.
If $p_i$ has performed at least one instance of \textit{startOp} and is not performing 
\textit{startOp}, $e_i$ is equal to the epoch most recently announced by $p_i$.
We say that $p_i$ has checked process $p_j$ for epoch $e$ if:
\begin{itemize}
    \item $i = j$ and $p_i$ has read $e$ from $E$, or 
    \item $i \neq j$ and after $p_i$ read $e$ from $E$, $p_i$ read $announce_j$ and saw that either $epoch_j = e$ or $quiescent_j = true$.
\end{itemize}
Process $p_i$ uses a local variable $c_i$, where  $1 \le c_i \le N$,
to keep track of the number of processes that it has checked for $e_i$.
In every configuration in which $e_i \neq \bot$, 
\begin{itemize}
    \item $p_i$ has checked $p_{(i + a) \bmod N}$ for $e_i$, where $0 \le a < c_i$, and
    \item if $c_i < N$, then $p_i$ will next read the announcement register of process $p_{(i + c_i) \bmod N}$ (to try to check it for $e_i$).
\end{itemize}
When $c_i$ is equal to $N$, $p_i$ can increment $E$ from $e_i$ to $e_i + 1$.
The local variable $bag_i$ stores the index of the limbo bag that $p_i$ is currently using.
Initially, the value of $bag_i$ is $b - 1$. The value of $bag_i$ is incremented (modulo $b$) when $p_i$ 
reads an epoch that it has not previously read from $E$.
When this happens, we say that $p_i$ \textit{rotates its limbo bags}. 
Following this rotation, if $limboBag_i[bag_i]$ is not empty, 
the records it contains are removed and reclaimed.
A record that is placed into a limbo bag by $p_i$ will be reclaimed after $p_i$ rotates its limbo bags $b$ times. 

To perform an instance of \textit{startOp}, process $p_i$ first reads the current value of the epoch $e$ from $E$. 
Then $p_i$ checks if $e = e_i$ to determine if it had previously announced $e$.
If not, it performs an instance of \textit{rotateAndReclaim()},
updates $c_i$ to 1 and updates $e_i$ to $e$.
Next, $p_i$ reads from the announcement register of $p_{(i + c_i) \bmod N}$.
If that process was quiescent or it has announced $e$, then $p_i$ will increment $c_i$.
Following this increment, if $c_i$ is equal to $N$, $p_i$ will perform a \textit{CAS} to attempt to increment $E$ from $e$ to $e+1$.
Finally, $p_i$ writes $\twoField{e}{false}$ to $announce_i$, indicating the largest 
epoch it has read from $E$ and that it is no longer quiescent.
\begin{figure}[H] 
\begin{algorithmic}[1]
\alglinenoNew{algDebra}
    \State $startOp()$ by process $p_i$
    \Indent
        \State $e \gets E.read()$ \label{eRead}
        \If {$e \neq e_i$} \label{check_e} \Comment{New epoch since $p_i$ last read $E$.}
            \State $rotateAndReclaim()$ \label{eDifferentCheckEnd}
            \State $c_i \gets 1$ \label{resetCheckNext}
            \State $e_i \gets e$ \label{set_e_i}
        \EndIf
        \State $id \gets (i + c_i) \bmod N$
        \State $\twoField{e'}{q'} \gets announce_{id}.read()$ \label{otherRead}
        \If {$e' = e$ or $q' = true$} \label{checkOther}\Comment{$p_{id}$ is quiescent or has announced $e$}
            \State $c_i \gets c_i + 1$ \label{updateCheckNext}
            \If {$c_i = N$} \Comment{All processes announced $e$ or quiescent since $E=e$.}
                \State $E.CAS(e, e + 1)$ \label{cas_E}
            \EndIf
        \EndIf 
        \State $announce_i.write(\twoField{e}{false})$ \label{announce_e}
    \EndIndent   
\alglinenoPush{algDebra}
\end{algorithmic}
\caption{Function performed by a process $p_i$ before starting an instance of a data structure operation.}
\label{startOp}
\end{figure}

The \textit{rotateAndReclaim()} function changes the current limbo bag used by $p_i$
to the one it least recently used, by updating $bag_i$ to $(bag_i + 1) \bmod b$.
Following this update, if $limboBag_i[bag_i]$ is not empty, every record 
inside of it is removed and reclaimed.
\begin{figure}[H] 
\begin{algorithmic}[1]
\alglinenoPop{algDebra}
    \State $rotateAndReclaim()$ by process $p_i$
    \Indent
        \State $bag_i \gets (bag_i + 1) \bmod b$ \Comment{Move to next limbo bag.}
        \State $record \gets limboBag_i[bag_i].remove()$
        \While{$record \neq \bot$} \Comment{Reclaim every record in $\mathit{limboBag_i[bag_i]}$.}
            \State Reclaim $record$
            \State $record \gets limboBag_i[bag_i].remove()$
        \EndWhile
    \EndIndent   
\alglinenoPush{algDebra}
\end{algorithmic}
\caption{Changes $p_i$'s limbo bag to be the one it has least recently used.}
\label{rotateAndReclaim}
\end{figure}

The \textit{reclaimLater(r)} function is used by a process 
to insert a pointer $r$ to a record into its current limbo bag so that the record may later be reclaimed.
\begin{figure}[H] 
\begin{algorithmic}[1]
\alglinenoPop{algDebra}
    \State $reclaimLater(r)$ by process $p_i$
    \Indent
        \State $limboBag_i[bag_i].insert(r)$
    \EndIndent   
\alglinenoPush{algDebra}
\end{algorithmic}
\caption{Function used by a process $p_i$ to place a pointer $r$ to a record into a limbo bag.}
\label{reclaimLater}
\end{figure}

To perform an instance of \textit{endOp}, process $p_i$ performs a single write to $announce_i$ to
update $quiescent_i$ to $true$, leaving $epoch_i$ (which is equal to $e_i$) unchanged.
\begin{figure}[H] 
\begin{algorithmic}[1]
\alglinenoPop{algDebra}
    \State $endOp()$ by process $p_i$
    \Indent
        \State $announce_i.write(\twoField{e_i}{true})$
    \EndIndent   
\alglinenoPush{algDebra}
\end{algorithmic}
\caption{Function performed by a process $p_i$ after completing an instance of a data structure operation.}
\label{endOp}
\end{figure}

\subsubsection{Proof of Correctness}

Consider an implementation of a data structure using $b$-bag DEBRA 
such that every record is inserted into at most one limbo bag and is inserted into it at most once.
Since any record is removed from a limbo bag and reclaimed at most once, this 
implies no record is reclaimed more than once.

Furthermore, suppose that, in this implementation, after a record is inserted into a limbo bag by an operation instance $I$, it is only accessed by operation instances that are at most distance $b-2$ from $I$ in the contention graph.
We prove that no record is accessed after it is reclaimed.

First, we prove that if a process increments the epoch, then every process was quiescent at some point since the epoch started.
Next, we prove that, for any $k \geq 1$, if the epoch is incremented at least $k$ times after the end of an operation instance $I$, 
then every operation instance at most distance $k - 1$ from $I$ in the contention graph has terminated.
Then we prove that, when a process reclaims a record $r$, at least $b - 1$
increments of the epoch have taken place since the end of the operation instance $I$ in which this process inserted $r$ into a limbo bag.
Finally, from these results we show that no process accesses $r$ 
after it is reclaimed.

Suppose process $p_i$ increments the epoch. 
To show that every process has been quiescent at some point since the epoch started, we show that 
$p_i$ has checked every process for the current epoch and if $p_i$ has checked some process for an epoch, then
the process has been quiescent at some point since the epoch started.

\begin{lemma} \thlabel{quiescentSince}
    If $p_i$ has checked $p_j$ for epoch $e$, then $p_j$ was quiescent at some point since $E$ was first equal to $e$.
\end{lemma}
\begin{proof}
    Process $p_i$ is quiescent when it reads $e$ from $E$.
    Now suppose $p_i$ has checked $p_j$ for $e$, where $j \neq i$.
    Then, after $p_i$ read $e$ from $E$ for the first time, 
    $p_i$ read $announce_j$ and saw that either $quiescent_j = true$ or $epoch_j = e$.
    In the first case, $p_j$ was quiescent when $p_i$ read that $quiescent_j = true$.
    In the second case, $p_j$ was quiescent between when it read $e$ from $E$ and when it wrote $e$
    to $announce_j$.
\end{proof}

\begin{lemma} \thlabel{checkedProcesses}
    If $e_i \neq \bot$, $p_i$ has checked $p_{(i + a) \bmod N}$ for $e_i$, for $0 \le a < c_i$.
\end{lemma}
\begin{proof}
    We prove the claim by induction on the sequence of local steps performed by process $p_i$.
    Initially, $e_i = \bot$ so the claim is true.
    Process $p_i$ only updates $e_i$ to values it reads from $E$ and $E$ always contains non-negative integers.
    Therefore, once $e_i \neq \bot$, this remains the case for the rest of the execution.
    If a local step does not change $c_i$ or $e_i$ then the claim still holds.
    Only the local steps performed by $p_i$ on lines \ref{resetCheckNext}, \ref{set_e_i} and
    \ref{updateCheckNext} of \textit{startOp} may change $c_i$ or $e_i$.
    Suppose the claim is true before one of these local steps.

    Suppose that $c_i$ is updated to $1$ on line \ref{resetCheckNext}.
    If $e_i = \bot$ then the claim holds.
    If $e_i \neq \bot$, process $p_i$ must have read $e_i$ from $E$ during its previous instance of \textit{startOp}.
    Therefore $p_i$ has checked $p_i$ for $e_i$
    and the claim holds with $c_i = 1$.

    Suppose that $e_i$ is updated to $e$ on line \ref{set_e_i}.
    Then on the previous line, $c_i$ was updated to $1$.
    Since $p_i$ read $e$ from $E$ on line \ref{eRead}, it has checked $p_i$ for $e$.
    Therefore the claim holds with $c_i = 1$.

    Suppose that $c_i$ is incremented from $c$ to $c+1$ on line \ref{updateCheckNext}.
    By the test on line \ref{check_e} and the assignment on line \ref{set_e_i}, $e_i$ 
    is equal to the epoch $e$ that $p_i$ read from $E$.
    Since the test on line \ref{checkOther} was true, $p_i$ has checked $p_{(i + c) \bmod N}$
    for $e_i$.
    By the induction hypothesis, prior to line \ref{updateCheckNext}, $p_i$ had checked $p_{(i + a) \bmod N}$ for $e_i$, for $0 \le a < c$.
    Therefore, following the increment on line \ref{updateCheckNext}, the claim holds for $c_i = c + 1$.
    
    Therefore, the claim always holds.
\end{proof}

Suppose process $p_i$ performs a \textit{CAS} that successfully increments the epoch from $e$ to $e + 1$ on line \ref{cas_E} of \textit{startOp}.
Since $c_i = N$ on the previous line, \thref{checkedProcesses} implies $p_i$ has 
checked every process for $e_i$ prior to the \textit{CAS}.
Note that $e_i$ must be equal to $e$,
by the test on line \ref{check_e} and the assignment on line \ref{set_e_i}.
Therefore, $p_i$ has checked every process for $e$ prior to the \textit{CAS}.
By \thref{quiescentSince},
this implies that every process was quiescent at some point while $E$ was equal to $e$.
\begin{corollary} \thlabel{checkedEveryProcess}
    If $E$ is incremented from $e$ to $e + 1$, then every process was quiescent at some point while $E$ was equal to $e$.
\end{corollary}

Using \thref{checkedEveryProcess}, we can now prove the following lemma.
\begin{lemma} \thlabel{otherInstancesTerminated}
    Suppose that, during the execution, $E$ is incremented at least $k \ge 1$ 
    times after the end of some operation instance $I$.
    Then every instance that is at most distance $k - 1$ from $I$ in the contention graph ended before the $k$-th of these increments occurred.
\end{lemma}
\begin{proof}
    We prove the claim by induction on $k$. 
    Let $C_{k}$ be the configuration immediately following the $k$-th increment of $E$ after the end of $I$.
    
    Consider the base case where $k = 1$.
    Then $C_{k}$ is the configuration immediately following the first increment of $E$ that occurs after the end of $I$.
    The only instance that is at distance 0 from \textit{I} in the contention graph is \textit{I} itself.
    Therefore the claim is true.
    
    Let $k \ge 2$, suppose $E$ was incremented at least $k$ times after the end of $I$, and suppose
    the claim holds for $k - 1$.
    Consider an arbitrary instance $Z$ that is at most distance $k - 1$ from $I$ in the contention graph.
    If $Z$ is less than distance $k - 1$ from $I$ in the contention graph, then by the induction
    hypothesis, $Z$ ended before $C_{k - 1}$.
    So, suppose $Z$ is exactly distance $k - 1$ from $I$ in the 
    contention graph.
    Then, by definition of the contention graph, 
    $Z$ must overlap with some instance $Z'$
    that is exactly distance $k - 2$ from $I$. 
    By the induction hypothesis, $Z'$ ended before $C_{k - 1}$.
    Since $Z$ and $Z'$ overlap, $Z'$ did not end before $Z$ started.
    Therefore $Z$ started before $C_{k - 1}$.
    Let $e$ be the value of $E$ in $C_{k - 1}$.
    By \thref{checkedEveryProcess}, before $E$ was incremented to $e + 1$, every process must have been quiescent at some point while $E$ was equal to $e$.
    Since $C_{k - 1}$ is the first configuration in which $E = e$, the process that performed $Z$
    was quiescent at some point between $C_{k - 1}$ and $C_{k}$.
    This implies that the process finished performing $Z$ before $C_{k}$.
    So the claim holds for $k$.
\end{proof}

Recall that, when using DEBRA with $b$-bags, a process that inserts a record into a limbo bag during some operation instance 
rotates its limbo bags $b$ times between the end of the operation instance and when the record is reclaimed.
\begin{lemma} \thlabel{b_min_one_incrs}
    Suppose a record, $r$, is inserted into a limbo bag during an operation instance $I$.
    If $r$ is reclaimed during the execution, $E$ was incremented at least $b - 1$ times 
    between the end of $I$ and when $r$ is reclaimed.
\end{lemma}
\begin{proof}
    Let $p_i$ be the process that performs $I$ and later reclaims $r$.
    Between the end of $I$ and when $r$ is reclaimed, $p_i$ rotates its limbo bags $b$ times.
    Immediately before rotating its limbo bags, $p_i$ reads a value from $E$
    that it has not previously read.
    Consider the first step after the end of $I$ in which $p_i$ reads a value from $E$ that it has not previously read.
    Let $C$ be the configuration immediately following the step.
    Between $C$ and when $r$ is reclaimed, $p_i$ reads $b - 1$ new values from $E$. 
    Since processes only increment $E$, this implies that $E$ was incremented at least $b - 1$ times between $C$ and when $r$ is reclaimed.
\end{proof}

\begin{lemma} \thlabel{b_debra_instances_terminated}
    Suppose a record, $r$, is inserted into a limbo bag during an operation instance, $I$.
    If $r$ is reclaimed during the execution, every instance at most distance $b - 2$ from $I$ 
    ends before $r$ is reclaimed.
\end{lemma}
\begin{proof}
    By \thref{b_min_one_incrs}, if $r$ is reclaimed, then at least $b - 1$ increments
    of $E$ take place between the end of $I$ and when $r$ is reclaimed.
    By \thref{otherInstancesTerminated}, this implies that every instance at most distance 
    $b - 2$ from $I$ in the contention graph has ended before $r$ is reclaimed.
\end{proof}

\begin{theorem} \thlabel{b_debra_safe}
    Consider a record, $r$, that is inserted into a limbo bag during an operation instance, $I$.
    Suppose that, after $r$ is inserted into a limbo bag, it is accessed by operation instances that are at most distance $b - 2$ from $I$. 
    Then $r$ is not accessed after it is reclaimed.
\end{theorem}
\begin{proof}
    By \thref{b_debra_instances_terminated}, before $r$ is reclaimed, every instance at most distance $b - 2$ from $I$ in the contention graph
    has ended.
    Therefore, $r$ is not accessed after it is reclaimed.
\end{proof}

\subsection{Memory Reclamation in our Implementation of Ko's Trie}
\label{trieMemoryReclamation}

Our implementation of Ko's Trie uses DEBRA with 5 limbo bags to reclaim memory.
Dynamically allocated records that will no longer be accessed in our implementation
are inserted into limbo bags.
These records are eventually reclaimed by our implementation,
provided every process that performs an operation instance eventually completes 
this instance.
Since the implementation is not wait-free, 
there are many executions in which processes could take an unbounded number of 
steps without completing an operation instance that they invoked.
In such executions, processes cannot reclaim memory.
Since the implementation is lock-free, such executions can only occur 
if some process $p$ starts an operation instance and either:
\begin{itemize}
    \item $p$ stops performing steps before completing this operation instance, or
    \item despite $p$ performing steps, other processes continue to 
    complete operation instances and delay $p$ from completing its operation instance.
\end{itemize}

In this section, we show that once a PredecessorNode or an UpdateNode is inserted into a limbo bag,
it is only accessed by operation instances at most distance 3 in the contention graph 
from the instance that inserted it.
This allows us to show that after a record is reclaimed, it is not accessed. 
We address how NotifyNodes are reclaimed in our implementation later.

First, we show that if an operation instance encounters a PredecessorNode or an UpdateNode that announces another
operation instance, then these instances overlap.

\begin{lemma} \thlabel{pall_overlaps_pOp}
    Consider a PredecessorNode, \textit{pNode}, that is created by an operation instance, $pOp$.
    Suppose an operation instance, $aOp$, encounters \textit{pNode} while traversing the \textit{PALL}.
    Then $aOp$ overlaps $pOp$.
\end{lemma}
\begin{proof}
    Before $aOp$ encounters \textit{pNode}, $pOp$ inserted \textit{pNode} into the \textit{PALL} and $pOp$ did not remove \textit{pNode} from the \textit{PALL} before $aOp$ started traversing the \textit{PALL}.
\end{proof}

\begin{lemma} \thlabel{uall_ruall_overlaps_uOp}
    Consider an UpdateNode, \textit{uNode}, that is owned by an operation instance, $uOp$.
    Suppose an operation instance, $aOp$, encounters \textit{uNode} while traversing the \textit{UALL} or \textit{RUALL}. 
    Then $aOp$ overlaps $uOp$.
\end{lemma}
\begin{proof}
    Before \textit{uNode} was inserted into the \textit{UALL} or \textit{RUALL}, $uOp$ inserted \textit{uNode} into a LatestList.
    Therefore $aOp$, which encounters \textit{uNode} in the \textit{UALL} or \textit{RUALL}, does not finish before $uOp$ starts.
    
    If \textit{uNode} had been removed from either the \textit{UALL} or the \textit{RUALL} by $uOp$ before $aOp$ started traversing it, 
    $aOp$ would not have encountered \textit{uNode} during its traversal.
    Since \textit{uNode} was in one of the lists after $aOp$ started traversing, 
    $uOp$ does not finish before $aOp$ starts.
    Thus, $aOp$ overlaps $uOp$.
\end{proof}

We use these lemmas to prove that, after a PredecessorNode, \textit{pNode}, is inserted into a 
limbo bag
by an operation instance, $pOp$, it is only accessed by instances that overlap with $pOp$.
Note that $pOp$ 
was operation instance that 
created \textit{pNode}.


\begin{lemma} \thlabel{distance1lemma}
    Any operation instance that accesses \textit{pNode} after it is inserted into a limbo bag 
    is at distance 1 from the $pOp$ in the contention graph.
\end{lemma}
\begin{proof}
Suppose an operation instance $aOp \neq pOp$ obtains a pointer to \textit{pNode}.
If it encounters \textit{pNode} while traversing the \textit{PALL},
by \thref{pall_overlaps_pOp}, $aOp$ overlaps $pOp$.
Otherwise, while $aOp$ traverses the \textit{RUALL} it encounters a DelNode, $dNode$, such that $dNode$ is owned by $pOp$
and $dNode.delPredNode$ points to \textit{pNode}.
By \thref{uall_ruall_overlaps_uOp}, $aOp$ overlaps $pOp$.
\end{proof}

\medskip

First we consider the circumstances under which an UpdateNode is inserted into a limbo bag, 
and then we will prove that, once in a limbo bag, it is only accessed by instances at most distance 3
from the instance that inserted it into the limbo bag.
The only time an InsertNode is inserted into a limbo bag is immediately following a \textit{SWAP} 
step which removes it from a LatestList.

Every DelNode, \textit{dNode}, has a field, \textit{dCount}, which is a \textit{fetch-and-add} object.
If \textit{dNode} is in the Trie in the initial configuration, its \textit{dCount} is initially the number of TrieNodes that point to it plus one for the pointer to it from $\mathit{latest[dNode.key]}$. 
If \textit{dNode} is inserted into a LatestList by an operation instance, \textit{dOp}, its \textit{dCount} is 
initialized to two, one for the pointer to \textit{dNode} from $\mathit{latest[dNode.key]}$ and 
one since \textit{dOp} may still update TrieNodes to point to \textit{dNode}.
Immediately before trying to update a TrieNode to point to \textit{dNode}, \textit{dOp} will increment \textit{dNode.dCount}.
If \textit{dOp} fails to update this TrieNode to point to \textit{dNode}, \textit{dOp} will decrement \textit{dNode.dCount}, cancelling out this increment.
When \textit{dOp} will no longer update any TrieNodes to point to \textit{dNode}, it decrements \textit{dNode.dCount}.
Immediately after a process removes \textit{dNode} from the LatestList for $dNode.key$ or updates a TrieNode which pointed to \textit{dNode} to point to another DelNode, it decrements \textit{dNode.dCount}.
If a decrement by a process lowers \textit{dNode.dCount} to 0, this process will insert \textit{dNode} into its current 
limbo bag.

If \textit{dNode} is in the Trie in the initial configuration, its \textit{dCount} is non-increasing, so at most one process will lower its \textit{dCount} to 0.
When this occurs, \textit{dNode} has been removed from the LatestList for $dNode.key$ and no TrieNode points to \textit{dNode}.
Suppose, instead, \textit{dNode} was inserted into a LatestList by \textit{dOp}.
Then \textit{dNode.dCount} is at least the number of TrieNodes pointing to \textit{dNode}, plus one 
if \textit{dNode} is in the LatestList for $dNode.key$, plus one if \textit{dOp} has not finished trying to 
update TrieNodes to point to \textit{dNode}.
Once \textit{dOp} will no longer update any other TrieNodes to point to \textit{dNode}, \textit{dNode.dCount} is non-increasing,
so at most one process will lower it to 0.
Hence, at most one process will insert \textit{dNode} into a limbo bag.

\begin{observation} \thlabel{no_latest_trie}
    After an UpdateNode, \textit{uNode}, is inserted into a limbo bag, 
    no LatestList contains \textit{uNode} and no TrieNode points to \textit{uNode}. 
\end{observation}

\begin{lemma} \thlabel{distance3lemma}
Any operation instance that accesses an UpdateNode, \textit{uNode},
after it has been inserted into a limbo bag
is at most distance 3 in the contention graph from the operation instance that inserted
\textit{uNode} into the limbo bag.
\end{lemma}
\begin{proof}
    Suppose that an operation instance, $qOp$, inserts \textit{uNode} into a limbo bag.
    Let $aOp$ be an operation instance that later accesses \textit{uNode}.
    Then $aOp$ does not end before $qOp$ starts.
    A pointer to \textit{uNode} may be obtained by $aOp$ in one of these ways:
    \begin{itemize}
        \item while traversing a LatestList,
        \item while accessing a TrieNode,
        \item while traversing the \textit{UALL} or \textit{RUALL},
        \item reading the \textit{ruallPosition} of a PredecessorNode it encounters while traversing the \textit{PALL}, or
        \item from the \textit{updateNode} field of a NotifyNode, while traversing the \textit{notifyList} of a PredecessorNode 
        that it encounters while traversing the \textit{PALL}.
    \end{itemize}

    If $qOp$ does not end before $aOp$ starts,
    then $aOp$ overlaps $qOp$.
    Thus, they are at most distance 1 from one other in the contention graph.
    By \thref{no_latest_trie}, this is the case if $aOp$ obtains a pointer to \textit{uNode} while traversing
    a LatestList or while accessing a TrieNode.
    
    For the rest of the proof, suppose $qOp$ ends before $aOp$ starts.
    Note that a pointer to an UpdateNode which is initially in the Trie may only be obtained while traversing a LatestList or accessing a TrieNode.
    Hence, in the remaining three cases, we may assume that \textit{uNode} is owned by an operation instance, $uOp$. 
    
    Suppose $aOp$ obtains a pointer to \textit{uNode} while traversing the \textit{UALL} or \textit{RUALL}.
    By \thref{uall_ruall_overlaps_uOp}, $aOp$ overlaps $uOp$. 
    By \thref{no_latest_trie}, $uOp$ inserts \textit{uNode} into the LatestList before $qOp$ inserts it into a limbo bag, thus
    $uOp$ starts before $qOp$ ends. 
    Since $uOp$ overlaps $aOp$, which starts after $qOp$ ends, this implies that
    $uOp$ overlaps $qOp$. 
    Hence, $aOp$ is at most distance 2 from $qOp$ in the contention graph.

    Suppose $aOp$ obtains a pointer to \textit{uNode} while reading the \textit{ruallPosition} of a PredecessorNode, \textit{pNode}, 
    that it accesses while traversing the \textit{PALL}.
    By \thref{pall_overlaps_pOp}, $aOp$ overlaps the operation instance, $pOp$, that 
    created \textit{pNode}.
    Since $aOp$ reads \textit{uNode} from \textit{pNode.ruallPosition}, $pOp$ encountered \textit{uNode} 
    while traversing the \textit{RUALL}.
    So, by \thref{uall_ruall_overlaps_uOp}, $pOp$ overlaps $uOp$.
    Before \textit{uNode} is inserted into the \textit{RUALL}, $uOp$ inserts it into a 
    LatestList.
    By \thref{no_latest_trie}, after $qOp$ inserts \textit{uNode} into a limbo bag,
    no LatestList contains \textit{uNode}, so $uOp$ put \textit{uNode} into this LatestList before $qOp$ put \textit{uNode} into a limbo bag.
    If $qOp$ overlaps $uOp$, then $aOp$ is at most distance 3 from $qOp$, 
    since $aOp$ overlaps $pOp$ and $pOp$ overlaps $uOp$.
    Otherwise, $uOp$ ends before $qOp$ starts.
    Since $qOp$ ends before $aOp$ starts and $pOp$ overlaps both $uOp$ and $aOp$,
    it follows that $pOp$ overlaps $qOp$.
    Hence, $aOp$ is at most distance 2 from $qOp$.

    Suppose $aOp$ obtains a pointer to \textit{uNode} from the \textit{updateNode} field of a NotifyNode, \textit{nNode}, while traversing the \textit{notifyList} of a PredecessorNode, \textit{pNode},
    that it accesses while traversing the \textit{PALL}.
    Let $nOp$ be the operation instance which inserted \textit{nNode} into \textit{pNode.notifyList} (while traversing the \textit{PALL}).
    By \thref{pall_overlaps_pOp}, both $aOp$ and $nOp$ overlap the operation instance, $pOp$, that 
    creates \textit{pNode}.
    Recall that $uOp$ inserted \textit{uNode} into a LatestList.
    If $nOp \neq uOp$, then $nOp$ encounters \textit{uNode} at the head of this LatestList.
    By \thref{no_latest_trie}, after $qOp$ inserts \textit{uNode} into a limbo bag, 
    no LatestList contains \textit{uNode}.
    Thus, $nOp$ starts before $qOp$ put \textit{uNode} into a limbo bag.
    If $nOp$ overlaps $qOp$, then $aOp$ is at most distance 3 
    from $qOp$, since $aOp$ overlaps $pOp$ and $pOp$ overlaps $nOp$.
    Otherwise, $nOp$ ends before $qOp$ starts. 
    Since $qOp$ ends before $aOp$ starts and $pOp$ overlaps both $nOp$ and $aOp$,
    it follows that $pOp$ overlaps $qOp$.
    Hence $aOp$ is at most distance 2 from $qOp$.

    
    Therefore, in all of these cases, $aOp$ is at most distance 3 from $qOp$ in the contention graph.
\end{proof}

Now we can show that our memory reclamation scheme works properly.
\begin{theorem} \thlabel{UpdatePredecessorReclamationSafe}
    PredecessorNodes and UpdateNodes are not accessed after they are reclaimed.
\end{theorem}
\begin{proof}
    Consider an UpdateNode or a PredecessorNode, $node$, that is inserted into a limbo bag by an operation instance, $qOp$.
    By \thref{distance1lemma} or \thref{distance3lemma}, after $node$ is inserted into a limbo bag, 
    it is only accessed by operation instances that are at most distance 3 from $qOp$ in the contention graph.
    Therefore, by \thref{b_debra_safe} with $b = 5$, no process accesses $node$ after it is reclaimed.
\end{proof}

In our implementation, NotifyNodes are not inserted into limbo bags.
Rather, when reclaiming a PredecessorNode, a process 
will reclaim each of the NotifyNodes in the PredecessorNode's \textit{notifyList}.
By \thref{UpdatePredecessorReclamationSafe}, no NotifyNodes are inserted into 
the PredecessorNode's \textit{notifyList} after it is reclaimed.
\begin{theorem}
    NotifyNodes are not accessed after they are reclaimed.
\end{theorem}
\begin{proof}
    Let \textit{nNode} be a NotifyNode that is reclaimed during the execution.
    Let $pOp$ be the operation instance that creates the PredecessorNode, \textit{pNode}, 
    into whose $notifyList$ \textit{nNode} is inserted.
    Suppose an operation instance, $aOp$, accesses \textit{nNode} after it was inserted into \textit{pNode.notifyList}.
    In order to access \textit{nNode}, $aOp$ must encounter \textit{pNode} while traversing the \textit{PALL}.
    Thus, by \thref{pall_overlaps_pOp}, $aOp$ is at most distance 1 from $pOp$.
    Recall that $pOp$ is the operation instance that inserts \textit{pNode} into a limbo bag.
    By \thref{b_debra_instances_terminated} with $b = 5$, every instance that is at most distance 3 from $pOp$ 
    ends before \textit{pNode} is reclaimed, so $aOp$ ends before \textit{pNode} is reclaimed.
    Since \textit{nNode} is reclaimed while reclaiming \textit{pNode}, $aOp$ ends before \textit{nNode} is reclaimed as well.
\end{proof}

\chapter{An Experimental Comparison of Dynamic Ordered Set Implementations}
\label{experimentChapter}
This chapter concerns the practical experiments we performed to evaluate the performance 
of implementations of lock-free dynamic ordered sets.
First, we provide the specification of the machine we ran our experiments on 
and give a simplified description of its computational model.
Next, we discuss the implementations we compared against our implementation of Ko's Trie.
Then we explain our experimental procedure in detail and the experiments that we performed.
Finally, we present the results of our experiments and what we learned from them.

\section{Machine Specification}
We performed our experiments on the \href{https://mc.uwaterloo.ca/facilities.html}{Nasus} machine of the 
University of Waterloo's Multicore Lab.
The machine has two sockets, and each socket contains an AMD EPYC 7662 CPU.
Each CPU contains 64 cores and every core has two hardware threads.
The main memory on the system consists of 256GB of DDR4 RAM.
The operating system of the machine is Ubuntu 18.04 LTS with Linux kernel version 5.4.

On this system, a process is an instance of a program that has started but not yet terminated and 
the number of processes running on the machine can change over time.
Every hardware thread can have a single process assigned to execute upon it.
The operating system of the machine has a program called a scheduler,
which assigns processes to execute on hardware threads.
Processes that are assigned to hardware threads perform their instructions in parallel, asynchronously.
If the number of processes executing on the system exceeds the number of hardware threads on the system,
processes not assigned to a hardware thread will have to wait for the scheduler to 
give them a turn to execute.
The scheduler may stop a process assigned to a hardware thread
from running and assign another process to the hardware thread
in its place.
When this happens, the state of the running process is stored into memory 
whereas the state of the process which will now run is loaded into the local memory of the hardware thread.
This storing and loading of process states is called a \textit{context switch}.
The scheduler may, unless otherwise instructed, also reassign a process from executing on one hardware thread to another, which also requires a context switch.
We can command the scheduler to assign certain processes to certain hardware threads, 
which reduces context switches.
This is called \textit{thread-pinning}; we say that the processes have been \textit{pinned} to these hardware threads.
We will discuss other benefits of thread-pinning later.
A process that is running may create other processes and also ask the system 
whether other processes have finished.
When a new process is created, the process that created it specifies the program 
the new process should run and parameters to the program.

A hardware thread has a fixed number of 8 byte local registers which 
it may read and update nearly instantaneously.
A process assigned to a hardware thread may read and modify words in shared 
memory.
While accesses to memory are slow, words in memory vastly outnumber 
the number of local registers a hardware thread has.
Accessing memory also allows a process to communicate with other processes.
The native word size on the system is 8 bytes, but words of 1, 2 or 4 bytes may also be accessed.
Processes assigned to hardware threads have access to the following atomic primitives, where $x1$ and $x2$ are local registers and 
$p$ is the address of a word in memory:
\begin{itemize}
    \item \textit{Read($p$, $x1$)}: Read the word located at $p$ into $x1$,
    \item \textit{Write($p$, $x1$)}: Write $x1$ into the word located at $p$,
    \item \textit{Fetch-And-ADD($p$, $x1$)}: 
    Atomically compute the sum of $x1$ and the word located at $p$,
    storing it in this word. Returns the value of the word
    immediately before the primitive was performed.
    \item \textit{Fetch-And-AND($p$, $x1$)}: Atomically compute the 
    bitwise AND of $x1$ and the word located at $p$,
    storing it in this word. Returns the value of the word 
    immediately before the primitive was performed.
    \item \textit{SWAP($p$, $x1$)}: Atomically write $x1$ into the word 
    at $p$. Returns the value of the word immediately before the primitive was performed.
    \item \textit{CAS($p$, $x1$, $x2$)}: Atomically check if the word at 
    $p$ is equal to $x1$ and if it is, write $x2$ to the word.
    Returns the value of the word immediately before the primitive was performed.
\end{itemize}
A process may request the system to dynamically allocate a contiguous segment of words in memory, receiving a pointer to the first byte in the allocated segment, 
or a value indicating that the system was out of memory to allocate.
When a dynamically allocated segment will no longer be accessed, a 
process may indicate to the system that it may reclaim the segment.
Every process may read the current time 
into one of its registers from a globally accessible clock.

The cache model used by the machine is similar to the cache coherency with write back model:
Main memory is evenly divided into \textit{cache lines} of 64 bytes.
Accessing main memory can be slow, so every hardware thread has access to a series of caches,
which each store a limited number of copies of cache lines from main memory.
These caches are progressively larger but also more remote from the hardware thread
and thus slower for it to access.
Every CPU core has an L1 cache of 32KB and an L2 cache of 512KB. 
The L1 cache of every CPU core is divided evenly into private local caches for the two hardware threads on the core, 
such that every hardware thread has a local cache of 16KB.
The L2 cache of this core is shared between the two hardware threads on the core.
The hardware threads of every core also share an L3 cache of 16MB with the hardware threads of three other 
nearby cores on the same CPU.
The cache lines stored by a hardware thread's local cache may either be in 
\textit{shared} or \textit{exclusive} mode.
To read a word in memory, 
a process must have a shared or exclusive copy of the cache line 
that contains the word in its local cache.
To modify this word, a process must have an 
exclusive copy of this cache line in its local cache.
A \textit{cache miss} happens whenever a process fails to locate a cache line it needs in one of its caches,
forcing the process to access a more remote cache or main memory in order 
to load a copy of this cache line into this cache.
If a cache is full when a new cache line is being loaded into it, 
the least recently accessed cache line in the cache will be overwritten with the new one.
When a process obtains an exclusive copy of a cache line, 
it invalidates copies of the cache line in the caches of other processes,
effectively erasing them.
When a process is assigned to a hardware thread, its local cache is empty
and the process will have to access more remote caches or main memory 
in order to access any words in memory.
Pinning this process to this hardware thread will ensure that the process 
does not have to start from an empty cache again.

Suppose a \textit{cache miss} occurs in the local cache of a process, $p$, as it tries to read a word.
First, if another process, $q$, has an exclusive copy of the cache line containing the word, 
this copy is downgraded to a shared copy and, if modified, is written back to $q$'s caches and memory.
Then the cache line containing the word is loaded from a more remote cache 
or main memory into $p$'s caches that are missing the cache line, 
including $p$'s local cache in shared mode.

Suppose a \textit{cache miss} occurs in the local cache of $p$ as it tries to modify a word.
First, if another process, $q$, has an exclusive copy of this cache line, this copy, if modified, 
is written back to $q$'s caches and main memory. 
Next, all copies of the cache line that do not match the cache line in main memory are invalidated, 
as well as all copies that are in local caches or remote caches that $p$ does not have access to.
Finally, this cache line is loaded from a more remote cache 
or main memory into $p$'s caches that are missing the cache line, 
including $p$'s local cache in exclusive mode.

This machine uses a Non-Uniform Memory Access (NUMA) architecture, 
in which the performance of a memory access by a process depends on 
the memory location relative to the hardware thread the process is assigned to.
Accesses by the process to memory that is more local to this hardware thread 
will have significantly better performance compared to memory that is remote to it. 
Suppose two processes, $p$ and $q$, repeatedly modify words that are in the 
same cache line.
If $p$ and $q$ are assigned on hardware threads on the same core, 
they will share an L2 cache, thus
neither experiences L2 cache misses following their first accesses of 
the cache line.
If $p$ and $q$ are assigned to hardware threads of distinct cores which 
share an L3 cache, neither will experience L3 cache misses 
following their first accesses of the cache line.
However, if they are assigned to hardware threads on 
distinct CPUs or cores that do not share an L3 cache on the same CPU, 
$p$ and $q$ will continue to experience L3 cache misses.
Pinning these processes to nearby hardware threads can prevent 
them from experiencing expensive cache misses.


Many distinct words in memory may be stored in a single cache line.
Suppose two processes, $p_0$ and $p_1$, repeatedly access two distinct words, $A$ and $B$, in the same cache line.
If $p_0$ writes to $A$, then $p_1$ writes to $B$, and $p_0$ again writes to $A$, $p_0$ will have 
a cache miss the second time it writes to $A$, 
even though no other process has modified $A$ since $p_0$ last accessed it.
Despite accessing different words, $p_0$ and $p_1$ were competing over the same cache line,
an effect which is called \textit{false sharing}.
False sharing can have a pronounced detrimental effect on the performance of shared data structures\cite{falsesharing}.
To avoid false sharing, implementations of shared memory algorithms can put extra space between 
distinct variables stored in shared memory, thus ensuring the variables are in different cache lines 
and circumventing this effect.
However, since the cache sizes of processes is limited, placing variables in many different cache lines can 
also increase the number of cache misses that occur, so this strategy must be applied carefully.

\section{Other Data Structures Considered}
In our experiments, we compared our implementation of Ko's Trie against these existing dynamic ordered set implementations from the literature:
\begin{itemize}
    \item Fomitchev and Ruppert's lock-free Skip List
    \item Fatourou and Ruppert's wait-free Augmented Trie
\end{itemize}
We chose these implementations because, like Ko's Trie, they are implemented from \textit{CAS},
they are not based on a universal construction and they offer similar operations.

A sequential skip list is a probabilistic data structure that implements a dynamic ordered 
set over an unbounded universe.
It has expected amortized time complexity $O(\log m)$, where $m$ is the number of elements
in the dynamic set represented by the data structure.
The data structure consists of a number of linked lists that are sorted in ascending order and are stacked upon each other in levels.
Typically the number of levels in the skip list is a fixed number.
The linked list at the bottom level (level 0) contains all elements in the data structure, 
where every element is contained in a node.
The linked list at level $j$, where $j > 0$, contains a subset of the elements stored in the linked list 
at level $j-1$.
Moreover, every node in the linked list in level $j > 0$ stores a pointer to the node containing the same element 
in the linked list in level $j - 1$.

When inserting an element into a data structure, it is first inserted into the linked list at the bottom-most level.
When a node is inserted at a particular level which is not the top level, a node containing the same element 
is inserted into the level above it with probability $\frac{1}{2}$.
This ensures that, in expectation, the linked list at level $j$ in the skip list contains $\frac{m}{2^j}$ nodes.
Searches for an element $e$ through the skip list start from the top level linked list.
If a node containing element $e$ is found on a level $j >0$, the search is successful,
otherwise the search is resumed in the linked list for level $j - 1$, starting from the downward pointer 
of the node containing the latest element is smaller than $e$.
If the search continues to the bottom level linked list and no node containing element $e$ is found,
the search fails.
Since there are fewer nodes at higher levels of the skip list, searches can advance very quickly through the skip list before moving down to the lower levels.
To remove an element from the data structure, first the node containing the element is removed from the linked list at the bottom-most level.
This node is called the \textit{root node}. Then, every node containing the same element located above the \textit{root node} is removed from top to bottom.

Fomitchev and Ruppert\cite{10.1145/1011767.1011776,fomitchevThesis} gave an implementation of a lock-free SkipList from \textit{CAS}.
In their implementation, the linked list on each level of the SkipList uses their lock-free linked list implementation.
Their implementation mirrors the sequential implementation of a skip list, except that an element is in the SkipList 
if and only if an unmarked node containing the element is in the bottom-most linked list. 
In addition, nodes at higher levels of the SkipList 
store a pointer to the \textit{root node} containing the same element in the bottom-most level of the SkipList.
If, while traversing the SkipList, a process encounters a node, $mNode$, whose \textit{root node} is marked, 
the process will remove $mNode$ from the SkipList.
The amortized step complexity of any operation, $op$, on the SkipList 
is $O(L(op) + \dot{c}(op) + H)$
where $L(op)$ is the number of operation instances that precede $op$, $\dot{c}(op)$ is the point contention of $op$ and $H$ is the number of levels in the SkipList.
However, Fomitchev and Ruppert believe that situations in which $op$ takes 
$\omega(\log n(op) + \dot{c}(op))$ steps should be rare, where 
$n(op)$ is the number of elements in the SkipList when $op$ begins.

Fatourou and Ruppert\cite{fatourou_ruppert_augtrie,fatourou2024lockfreeaugmentedtrees} published a wait-free implementation of 
a data structure they call an Augmented Trie.
It implements a dynamic ordered set over the universe $U = \{0, \dots, 2^k-1\}$,
for some $k \ge 0$.
The data structure consists of a binary trie over the same universe, 
which we call the \textit{main trie},
except that the bit of every node in the main trie is replaced with a pointer 
to the root of an immutable version of the subtrie rooted at this node.
We will refer to an immutable version of a subtrie as a \textit{version subtrie}
and we will call nodes in a version subtrie \textit{version nodes}.
Every internal version node stores pointers to 
its left and right children which are also version 
nodes.
Instead of storing the \textit{OR} of the bits 
of the leaves in its subtrie, every version node stores a value 
which is equal to the sum of these bits.
The root of the main trie points to the root of the \textit{version trie} 
over $U$ for the set that is represented by the data structure.
Since the root of the version trie stores the sum of the bits of
the leaves in the version trie, processes can obtain the size of the 
set represented by the Augmented Trie in $O(1)$ steps.

Query operation instances on the Augmented Trie start by reading the root of the main trie to 
obtain a pointer to the root of the 
version trie, before performing a top-down version of their sequential algorithms.
To search for a key, $i$, a process will traverse downwards in the version trie to the leaf for key $i$.
The process will return \textit{True} if the leaf for key $i$ has bit 1, 
otherwise the process will return \textit{False}.
In our implementation, if the process encounters an ancestor of the leaf whose value is 0, 
the bit of the leaf must be 0, so a process may return \textit{False} 
immediately, instead of continuing to traverse down to the leaf.

To perform a predecessor instance for a key, $i$, a process will traverse 
downwards in the version trie towards the leaf for key $i$.
If a version node on this path has a left sibling whose value is at least 1, the process stores a pointer to this left sibling in a local variable, \textit{lNode}.
When the process arrives at a version node with value 0 or the leaf for key $i$,
it returns the key of the rightmost leaf with bit 1 in the subtrie rooted at the version node pointed to by \textit{lNode}, or $-1$ if no pointer 
has been stored in \textit{lNode}.

To insert a key, $i$, into the Augmented Trie, a process, $p$, will first check the bit of the 
version node pointed to by the leaf in the main trie associated with key $i$.
If this version node has bit 0, $p$ will use a \textit{CAS} to attempt to update this leaf in the main trie to 
an allocated version leaf with bit 1.
Next, $p$ will visit the ancestors of this leaf from lowest to highest.
For each ancestor, \textit{aNode}, $p$ will first read the version nodes pointed to by the left and right child of \textit{aNode},
and it will allocate a new internal version node that points to these version nodes 
and stores their sum.
Then $p$ will use a \textit{CAS} to update \textit{aNode} to point to this 
internal version node.
If this \textit{CAS} is unsuccessful, $p$ will try one more time.
It reads the version nodes pointed to by \textit{aNode}'s children again, 
re-initializes the immutable internal node to point to these nodes and store their 
sum, and then $p$ performs another \textit{CAS} on \textit{aNode} to try to update it to point to the internal node.
Once $p$ has visited all of the ancestors of the leaf for key $i$, it has finished.

The steps to remove a key, $i$, are the same as those to insert it except that, in the first step, 
if the leaf for key $i$ in the main trie points to a version node whose bit is 1, 
the process will try to update it to an allocated version node whose bit is 0.

A query operation instance is linearized when it reads the root of the version trie from the root of the 
\textit{main trie}.
Describing when update operation instances are linearized is more complicated. 
Fatourou and Ruppert describe the \textit{arrival point} of update operation instances at nodes in the main trie.
In their algorithm, every update operation instance for key $x$ has an \textit{arrival point} at every node on the path from the leaf for key $x$ to the root in the main trie.
An update operation instance is linearized when its \textit{arrival point} at the root of the main trie occurs.
The \textit{arrival point} of an \textit{insert(x)} operation instance at the leaf for key $x$ is either:
\begin{itemize}
    \item the moment at which \textit{op} performs a successful \textit{CAS} to update the leaf for key $x$ to point to a new version node with bit 1
    \item the moment at which an insert instance concurrent with \textit{op} performs a successful \textit{CAS} on 
    the leaf for key $x$, causing \textit{op}'s \textit{CAS} to fail, or,
    \item the moment at which \textit{op} sees that leaf $x$ is pointing to a version node with bit 1, 
    so \textit{op} does not perform a \textit{CAS} on leaf $x$.
\end{itemize}
The \textit{arrival point} of a \textit{remove(x)} operation instance at the leaf for key $x$ is defined symmetrically.
The \textit{arrival point} of \textit{op} at an internal node is the earliest point following \textit{op}'s
\textit{arrival point} at a child of the internal node in which some process first reads this child 
and then performs a successful \textit{CAS}, updating the internal node to point to a new version node.

We make two improvements in our implementation.
First, in the original implementation, an update operation instance may update an internal node, 
\textit{mNode}, in the main 
trie to point to a new version node whose contents are identical to the version node 
that \textit{mNode} previously pointed to.
In our implementation, if the instance reads \textit{mNode} and sees that it points to a version node whose 
children are the same as those pointed to by the children of \textit{mNode}, 
this read of \textit{mNode} can be considered the instance's \textit{arrival point} at \textit{mNode}.
Second, in our implementation, every version leaf is given a \textit{completed} field
which is initially \textit{False}.
Consider a process that updates a leaf in the main trie to point to a version leaf.
Once this process has had its \textit{arrival point} at the root of the main trie, 
it will set the \textit{completed} field of this version leaf to \textit{True}.
If an \textit{insert(x)} instance sees that the leaf associated with key $x$ in the main trie points to a version leaf with bit 1 whose \textit{completed} field is \textit{True}, 
the instance can return instead of helping to update ancestors of this leaf to point to new version nodes.
Similarly, if a remove instance sees that the leaf associated with key $x$ in the main trie points to a version leaf 
with bit 0 whose \textit{completed} field is \textit{True}, 
the instance returns instead of performing helping.
This helped improve performance by over 40\% in some of our experiments.
Search operation instances are also made more efficient in this way.
A search for key $x$ in the AugmentedTrie will check if the leaf for key $x$ in the main trie 
points to a version leaf whose \textit{completed field} is \textit{True}.
If so, the operation instance which created the version leaf has had its \textit{arrival point} 
at the root.
Therefore, the search can simply return whether the bit of the version leaf is 1 
instead of traversing traversing the entire version trie.
If not, the search instance reads a pointer to the root of the version trie and performs the top-down algorithm as previously described.

The step complexity of insert, remove, search and predecessor is at most $O(\log u)$, making this 
algorithm wait-free.
One aspect of this algorithm that makes it more expensive, however is that an update operation instance allocates up to $\log u + 1$ immutable node in the worst case.
However, because query operations are able to get an immutable version of the 
Augmented Trie by reading the pointer of the root node of the main trie,
the Augmented Trie is able to support a wide variety of order-statistic 
queries which are not supported by Ko's Trie or the SkipList.
Thus, while the implementation is more expensive, it 
is more versatile and it is wait-free.

As in our implementation of Ko's lock-free Trie, we reclaim memory from these
other dynamic ordered set implementations using 5-bag DEBRA.
This table gives a summary of the step complexity of the data structure implementations 
we consider for our experiments.
\begin{figure}[H]
\centering
\begin{tabular}{|l|l|}
\hline 
Name & Step Complexity \\
\hline
Ko Trie & \makecell{Insert: Amortized $O(\log u + c^2)$ \\
                    Delete, Predecessor: Amortized $O(\log u + c^2 +\widetilde{c})$)\\Search: O(1)}  \\
\hline
FR SkipList & All Ops: Amortized $O(L + c + H)$\\
\hline
FR Augmented Trie & All Ops: O($\log u$) \\
\hline
\end{tabular}
\caption{Step complexity of implementations we consider. \textit{u} = number of elements in universe, \textit{c} = point contention, $\widetilde{c}$ = overlapping interval contention, \textit{L} = number of previous operations, \textit{H} = number of levels in SkipList.
 }

\end{figure}

\section{Implementation of Experiments}
In this section, we describe the design of the procedures we used to evaluate these dynamic ordered set implementations.
In our experiments, we would like to assess the performance of these data structures 
under a variety of different workloads and process counts.
Consider a fixed universe, $U$, and let $D$ be an implementation of a dynamic ordered set 
over $U$.
We would like to assess how many operation instances can be performed on $D$ by $N$ processes 
over a duration of $t$ seconds, when every process chooses operation instances according 
to some distribution and selects keys for these instances uniformly from $U$.

We consider operation distributions in which every process performs a ratio of $I > 0$ insertions to $R > 0$ removals to $S \ge 0$ searches to $P \ge 0$ predecessor operation instances, such that $I + R + S + P > 0$.
In particular, whenever a process is about to perform an operation instance, 
it performs an insert instance with probability $\frac{I}{I + R + S + P }$,
a remove instance with probability $\frac{R}{I + R + S + P }$,
a search instance with probability $\frac{S}{I + R + S + P }$ 
and a predecessor instance with probability $\frac{P}{I + R + S + P }$.
Let $C$ be the set represented by $D$.
When an insert operation instance is performed, it adds a key to $C$ with probability $1 - \frac{|C|}{|U|}$
whereas a remove instance on $D$ removes a key from $C$ with probability $\frac{|C|}{|U|}$.
Therefore when an operation instance is performed by a process, it inserts a key into $C$ 
with probability $\frac{I}{I + R + S + P} \cdot (1 - \frac{|C|}{|U|}$)
and removes a key from $C$ with probability $\frac{R}{I + R + S + P} \cdot \frac{|C|}{|U|}$.
When $|C| = \frac{I}{I+R}|U|$, an operation instance will add a key 
or remove a key with equal probability:

When $C$ contains less than $\frac{I}{I+R}|U|$ keys, 
processes will, in expectation, insert more keys into $C$ than they remove from $C$.
When $C$ contains more than $\frac{I}{I+R}|U|$ keys, 
processes will, in expectation, remove more keys from $C$ than they 
insert into $C$.
Intuitively, one would expect that if the experiment runs long enough, 
the number of keys in $C$ would eventually become close to $\frac{I}{I+R}|U|$.
When $C$ deviates from this size, the probability that the size of $C$ 
will become closer to $\frac{I}{I+R}|U|$ is higher than the probability that
it will become further from $\frac{I}{I+R}|U|$ as the result of an operation.
We say that $D$ is in a \textit{steady state} when $C$ contains 
approximately $\frac{I}{I+R}|U|$ keys.
In practice, experiments that run for a sufficient time in which $D$ does not eventually reach a \textit{steady state} or in which $D$ leaves a \textit{steady state} are very rare.

The size of $C$ can have a major impact on the performance of dynamic ordered set implementations.
When $C$ is very close to full, there is a high probability that an insert operation instance tries to insert a key that is already in $C$, since keys are drawn uniformly.
Such an insert operation instance performed on our implementation of Ko's trie or the Augmented Trie could finish in $O(1)$ steps, 
whereas an instance that successfully inserts a key into $C$ 
could require $\Omega(\log |U|)$ steps.
Similarly, when $C$ is very close to empty, there is a high probability that a 
remove operation instance will try to remove a key that is not in $C$.
A remove operation instance that tries to remove a key that is not in $C$
can finish in $O(1)$ steps in our implementation of Ko's trie or 
the Augmented Trie, whereas ones that successfully remove a key 
could require $\Omega(\log |U|)$ steps.
Successful remove operation instances on our implementation of Ko's 
trie could be particularly expensive as they perform up to two 
embedded predecessor instances.

If $D$ is initially not in a \textit{steady state}, 
its size could change significantly during the experiment since 
we expect the size of $C$ to approach $\frac{I}{I+R}|U|$ keys during the experiment.
This could make our results very inconsistent.
Before performing any of our experiments on $D$,
we insert keys that are randomly and uniformly chosen with replacement into $D$
until it contains exactly $\frac{I}{I+R}|U|$ keys.
This ensures that $D$ is in a \textit{steady state} at the very beginning of the 
experiment.

When the program for our experiments is launched, a single process, $p_0$, is created.
First, $p_0$ allocates and initializes memory that will be used in 
the dynamic set data structure, $D$.
Next $p_0$ creates a local set which it fills with uniformly chosen keys until it contains exactly 
$\frac{I}{I+R}|U|$ keys, before inserting all elements of that local set into $D$.
Since $p_0$ allocated and initialized the memory in $D$, 
much of the memory in $D$ could still be in $p_0$'s cache.
This could give $p_0$ an unfair advantage compared to other processes
if it participated in the experiment.
Therefore we do not have $p_0$ participate in the experiment.
Processes assigned to hardware threads that share an $L2$ cache 
or an $L3$ cache with the hardware thread $p_0$ was assigned to could 
still have smaller unfair advantages.
However, only a small number of processes would have these advantages 
and once other processes have the memory in $D$ in their caches, 
the advantages would cease to exist.
Before creating the other processes, $p_0$ sets two words in memory, $readyCount$ and $done$, to 0 and $false$ respectively.
Finally, $p_0$ creates $N$ processes, $p_1, \dots, p_N$, before going to sleep for $t + 0.2$ seconds.
This should ensure that $p_0$ will not be active while the other processes are performing operation instances on $D$.

When $p_i$ becomes active, for $1 \le i \le N$, it performs a \textit{fetch-and-increment} on \textit{readyCount}.
Next $p_i$ busy waits, reading $readyCount$ until $N$ is read and thus until every 
process is ready. 
When this happens, $p_i$ will read the current time $c$ and will write $c + t$ into a 
local variable called $endTime$.
While the current time is less than $endTime$ and $done$ is false, 
$p_i$ will repeat the following 50 times:
\begin{itemize}
    \item It draws a random key $k$ uniformly from $U$.
    \item It performs $D.insert(k)$ with probability $\frac{I}{I + R + S + P}$, 
    $D.remove(k)$ with probability $\frac{R}{I + R + S + P}$, $D.search(k)$ with probability $\frac{S}{I + R + S + P}$ 
    or $D.predecessor(k)$ with probability $\frac{P}{I + R + S + P}$.
\end{itemize}
Processes read $done$ and the current time after every 50 operation instances to 
ensure the overhead of the experiment remains small.
When $endTime$ is reached or $p_i$ reads $true$ from \textit{done},
$p_i$ stops performing operation instances.
If $p_i$ did not read \textit{true} from \textit{done}, 
it writes \textit{true} to done.
Process $p_i$ will write the number of operations on $D$ it performed into memory.
Then $p_i$ will perform a \textit{fetch-and-increment} on \textit{finishedCount}, which 
is initially 0 and stores the number of processes that have finished the experiment.
Next, $p_i$ busy waits until it reads $N$ from \textit{finished count} and thus every process 
has finished performing operation instances on $D$.
Finally, $p_i$ will reclaim every record in its limbo bags before terminating.


When $p_0$ wakes up from sleeping for $t + 0.2$ seconds, it 
waits for the other $N$ processes to terminate.
Once they terminate, $p_0$ calculates the total of the operations performed by these processes.
We call this total the \textit{throughput}, which is used 
as a measure of the performance of $D$.
The parameters of the experiment and the \textit{throughput} are output by $p_i$ 
into a file called \textit{resultData.csv}.
Finally, $p_0$ will reclaim all remaining shared memory belonging to the data structure 
and the records inside it before terminating.


The code that we wrote for the experiments is available at \url{https://github.com/Jakjm/LockFreeDynamicSetExperiments}.
The program for the experiments was compiled using the GNU C++ compiler (G++) Version 13.1.0 with the strongest optimization option ($-O3$), compiling using the C++ 2017 Standard (\textit{std=c++17}), using the POSIX Threads library and the C++ atomic standard library.
We used Microsoft's \textit{mimalloc} memory allocation algorithm which was also compiled with G++ 13.1.0.
Valgrind, Google's AddressSanitizer (ASan) and GDB were used to help detect and correct bugs such as memory leaks and segmentation faults in our implementations.

When implementing Ko's Trie and the other data structures considered, 
some shared variables that encountered a high amount of contention were organized into different cache lines
to reduce the chances of false sharing occurring, while also keeping in mind the limited sizes of these caches.
The fields of an UpdateNode are divided into four different cache lines:
one cache line related to the \textit{LatestLists},
one cache line for the fields related to the \textit{UALL},
one for the fields related to the \textit{RUALL}
and one cache line for any other remaining fields.
Similarly, every PredecessorNode has one cache line for fields related to the \textit{PALL},
one cache line for the head of its \textit{notifyList}
and another cache line for its \textit{ruallPosition} field.
Organizing the fields in this way reduced the proportion of L3 cache misses compared to L3 cache accesses with 4 or more processes.

\section{Experiments and Results}
We compared the performance of Ko's Trie, the Augmented Trie and the SkipList 
using the procedure described
in the previous section.
In each experiment, a fixed number of processes performed operation 
instances on one of these dynamic ordered set data structures for 5 seconds.
We chose a duration of 5 seconds because it was long enough that
a significant number of operations can be performed and because testing
revealed that the performance remained consistent between 5 seconds and tests of longer durations.

In most of our experiments, the data structure being tested
stored a dynamic subset of elements from a universe, $U$, of $2^{20}$ keys.
We chose this universe size because it is 
small enough to ensure our system would not run out of memory during an experiment.
This size also ensures that each data structure would require more memory than the size of the largest cache accessible to a process but less than the amount of memory on the entire system.
For example, when our implementation of Ko's Trie contains exactly half of $U$, it requires 
at least $2^{20}$ TrieNodes, $2^{20}$ LatestLists, $2^{20}$ DelNodes and 
$2^{19}$ InsertNodes to represent this set.
On the system we used for our experiments, a DelNode occupies 256 bytes, an InsertNode occupies 256 bytes, a TrieNode occupies 8 bytes and a LatestList occupies 8 bytes, so at least 419MB is required for our implementation to represent a dynamic set storing half the keys in $U$.
This is large enough to exceed the size of even the L3 cache on our experimental system 
while being small enough to avoid risking exhausting the usable memory.
The number of nodes in the SkipList or version nodes in the Augmented Trie 
needed to represent a set of $\frac{|U|}{2}$ keys ensures that, 
while needing more memory than the amount of L3 cache space,
there is little risk of those implementations of exhausting the amount of usable memory either.

We considered different distributions of operations for our experiments.
Prior to each experiment in which a ratio of $I$ insert instances to $R$ remove instances were 
performed, exactly $\frac{I}{I + R}|U|$ elements were inserted into the data structure 
that the experiment was to be performed upon.
As discussed in the previous section, this would leave the data structure in a \textit{steady state},
from which we would expect its size to not significantly deviate during the experiment.

For each graph displaying the results of a series of experiments, 
we have another graph to the right which focuses on the data points for lower numbers of processes.
This helps highlight the performance of Ko's Trie at these smaller process counts.
We ran experiments five times for each data structure, number of processes, thread-pinning policy, 
universe size and operation distribution considered.
Our graphs plot the average of these five runs.

We first performed experiments in which processes were equally likely to perform 
each kind of operation and the universe contained $2^{20}$ keys.
We considered every number of processes, $N$, in the set 
$\{1,2,3,4,5,6\} \bigcup \{2^3,2^4,2^5,2^6,2^7,2^8\}$
and in these experiments, the $N$ processes were 
pinned to the first $N$ hardware threads of the system.
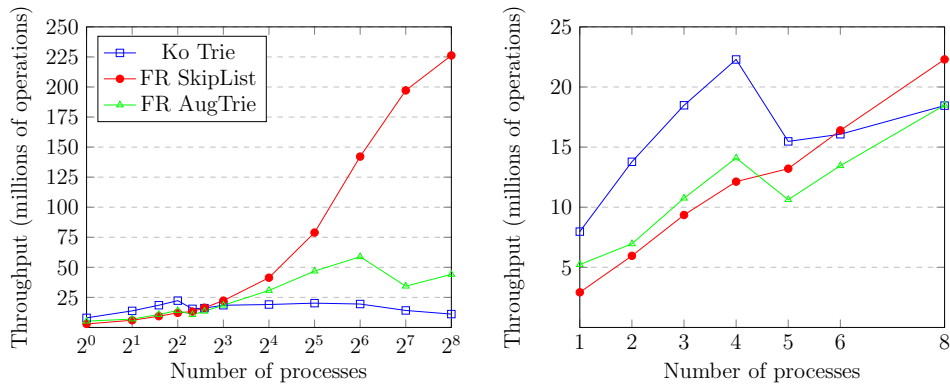
\begin{figure}[H]    
    \centering
    \begin{subfigure}[b]{1\columnwidth}
    \centering
    \scalebox{0.7}{
    \begin{subfigure}[b]{0.55\columnwidth}
    \centering
        \begin{tikzpicture}
    	\begin{axis}[
    		xlabel={Number of processes},
    		ylabel={Throughput (millions of operations)},
    		xmin=1, xmax=256,
    		xmode=log,
    		log basis x={2},
    		ymin=0, ymax=250,
    		xtick={1,2,4,8,16,32,64,128,256},
    		ytick={25,50,75,100,125,150,175,200,225,250},
    		legend pos=north west,
    		ymajorgrids=true,
    		grid style=dashed,
    	]
    
    	\addplot[
    		color=blue,
    		mark=square,
    		]
    		coordinates {
    		(1,7.971000)(2,13.777460)(3,18.477520)(4,22.288050)(5,15.474200)(6,16.069610)
            (8,18.450310)(16,19.066200)(32,20.203090)(64,19.466700)(128,14.132090)(256,11.104840)
    		};
    	\addplot[
    		color=red,
    		mark=*,
    		]
    		coordinates {
    		(1,2.919680)(2,5.956140)(3,9.353590)(4,12.127190)(5,13.208050)(6,16.387370)(8,22.293970)(16,41.379270)(32,78.919450)(64,142.112180)(128,197.158070)(256,226.240700)
    		};
      	\addplot[
    		color=green,
    		mark=triangle,
    		]
    		coordinates {
    		(1,5.23196)(2,6.95431)
      (3,10.75098)(4,14.09173)(5,10.62854)(6,13.45107)(8,18.48466)(16,30.77288)(32,46.82604)(64,58.79768)(128,34.30858)(256,44.11594)
    		};
        \legend{Ko Trie,FR SkipList,FR AugTrie}
    	\end{axis}
        \end{tikzpicture}
    \end{subfigure}
    }
    \scalebox{0.7}{
    \begin{subfigure}[b]{0.55\columnwidth}
    \centering
        \begin{tikzpicture}
    	\begin{axis}[
    		xlabel={Number of processes},
    		ylabel={Throughput (millions of operations)},
    		xmin=1, xmax=8,
    		log basis x={2},
    		ymin=0, ymax=25,
    		xtick={1,2,3,4,5,6,8},
    		ytick={5,10,15,20,25},
    		legend pos=north west,
    		ymajorgrids=true,
    		grid style=dashed,
    	]
    
    	\addplot[
    		color=blue,
    		mark=square,
    		]
    		coordinates {
    		(1,7.971000)(2,13.777460)(3,18.477520)(4,22.288050)(5,15.474200)(6,16.069610)
            (8,18.450310)(16,19.066200)(32,20.203090)(64,19.466700)(128,14.132090)(256,11.104840)
    		};
    	\addplot[
    		color=red,
    		mark=*,
    		]
    		coordinates {
    		(1,2.919680)(2,5.956140)(3,9.353590)(4,12.127190)(5,13.208050)(6,16.387370)(8,22.293970)(16,41.379270)(32,78.919450)(64,142.112180)(128,197.158070)(256,226.240700)
    		};
      	\addplot[
    		color=green,
    		mark=triangle,
    		]
    		coordinates {
    		(1,5.23196)(2,6.95431)
      (3,10.75098)(4,14.09173)(5,10.62854)(6,13.45107)(8,18.48466)(16,30.77288)(32,46.82604)(64,58.79768)(128,34.30858)(256,44.11594)
    		};
    	\end{axis}
        \end{tikzpicture}
    \end{subfigure}
    }
     \hfill
    \caption*{Experiments for an equal distribution of operations, where the $N$ processes are 
    pinned to the first $N$ hardware threads.}
    \end{subfigure}
\end{figure}
We ran these experiments again for each data structure,
except that now we pinned processes to 
\textit{even} numbered hardware threads,
ensuring that no two processes would be assigned to hardware threads on the same core.
There are only 128 such hardware threads on the system 
so we did not perform experiments with more than $128 = 2^7$ processes in this case.
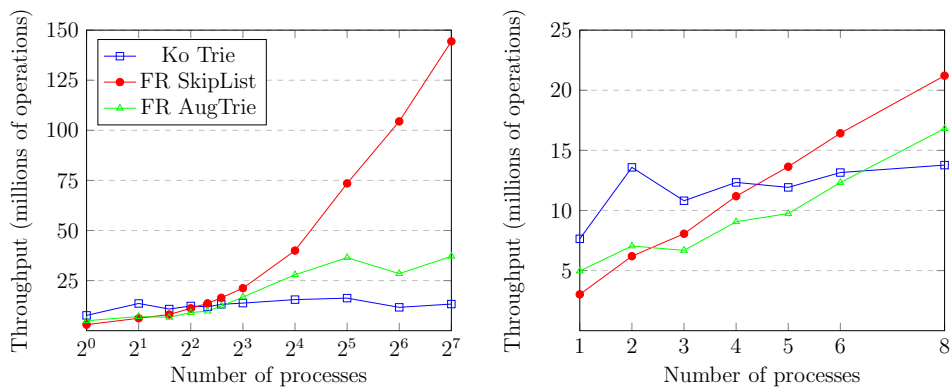
\begin{figure}[H]
    \begin{subfigure}[b]{1\columnwidth}
    \centering
    \scalebox{0.7}{
    \begin{subfigure}[b]{0.55\columnwidth}
    \centering
        \begin{tikzpicture}
    	\begin{axis}[
    		xlabel={Number of processes},
    		ylabel={Throughput (millions of operations)},
    		xmin=1, xmax=128,
    		xmode=log,
    		log basis x={2},
    		ymin=0, ymax=150,
    		xtick={1,2,4,8,16,32,64,128,256},
    		ytick={25,50,75,100,125,150,175,200,225,250},
    		legend pos=north west,
    		ymajorgrids=true,
    		grid style=dashed,
    	]
    
    	\addplot[
    		color=blue,
    		mark=square,
    		]
    		coordinates {
    		(1,7.642370)(2,13.583410)(3,10.810100)(4,12.336160)(5,11.919790)(6,13.155060)
            (8,13.773000)(16,15.508410)(32,16.251650)(64,11.658040)
            (128,13.2877)
    		};
    	\addplot[
    		color=red,
    		mark=*,
    		]
    		coordinates {
    		(1,3.021780)(2,6.195340)(3,8.064570)(4,11.185750)(5,13.640480)(6,16.422640)(8,21.212270)(16,39.981270)(32,73.516140)(64,104.449360)
            (128,144.39724)
    		};
      	\addplot[
    		color=green,
    		mark=triangle,
    		]
    		coordinates {
    		(1,4.96356)(2,7.04672)(3,6.6786)(4,9.06278)(5,9.74435)(6,12.32015)(8,16.80737)(16,27.90161)(32,36.46022)(64,28.48244)
            (128,37.10866)
    		};
        \legend{Ko Trie,FR SkipList,FR AugTrie}
    	\end{axis}
        \end{tikzpicture}
    \end{subfigure}
    }
    \scalebox{0.7}{
    \begin{subfigure}[b]{0.55\columnwidth}
    \centering
        \begin{tikzpicture}
    	\begin{axis}[
    		xlabel={Number of processes},
    		ylabel={Throughput (millions of operations)},
    		xmin=1, xmax=8,
    		log basis x={2},
    		ymin=0, ymax=25,
    		xtick={1,2,3,4,5,6,8,16,32,64,128,256},
    		ytick={5,10,15,20,25},
    		legend pos=north west,
    		ymajorgrids=true,
    		grid style=dashed,
    	]
    
    	\addplot[
    		color=blue,
    		mark=square,
    		]
    		coordinates {
    		(1,7.642370)(2,13.583410)(3,10.810100)(4,12.336160)(5,11.919790)(6,13.155060)
            (8,13.773000)(16,15.508410)(32,16.251650)(64,11.658040)
            (128,13.2877)
    		};
    	\addplot[
    		color=red,
    		mark=*,
    		]
    		coordinates {
    		(1,3.021780)(2,6.195340)(3,8.064570)(4,11.185750)(5,13.640480)(6,16.422640)(8,21.212270)(16,39.981270)(32,73.516140)(64,104.449360)
            (128,144.39724)
    		};
      	\addplot[
    		color=green,
    		mark=triangle,
    		]
    		coordinates {
    		(1,4.96356)(2,7.04672)(3,6.6786)(4,9.06278)(5,9.74435)(6,12.32015)(8,16.80737)(16,27.90161)(32,36.46022)(64,28.48244)
            (128,37.10866)
    		};
    	\end{axis}
        \end{tikzpicture}
    \end{subfigure}
    }
    \hfill
    \caption*{The same experiments for an equal distribution of operations, 
    except the $N$ processes are pinned to the first $N$ \textit{even} hardware threads.}
    \end{subfigure}
\end{figure}
This policy for pinning threads reduced the performance of all data structures, 
particularly at smaller process counts.
This is likely because the processes share L2 caches and L3 caches with only half of the number of processes 
that they share caches with when using the normal thread-pinning policy.
For the rest of this section, we performed experiments
in which the $N$ processes were pinned to the first $N$ hardware threads on the system.

At small numbers of processes, our implementation of Ko's Trie 
has much better performance than the other two 
implementations considered.
However, at higher numbers of processes, its performance does not continue to scale and 
it is outperformed by the SkipList.
We suspect the performance of Ko's Trie is reduced at higher numbers of processes because 
predecessor and update instances have to allocate more memory and perform more steps 
when many such instances overlap with each other.
Every update operation instance could have to allocate a NotifyNode for every predecessor 
instance it overlaps with, inserting the NotifyNode into the \textit{notifyList} of the PredecessorNode created by this predecessor instance.
Inserting a NotifyNode becomes expensive if other update instances are also inserting NotifyNodes into this \textit{notifyList}.
A remove operation instance can be particularly expensive, since it could perform up to two entire embedded predecessor instances.
If so, it usually has to insert a NotifyNode into each of the \textit{notifyNodes} PredecesorNodes created by 
these instances.
These two PredecessorNodes will remain in the \textit{PALL} throughout the remove instance, 
which could force overlapping update instances to allocate and insert a NotifyNode into the \textit{notifyLists} of both of them.
Even if a remove instance encounters no contention, it may have to allocate a DelNode, two PredecessorNodes and two NotifyNodes.
Predecessor instances on the Trie have to allocate PredecessorNodes,
whereas predecessor instances performed on the other two implementations do not need to allocate memory.
Predecessor instances may need store a significant amount of information about the 
predecessor and update operation instances they overlap with.
For example, they store pointers to UpdateNodes they encounter in the \textit{RUALL} in local sets
to help them determine the correct predecessor for their key.

The Augmented Trie is outperformed by Ko's Trie at smaller numbers of processes 
and by the SkipList at higher numbers of processes.
We suspect one cause for the poor performance of the Augmented Trie
is that update operation instances performed on it
must allocate $\Omega(\log |U|)$ version nodes in the worst case, even at lower numbers of processes,
whereas, at lower numbers of processes, the memory allocated by instances performed on Ko's Trie will be more limited. 

At higher numbers of processes, Fomitchev and Ruppert's SkipList 
vastly outperformed the Augmented Trie and our implementation of Ko's Trie.
This appears to confirm Fomitchev and Ruppert's belief that the upper bound on the amortized step complexity on instances on the SkipList is rather loose.
One likely reason for its superior performance is that the algorithms for the SkipList do 
not allocate as much memory as the two other algorithms.
Search, predecessor and remove operation instances on the SkipList 
do not need to allocate any records.
Insert operation instances that insert an element into the SkipList will, in expectation, allocate only 2 nodes, and it is very improbable that such an instance will allocate more than a constant number of nodes.
One other reason why the SkipList performs well is that it is unlikely for instances on the SkipList to suffer 
performance impacts due to contention.
Since the processes performing operation instances draw their keys uniformly,
when a sufficient number of elements are in the SkipList, 
there is a high probability that operation instances can complete their steps without directly encountering other operation instances. 
In contrast, in our implementation of Ko's Trie, update operation instances may contend with each other while inserting UpdateNodes into the \textit{UALL} or \textit{RUALL}, updating the relaxed binary trie or inserting NotifyNodes into the \textit{notifyLists} of PredecessorNodes created by overlapping Predecessor instances.

To help us better understand the performance of these data structures when updates are more common, 
we performed experiments in which updates were performed more frequently than queries.
First, we ran experiments in which inserts and removes are each performed with probability $\frac{4}{10}$ 
and searches and predecessor operations are each performed with probability $\frac{1}{10}$.
As before, the universe contained $2^{20}$ keys.
\begin{figure}[H]    
    \centering
    \begin{subfigure}[b]{1\columnwidth}
    \centering
    \scalebox{0.7}{
    \begin{subfigure}[b]{0.55\columnwidth}
    \centering
        \begin{tikzpicture}
    	\begin{axis}[
    		xlabel={Number of processes},
    		ylabel={Throughput (millions of operations)},
    		xmin=1, xmax=256,
    		xmode=log,
    		log basis x={2},
    		ymin=0, ymax=200,
    		xtick={1,2,4,8,16,32,64,128,256},
    		ytick={25,50,75,100,125,150,175,200,225,250},
    		legend pos=north west,
    		ymajorgrids=true,
    		grid style=dashed,
    	]
    
        \addplot[
    		color=blue,
    		mark=square,
    		]
    		coordinates {
    		(1,6.635680)(2,11.284220)(3,14.878610)(4,17.773120)(5,13.188290)(6,11.810070)
            (8,13.860490)(16,13.768340)(32,14.020800)(64,13.519810)(128,9.61098)(256,6.427040)
    		};
    	\addplot[
    		color=red,
    		mark=*,
    		]
    		coordinates {
    		(1,2.824170)(2,5.353310)(3,8.878150)(4,11.500990)(5,13.047860)(6,15.806980)(8,20.895500)(16,38.945570)(32,73.180540)(64,132.246250)(128,175.763900)(256,171.957040)
    		};
      	\addplot[
    		color=green,
    		mark=triangle,
    		]
    		coordinates {
    		(1,2.8503)(2,4.19358)(3,6.40828)(4,9.56977)(5,8.55895)(6,9.71285)(8,13.98798)(16,22.68256)(32,33.85894)(64,39.46849)(128,22.83229)(256,26.51203)
    		};
        \legend{Ko Trie,FR SkipList,FR AugTrie}
    	\end{axis}
        \end{tikzpicture}
    \end{subfigure}
    }
    \scalebox{0.7}{
    \begin{subfigure}[b]{0.55\columnwidth}
    \centering
        \begin{tikzpicture}
    	\begin{axis}[
    		xlabel={Number of processes},
    		ylabel={Throughput (millions of operations)},
    		xmin=1, xmax=8,
    		log basis x={2},
    		ymin=0, ymax=25,
    		xtick={1,2,3,4,5,6,8,16,32,64,128,256},
    		ytick={5,10,15,20,25},
    		legend pos=north west,
    		ymajorgrids=true,
    		grid style=dashed,
    	]
    
        \addplot[
    		color=blue,
    		mark=square,
    		]
    		coordinates {
    		(1,6.635680)(2,11.284220)(3,14.878610)(4,17.773120)(5,13.188290)(6,11.810070)
            (8,13.860490)(16,13.768340)(32,14.020800)(64,13.519810)(128,9.61098)(256,6.427040)
    		};
    	\addplot[
    		color=red,
    		mark=*,
    		]
    		coordinates {
    		(1,2.824170)(2,5.353310)(3,8.878150)(4,11.500990)(5,13.047860)(6,15.806980)(8,20.895500)(16,38.945570)(32,73.180540)(64,132.246250)(128,175.763900)(256,171.957040)
    		};
      	\addplot[
    		color=green,
    		mark=triangle,
    		]
    		coordinates {
    		(1,2.8503)(2,4.19358)(3,6.40828)(4,9.56977)(5,8.55895)(6,9.71285)(8,13.98798)(16,22.68256)(32,33.85894)(64,39.46849)(128,22.83229)(256,26.51203)
    		};
    	\end{axis}
        \end{tikzpicture}
    \end{subfigure}
    }
    \hfill
    \caption*{Experiments featuring an update-heavy distribution of operations.}
\end{subfigure}
\end{figure}
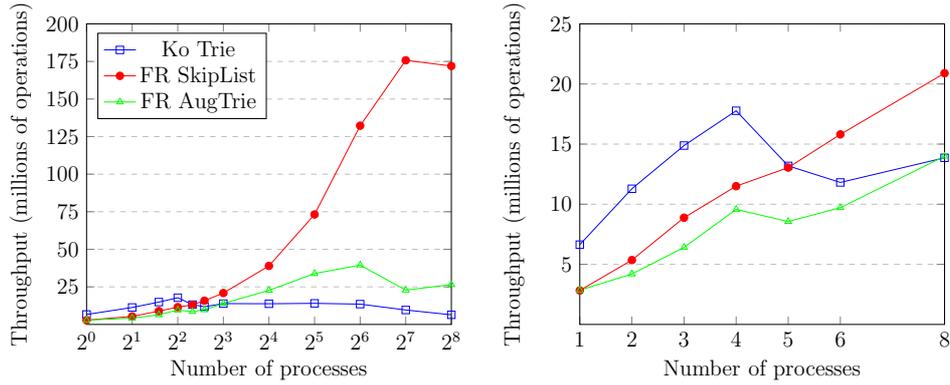
To focus on the performance of these data structures when insertions are more common than removals, 
we performed experiments in which insertions are performed with probability $\frac{4}{9}$,
removals are performed with probability $\frac{1}{9}$, and both searches and 
predecessor instances are performed with probability $\frac{2}{9}$.
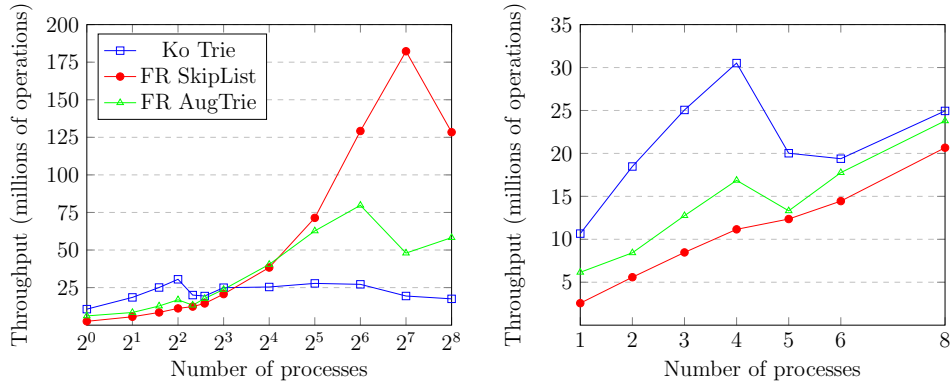
\begin{figure}[H]
    \centering
    \scalebox{0.7}{
    \begin{subfigure}[b]{0.55\columnwidth}
    \centering
        \begin{tikzpicture}
    	\begin{axis}[
    		xlabel={Number of processes},
    		ylabel={Throughput (millions of operations)},
    		xmin=1, xmax=256,
    		xmode=log,
    		log basis x={2},
    		ymin=0, ymax=200,
    		xtick={1,2,4,8,16,32,64,128,256},
    		ytick={25,50,75,100,125,150,175,200,225,250},
    		legend pos=north west,
    		ymajorgrids=true,
    		grid style=dashed,
    	]
    
        \addplot[
    		color=blue,
    		mark=square,
    		]
    		coordinates {
    		(1,10.66161)(2,18.47596)(3,25.062)(4,30.51852)(5,20.01874)(6,19.38482)
            (8,24.94848)(16,25.40199)(32,27.79433)(64,27.16029)(128,19.40175)(256,17.51305)
    		};
    	\addplot[
    		color=red,
    		mark=*,
    		]
    		coordinates {
    		(1,2.56029)(2,5.58354)(3,8.46965)(4,11.15824)(5,12.35982)(6,14.44723)(8,20.66381)(16,38.34709)(32,71.40252)(64,129.2008)(128,182.26422)(256,128.40427)
    		};
      	\addplot[
    		color=green,
    		mark=triangle,
    		]
    		coordinates {
    		(1,6.15761)(2,8.42564)(3,12.75153)(4,16.86428)(5,13.31558)(6,17.76136)(8,23.79129)(16,40.42481)(32,62.72107)(64,79.75481)(128,47.95819)(256,58.26223)
    		};
        \legend{Ko Trie,FR SkipList,FR AugTrie}
    	\end{axis}
        \end{tikzpicture}
    \end{subfigure}
    }
        \scalebox{0.7}{
    \begin{subfigure}[b]{0.55\columnwidth}
    \centering
        \begin{tikzpicture}
    	\begin{axis}[
    		xlabel={Number of processes},
    		ylabel={Throughput (millions of operations)},
    		xmin=1, xmax=8,
    		ymin=0, ymax=35,
    		xtick={1,2,3,4,5,6,8,16,32,64,128,256},
    		ytick={5,10,15,20,25,30,35,40,45,50},
    		ymajorgrids=true,
    		grid style=dashed,
    	]
    
        \addplot[
    		color=blue,
    		mark=square,
    		]
    		coordinates {
    		(1,10.66161)(2,18.47596)(3,25.062)(4,30.51852)(5,20.01874)(6,19.38482)
            (8,24.94848)(16,25.40199)(32,27.79433)(64,27.16029)(128,19.40175)(256,17.51305)
    		};
    	\addplot[
    		color=red,
    		mark=*,
    		]
    		coordinates {
    		(1,2.56029)(2,5.58354)(3,8.46965)(4,11.15824)(5,12.35982)(6,14.44723)(8,20.66381)(16,38.34709)(32,71.40252)(64,129.2008)(128,182.26422)(256,128.40427)
    		};
      	\addplot[
    		color=green,
    		mark=triangle,
    		]
    		coordinates {
    		(1,6.15761)(2,8.42564)(3,12.75153)(4,16.86428)(5,13.31558)(6,17.76136)(8,23.79129)(16,40.42481)(32,62.72107)(64,79.75481)(128,47.95819)(256,58.26223)
    		};
    	\end{axis}
        \end{tikzpicture}
    \end{subfigure}
    }
    \caption*{Experiments featuring an insert-heavy distribution of operations.}
\end{figure}
Now focusing on the performance when removals are performed more frequently than insertions, we 
we performed experiments in which removals are performed with probability $\frac{4}{9}$,
insertions are performed with probability $\frac{1}{9}$, and both searches and 
predecessor instances are performed with probability $\frac{2}{9}$.
\begin{figure}[H]
    \centering
    \scalebox{0.7}{
    \begin{subfigure}[b]{0.55\columnwidth}
    \centering
        \begin{tikzpicture}
    	\begin{axis}[
    		xlabel={Number of processes},
    		ylabel={Throughput (millions of operations)},
    		xmin=1, xmax=256,
    		xmode=log,
    		log basis x={2},
    		ymin=0, ymax=250,
    		xtick={1,2,4,8,16,32,64,128,256},
    		ytick={25,50,75,100,125,150,175,200,225,250},
    		legend pos=north west,
    		ymajorgrids=true,
    		grid style=dashed,
    	]
    
        \addplot[
    		color=blue,
    		mark=square,
    		]
    		coordinates {
    		(1,6.12328)(2,12.85706)(3,19.02618)(4,24.37058)(5,17.44808)(6,18.26204)
            (8,23.1545)(16,26.17512)(32,26.93262)(64,27.14396)(128,19.10533)(256,15.79172)
    		};
    	\addplot[
    		color=red,
    		mark=*,
    		]
    		coordinates {
    		(1,5.31117)(2,10.3041)(3,15.51798)(4,20.75956)(5,22.10382)(6,25.90239)(8,35.81399)(16,63.71257)(32,118.30213)(64,193.80485)(128,246.57704)(256,116.04338)
    		};
      	\addplot[
    		color=green,
    		mark=triangle,
    		]
    		coordinates {
    		(1,5.96637)(2,8.92158)(3,13.30123)(4,17.79034)(5,13.41222)(6,18.06965)(8,23.62774)(16,39.78196)(32,62.12476)(64,80.96441)(128,46.97393)(256,64.2598)
    		};
        \legend{Ko Trie,FR SkipList,FR AugTrie}
    	\end{axis}
        \end{tikzpicture}
    \end{subfigure}
    }
        \scalebox{0.7}{
    \begin{subfigure}[b]{0.55\columnwidth}
    \centering
        \begin{tikzpicture}
    	\begin{axis}[
    		xlabel={Number of processes},
    		ylabel={Throughput (millions of operations)},
    		xmin=1, xmax=8,
    		ymin=0, ymax=40,
    		xtick={1,2,3,4,5,6,8,16,32,64,128,256},
    		ytick={5,10,15,20,25,30,35,40,45,50},
    		ymajorgrids=true,
    		grid style=dashed,
    	]
    
        \addplot[
    		color=blue,
    		mark=square,
    		]
    		coordinates {
    		(1,6.12328)(2,12.85706)(3,19.02618)(4,24.37058)(5,17.44808)(6,18.26204)
            (8,23.1545)(16,26.17512)(32,26.93262)(64,27.14396)(128,19.10533)(256,15.79172)
    		};
    	\addplot[
    		color=red,
    		mark=*,
    		]
    		coordinates {
    		(1,5.31117)(2,10.3041)(3,15.51798)(4,20.75956)(5,22.10382)(6,25.90239)(8,35.81399)(16,63.71257)(32,118.30213)(64,193.80485)(128,246.57704)(256,116.04338)
    		};
      	\addplot[
    		color=green,
    		mark=triangle,
    		]
    		coordinates {
    		(1,5.96637)(2,8.92158)(3,13.30123)(4,17.79034)(5,13.41222)(6,18.06965)(8,23.62774)(16,39.78196)(32,62.12476)(64,80.96441)(128,46.97393)(256,64.2598)
    		};
    	\end{axis}
        \end{tikzpicture}
    \end{subfigure}
    }
    \caption*{Experiments featuring a remove-heavy distribution of operations.}
\end{figure}
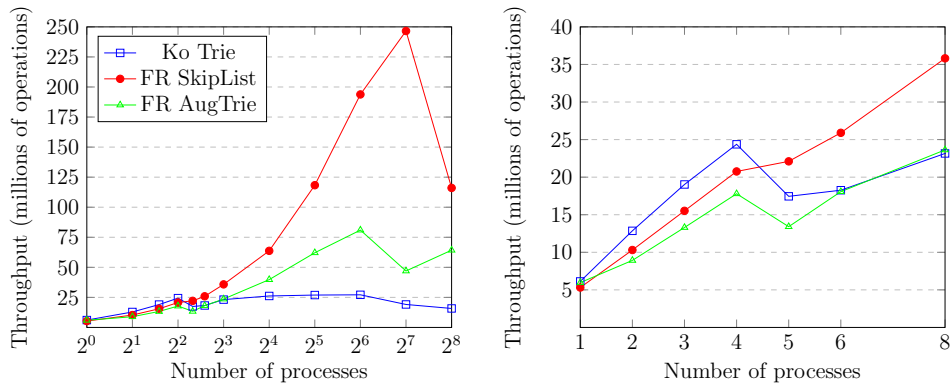

We can see a distinct drop in the performance of Ko's Trie in the update-heavy experiments.
Its performance improves in the experiments in which insertions are performed more frequently than removes compared 
to an equal distribution of all operations.
As expected, however, its performance suffers particularly in the experiments in which removes are performed more often than insertions.
This helps confirm our intuition that remove operation instances are much more expensive than insert operation instances due to their embedded predecessor instances.
The SkipList performed better in the experiments with update-heavy and remove-heavy workloads,
which is probably because processes in those workloads allocate fewer nodes comapared to insert-heavy workloads.
Since the implementation of insertions is similar to that of removals for the AugmentedTrie,
it is not very surprising that its performance is similar across these different workloads.

We suspected that the performance of Ko's Trie and the Augmented Trie would 
improve if query operation instances were more common compared to update operation instances.
We believed this effect would be less pronounced for the SkipList: while operation instances
may perform fewer allocations of nodes and fewer modifications to the SkipList,
most will still have to traverse to the lowest level of the SkipList to complete their operations.
The next few graphs display the results of experiments in which queries were more common compared to 
update operation instances.
First, we performed experiments in which searches and predecessor operations were each performed with probability $\frac{4}{10}$,
whereas inserts and removes were performed with probability $\frac{1}{10}$.
\begin{figure}[H]    
    \centering
    \begin{subfigure}[b]{1\columnwidth}
    \centering
    \scalebox{0.7}{
    \begin{subfigure}[b]{0.55\columnwidth}
    \centering
        \begin{tikzpicture}
    	\begin{axis}[
    		xlabel={Number of processes},
    		ylabel={Throughput (millions of operations)},
    		xmin=1, xmax=256,
    		xmode=log,
    		log basis x={2},
    		ymin=0, ymax=300,
    		xtick={1,2,4,8,16,32,64,128,256},
    		ytick={25,50,75,100,125,150,175,200,225,250,275,300},
    		legend pos=north west,
    		ymajorgrids=true,
    		grid style=dashed,
    	]
    
        \addplot[
    		color=blue,
    		mark=square,
    		]
    		coordinates {
    		(1,10.062150)(2,18.109250)(3,24.859760)(4,30.382910)(5,23.493560)(6,23.096230)
            (8,27.244720)(16,30.024170)(32,33.200240)(64,36.545320)(128,24.945850)(256,25.316780)
    		};
    	\addplot[
    		color=red,
    		mark=*,
    		]
    		coordinates {
    		(1,3.139270)(2,6.756910)(3,10.138440)(4,13.256630)(5,14.704900)(6,18.659640)(8,24.787520)(16,46.405460)(32,86.739740)(64,155.768620)(128,227.439300)(256,275.518620)
    		};
      	\addplot[
    		color=green,
    		mark=triangle,
    		]
    		coordinates {
    		(1,6.15668)(2,9.3268)(3,14.33952)(4,19.24577)(5,16.55103)(6,21.50708)(8,28.79547)(16,49.80019)(32,80.25253)(64,111.40033)(128,76.52083)(256,90.34206)
    		};
        \legend{Ko Trie,FR SkipList,FR AugTrie}
    	\end{axis}
        \end{tikzpicture}
    \end{subfigure}
    }
    \scalebox{0.7}{
    \begin{subfigure}[b]{0.55\columnwidth}
    \centering
        \begin{tikzpicture}
    	\begin{axis}[
    		xlabel={Number of processes},
    		ylabel={Throughput (millions of operations)},
    		xmin=1, xmax=8,
    		log basis x={2},
    		ymin=0, ymax=35,
    		xtick={1,2,3,4,5,6,8,16,32,64,128,256},
    		ytick={5,10,15,20,25,30,35,40,45,50},
    		legend pos=north west,
    		ymajorgrids=true,
    		grid style=dashed,
    	]
    
        \addplot[
    		color=blue,
    		mark=square,
    		]
    		coordinates {
    		(1,10.062150)(2,18.109250)(3,24.859760)(4,30.382910)(5,23.493560)(6,23.096230)
            (8,27.244720)(16,30.024170)(32,33.200240)(64,36.545320)(128,24.945850)(256,25.316780)
    		};
    	\addplot[
    		color=red,
    		mark=*,
    		]
    		coordinates {
    		(1,3.139270)(2,6.756910)(3,10.138440)(4,13.256630)(5,14.704900)(6,18.659640)(8,24.787520)(16,46.405460)(32,86.739740)(64,155.768620)(128,227.439300)(256,275.518620)
    		};
      	\addplot[
    		color=green,
    		mark=triangle,
    		]
    		coordinates {
    		(1,6.15668)(2,9.3268)(3,14.33952)(4,19.24577)(5,16.55103)(6,21.50708)(8,28.79547)(16,49.80019)(32,80.25253)(64,111.40033)(128,76.52083)(256,90.34206)
    		};
    	\end{axis}
        \end{tikzpicture}
    \end{subfigure}
    }
    \hfill
        \caption*{Experiments for a query-heavy distribution of operations.}
    \end{subfigure}
\end{figure}
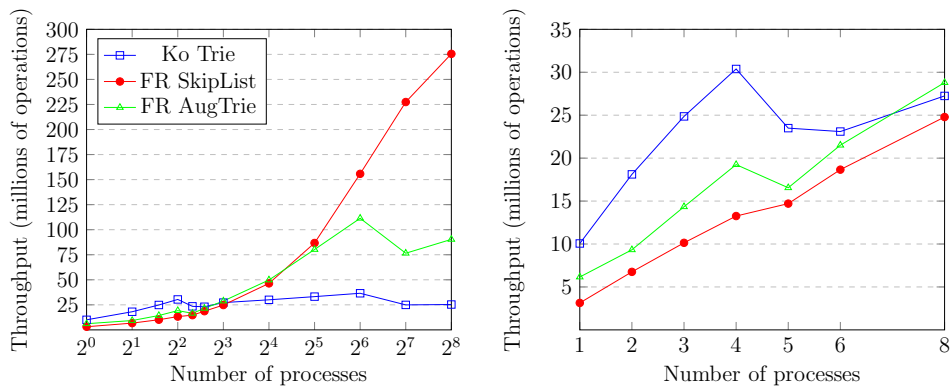
To better understand the performance when searches are more frequent, 
we next tested the data structures under a workload in which searches are the only queries that are performed and updates are performed infrequently.
Processes performed search operations with probability $\frac{8}{10}$ and both insertions and removals are performed with probability $\frac{1}{10}$. 
\begin{figure}[H]
    \centering
    \scalebox{0.7}{
    \begin{subfigure}[b]{0.55\columnwidth}
    \centering
        \begin{tikzpicture}
    	\begin{axis}[
    		xlabel={Number of processes},
    		ylabel={Throughput (millions of operations)},
    		xmin=1, xmax=256,
    		xmode=log,
    		log basis x={2},
    		ymin=0, ymax=225,
    		xtick={1,2,4,8,16,32,64,128,256},
    		ytick={25,50,75,100,125,150,175,200,225,250},
    		legend pos=north west,
    		ymajorgrids=true,
    		grid style=dashed,
    	]
    
        \addplot[
    		color=blue,
    		mark=square,
    		]
    		coordinates {
    		(1,13.01453)(2,25.0047)(3,36.96513)(4,48.1062)(5,36.61737)(6,41.47397)
            (8,51.95273)(16,60.6194)(32,68.17421)(64,62.46083)(128,43.30859)(256,35.09471)
    		};
    	\addplot[
    		color=red,
    		mark=*,
    		]
    		coordinates {
    		(1,3.52384)(2,6.82679)(3,10.08382)(4,13.36501)(5,15.39168)(6,18.69538)(8,24.90587)(16,45.87242)(32,85.83444)(64,153.89077)(128,221.10043)(256,104.71792)
    		};
      	\addplot[
    		color=green,
    		mark=triangle,
    		]
    		coordinates {
    		(1,5.38462)(2,8.7041)(3,15.06431)(4,21.50434)(5,21.52212)(6,28.25102)(8,42.15638)(16,74.37222)(32,120.11073)(64,153.17399)(128,90.53373)(256,118.14371)
    		};
        \legend{Ko Trie,FR SkipList,FR AugTrie}
    	\end{axis}
        \end{tikzpicture}
    \end{subfigure}
    }
        \scalebox{0.7}{
    \begin{subfigure}[b]{0.55\columnwidth}
    \centering
        \begin{tikzpicture}
    	\begin{axis}[
    		xlabel={Number of processes},
    		ylabel={Throughput (millions of operations)},
    		xmin=1, xmax=8,
    		ymin=0, ymax=55,
    		xtick={1,2,3,4,5,6,8,16,32,64,128,256},
    		ytick={5,10,15,20,25,30,35,40,45,50},
    		ymajorgrids=true,
    		grid style=dashed,
    	]
    
        \addplot[
    		color=blue,
    		mark=square,
    		]
    		coordinates {
    		(1,13.01453)(2,25.0047)(3,36.96513)(4,48.1062)(5,36.61737)(6,41.47397)
            (8,51.95273)(16,60.6194)(32,68.17421)(64,62.46083)(128,43.30859)(256,35.09471)
    		};
    	\addplot[
    		color=red,
    		mark=*,
    		]
    		coordinates {
    		(1,3.52384)(2,6.82679)(3,10.08382)(4,13.36501)(5,15.39168)(6,18.69538)(8,24.90587)(16,45.87242)(32,85.83444)(64,153.89077)(128,221.10043)(256,104.71792)
    		};
      	\addplot[
    		color=green,
    		mark=triangle,
    		]
    		coordinates {
    		(1,5.38462)(2,8.7041)(3,15.06431)(4,21.50434)(5,21.52212)(6,28.25102)(8,42.15638)(16,74.37222)(32,120.11073)(64,153.17399)(128,90.53373)(256,118.14371)
    		};
    	\end{axis}
        \end{tikzpicture}
    \end{subfigure}
    }
    \caption*{Experiments featuring an search-heavy distribution of operations.}
\end{figure}
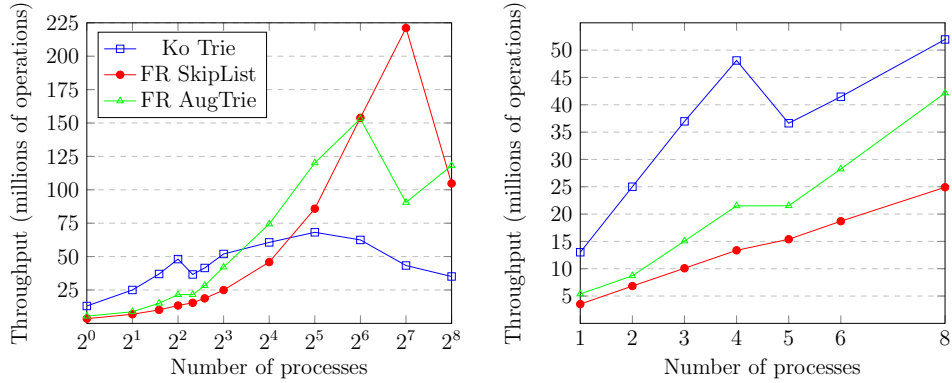
We also considered a distribution in which predecessor instances were performed with probability $\frac{8}{10}$,
searches were not performed, and both insertions and removals were performed with probability $\frac{1}{10}$.
\begin{figure}[H]
    \centering
    \scalebox{0.7}{
    \begin{subfigure}[b]{0.55\columnwidth}
    \centering
        \begin{tikzpicture}
    	\begin{axis}[
    		xlabel={Number of processes},
    		ylabel={Throughput (millions of operations)},
    		xmin=1, xmax=256,
    		xmode=log,
    		log basis x={2},
    		ymin=0, ymax=225,
    		xtick={1,2,4,8,16,32,64,128,256},
    		ytick={25,50,75,100,125,150,175,200,225,250},
    		legend pos=north west,
    		ymajorgrids=true,
    		grid style=dashed,
    	]
    
        \addplot[
    		color=blue,
    		mark=square,
    		]
    		coordinates {
    		(1,8.04417)(2,14.46376)(3,19.06758)(4,22.54988)(5,13.51762)(6,14.71572)
            (8,18.8716)(16,19.64934)(32,22.33547)(64,25.2116)(128,17.65135)(256,18.37143)
    		};
    	\addplot[
    		color=red,
    		mark=*,
    		]
    		coordinates {
    		(1,2.90639)(2,6.13641)(3,10.01733)(4,13.28881)(5,14.19672)(6,17.99629)(8,24.80377)(16,46.53304)(32,86.23938)(64,154.90445)(128,221.23762)(256,134.14849)
    		};
      	\addplot[
    		color=green,
    		mark=triangle,
    		]
    		coordinates {
    		(1,5.08839)(2,7.26068)(3,11.02185)(4,14.80887)(5,12.61614)(6,16.54895)(8,22.86479)(16,39.54315)(32,64.19982)(64,88.63299)(128,67.8434)(256,76.80396)
    		};
        \legend{Ko Trie,FR SkipList,FR AugTrie}
    	\end{axis}
        \end{tikzpicture}
    \end{subfigure}
    }
        \scalebox{0.7}{
    \begin{subfigure}[b]{0.55\columnwidth}
    \centering
        \begin{tikzpicture}
    	\begin{axis}[
    		xlabel={Number of processes},
    		ylabel={Throughput (millions of operations)},
    		xmin=1, xmax=8,
    		ymin=0, ymax=30,
    		xtick={1,2,3,4,5,6,8,16,32,64,128,256},
    		ytick={5,10,15,20,25,30,35,40,45,50},
    		ymajorgrids=true,
    		grid style=dashed,
    	]
    
        \addplot[
    		color=blue,
    		mark=square,
    		]
    		coordinates {
    		(1,8.04417)(2,14.46376)(3,19.06758)(4,22.54988)(5,13.51762)(6,14.71572)
            (8,18.8716)(16,19.64934)(32,22.33547)(64,25.2116)(128,17.65135)(256,18.37143)
    		};
    	\addplot[
    		color=red,
    		mark=*,
    		]
    		coordinates {
    		(1,2.90639)(2,6.13641)(3,10.01733)(4,13.28881)(5,14.19672)(6,17.99629)(8,24.80377)(16,46.53304)(32,86.23938)(64,154.90445)(128,221.23762)(256,134.14849)
    		};
      	\addplot[
    		color=green,
    		mark=triangle,
    		]
    		coordinates {
    		(1,5.08839)(2,7.26068)(3,11.02185)(4,14.80887)(5,12.61614)(6,16.54895)(8,22.86479)(16,39.54315)(32,64.19982)(64,88.63299)(128,67.8434)(256,76.80396)
    		};
    	\end{axis}
        \end{tikzpicture}
    \end{subfigure}
    }
    \caption*{Experiments featuring an predecessor-heavy distribution of operations.}
\end{figure}
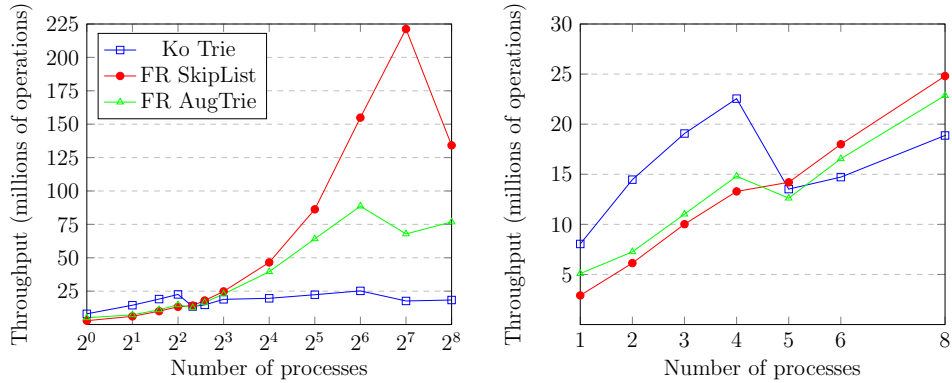

The performance of our implementation of Ko's Trie is poor when 
predecessors are frequently performed and updates are infrequent.
This is understandable, since predecessor instances may have some update operation instances 
that overlap with them, making the predecessor instances and the update instances perform more work.
However, as expected, Ko's Trie performs well in the query-heavy and search-heavy workloads. 
The performance of the Augmented Trie was much better in these sets of experiments, 
especially in experiments with a search-heavy distribution of operations,
in which it outperforms the SkipList 
at $2^4$, $2^5$ and $2^8$ processes 
and has equal performance to that of the SkipList at $2^6$ processes.
The implementations of search and predecessor on the SkipList are quite similar 
so it is unsurprising that the performance of the SkipList remains similar 
among these workloads in which queries were more common.
As in the other results we have seen, the SkipList outperforms the other implementations considered 
at higher numbers of processes and continues to scale extremely well.
 
We wanted to explore the impact of increasing the universe size on the performance of these implementations.
We finally tested the performance of each data structure when the size of the universe was increased to $2^{22}$ keys.
In these experiments, processes were equally likely to perform each kind of operation and the $N$ processes were 
pinned to the first $N$ hardware threads on the system.
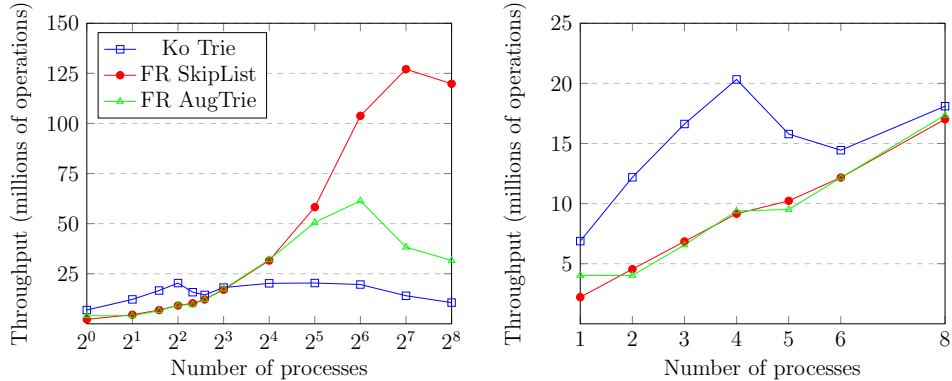
\begin{figure}[H]
    \centering
    \scalebox{0.7}{
    \begin{subfigure}[b]{0.55\columnwidth}
    \centering
        \begin{tikzpicture}
    	\begin{axis}[
    		xlabel={Number of processes},
    		ylabel={Throughput (millions of operations)},
    		xmin=1, xmax=256,
    		xmode=log,
    		log basis x={2},
    		ymin=0, ymax=150,
    		xtick={1,2,4,8,16,32,64,128,256},
    		ytick={25,50,75,100,125,150,175,200,225,250},
    		legend pos=north west,
    		ymajorgrids=true,
    		grid style=dashed,
    	]
    
        \addplot[
    		color=blue,
    		mark=square,
    		]
    		coordinates {
    		(1,6.879000)(2,12.185860)(3,16.619770)(4,20.327880)(5,15.784300)(6,14.439500)
            (8,18.100330)(16,20.213930)(32,20.366080)(64,19.608680)(128,14.01375)(256,10.59146)
    		};
    	\addplot[
    		color=red,
    		mark=*,
    		]
    		coordinates {
    		(1,2.218160)(2,4.545260)(3,6.850160)(4,9.156460)(5,10.229890)(6,12.169480)(8,17.014180)(16,31.583700)(32,58.226240)(64,103.810630)(128,127.07871)(256,119.805630)
    		};
      	\addplot[
    		color=green,
    		mark=triangle,
    		]
    		coordinates {
    		(1,4.02804)(2,4.02485)(3,6.55981)(4,9.38498)(5,9.51169)(6,12.14874)(8,17.35578)(16,32.01296)(32,50.57669)(64,61.40207)(128,38.25354)(256,31.58648)
    		};
        \legend{Ko Trie,FR SkipList,FR AugTrie}
    	\end{axis}
        \end{tikzpicture}
    \end{subfigure}
    }
        \scalebox{0.7}{
    \begin{subfigure}[b]{0.55\columnwidth}
    \centering
        \begin{tikzpicture}
    	\begin{axis}[
    		xlabel={Number of processes},
    		ylabel={Throughput (millions of operations)},
    		xmin=1, xmax=8,
    		ymin=0, ymax=25,
    		xtick={1,2,3,4,5,6,8,16,32,64,128,256},
    		ytick={5,10,15,20,25,30,35,40,45,50},
    		ymajorgrids=true,
    		grid style=dashed,
    	]
    
        \addplot[
    		color=blue,
    		mark=square,
    		]
    		coordinates {
    		(1,6.879000)(2,12.185860)(3,16.619770)(4,20.327880)(5,15.784300)(6,14.439500)
            (8,18.100330)(16,20.213930)(32,20.366080)(64,19.608680)(128,14.01375)(256,10.59146)
    		};
    	\addplot[
    		color=red,
    		mark=*,
    		]
    		coordinates {
    		(1,2.218160)(2,4.545260)(3,6.850160)(4,9.156460)(5,10.229890)(6,12.169480)(8,17.014180)(16,31.583700)(32,58.226240)(64,103.810630)(128,127.07871)(256,119.805630)
    		};
      	\addplot[
    		color=green,
    		mark=triangle,
    		]
    		coordinates {
    		(1,4.02804)(2,4.02485)(3,6.55981)(4,9.38498)(5,9.51169)(6,12.14874)(8,17.35578)(16,32.01296)(32,50.57669)(64,61.40207)(128,38.25354)(256,31.58648)
    		};
    	\end{axis}
        \end{tikzpicture}
    \end{subfigure}
    }
    \caption*{Experiments in which an equal distribution of operations was performed, where the universe consisted of $2^{22}$ keys.}
\end{figure}

Surprisingly, the performance of the SkipList drops significantly when the size of the universe is increased. 
With 32 processes and an equal distribution of all operations, the average throughput of the SkipList over 5 runs 
was 78.9 million operations when the universe contained $2^{20}$ keys, 
whereas the average throughput was 58.2 million operations when the universe contained $2^{22}$ keys, a drop of 26\%.
By comparison, the performance of our implementation of Ko's Trie and the Augmented Trie actually increased when the universe size is increased.
This increase is likely due to the reduced contention when updating the relaxed binary trie or the Augmented Trie, respectively.
We suspect that the SkipList performs more poorly because, with the increased size of the universe and, thus, the increased size of the SkipList, 
the nodes that processes frequently visit at higher levels of the SkipList do not fit into the caches of these processes as well.

In conclusion, our implementation of Ko's Trie is the best choice when the number of processes performing operations is expected to be low,
the workload consists of few remove operation instances and many insert and search instances.
Its performance does not suffer greatly when the size of the universe is increased.
The SkipList is the best choice when the number of processes performing operations is expected to be high, 
it scales extremely well compared to the other implementations we considered.
Its performance is reduced slightly as the size of the universe increases.
The Augmented Trie performs very well when query operation instances are more common compared to update instances, 
especially at higher numbers of processes.
It is also wait-free and supports a wide variety of operations like range-queries and size-queries that 
are not supported by the SkipList or Ko's Trie.
Therefore, while its performance is not always as good as that of Ko's Trie or the SkipList, its improved progress guarantees 
and functionality can make it a good option.

\printbibliography[heading=bibintoc]
\end{document}